\documentclass[reqno]{amsart}

\usepackage{fullpage}

\usepackage{ bbold }
\usepackage{booktabs}

\usepackage[]{tikz}
\usetikzlibrary{arrows}
\usetikzlibrary{calc}
\usepackage[T2A]{fontenc}
\usepackage{amssymb, amsmath, amsthm}
\usepackage[utf8]{inputenc}
\usepackage[english,russian]{}
\usepackage{xcolor}
\usepackage{graphicx}
\usepackage{caption}
\usepackage{subcaption}

\theoremstyle{plain}
\newtheorem{Th}{Theorem}[section]
\newtheorem{lemma}[Th]{Lemma}
\newtheorem{remark}[Th]{Remark}
\newtheorem{prop}[Th]{Proposition}
\newtheorem{cor}[Th]{Corollary}
\newtheorem{definition}[Th]{Definition}

 \usepackage{color,soul}
\definecolor{lightblue}{rgb}{.90,.95,1}
\sethlcolor{lightblue}

\renewcommand{\theequation}{\arabic{section}.\arabic{equation}}

\newcommand{\bb}[1]{\ensuremath{{{\color{gray}\blacklozenge}}_{\raisebox{-1pt}{\tiny $\mathrm{#1}$}}}}
\newcommand{\ww}[1]{\lozenge_{\raisebox{-1pt}{\tiny $\mathrm{#1}$}}}
\newcommand{\V}[1]{\mathcal{V}_{\raisebox{-1pt}{\tiny#1}}}

\newcommand{\om}{\Omega}

\newcommand{\dV}{\partial\V{}}
\newcommand{\vl}{v_{\raisebox{-1pt}{\tiny$\lambda$}}}
\newcommand{\vlbr}{v_{\raisebox{-1pt}{\tiny $\overline\lambda$}}} 
\newcommand{\uR}{u_{\raisebox{-1pt}{\tiny R}}}
\newcommand{\uI}{u_{\raisebox{-1pt}{\tiny I}}}

\newcommand{\intV}{\mathrm {Int}\V{}}

\newcommand{\dintb}[1]{\partial_{\mathrm {int}}\bb{#1}}
\newcommand{\dintw}[1]{\partial_{\mathrm {int}}\ww{#1}}
\newcommand{\diom}{\partial_{\mathrm {int}}\Omega}
\newcommand{\dib}{\partial_{\mathrm {int}}\bb{}}
\newcommand{\diw}{\partial_{\mathrm {int}}\ww{}}

\newcommand{\clom}{\overline{\om}}
\newcommand{\clb}{\bar{\bb{}}}
\newcommand{\clw}{\bar{\ww{}}}
\newcommand{\clbb}[1]{\bar{\ensuremath{{\color{gray}\blacklozenge}}}_{\raisebox{-1pt}{\tiny $\mathrm{#1}$}}}
\newcommand{\clww}[1]{\bar{\ensuremath{{\lozenge}}}_{\raisebox{-1pt}{\tiny $\mathrm{#1}$}}}

\newcommand{\dom}{\partial\Omega}
\newcommand{\db}{\partial\ensuremath{{\color{gray}\blacklozenge}}}
\newcommand{\dw}{\partial\lozenge}
\newcommand{\dbb}[1]{\partial\bb{#1}}
\newcommand{\dww}[1]{\partial\ww{#1}}

\newcommand{\bbb}{u_0}
\newcommand{\www}{v_0}

\newcommand{\Aa}{u_{-}}
\newcommand{\cc}{u_{+}}
\newcommand{\pp}{v_{\sharp}}
\newcommand{\qq}{v_{\flat}}

\newcommand{\F}{F}
\newcommand{\G}{G}
\newcommand{\HH}{H}
\newcommand{\ff}[1]{f_{{#1}}}

\newcommand{\h}{h}
\newcommand{\Cmd}{C_{\Omega^\delta}}
\newcommand{\Cm}{C_{\Omega}}
\newcommand{\omp}{\Omega^\prime}
\newcommand{\Cmm}{C_{\Omega^\prime}}

\newcommand{\re}{\mathrm {Re}}
\newcommand{\im}{\mathrm {Im}}

\newcommand{\mdel}{M^\delta}
\newcommand{\Ftilda}{\widetilde{F}}

\newcommand{\bvpugl}[1]{\tilde{u}^{\delta}_{#1}}
\newcommand{\bvupugl}[1]{u^{\ast\delta}_{#1}}
\newcommand{\wvpugl}[1]{\tilde{v}^{\delta}_{#1}}
\newcommand{\wvupugl}[1]{v^{\ast\delta}_{#1}}

\newcommand{\nbvpugl}[1]{\tilde{u}_{#1}}
\newcommand{\nbvupugl}[1]{u^{\ast}_{#1}}
\newcommand{\nwvpugl}[1]{\tilde{v}_{#1}}
\newcommand{\nwvupugl}[1]{v^{\ast}_{#1}}

\begin{document}

\title{Dimers in piecewise Temperleyan domains}
\large

\author[Marianna Russkikh]{Marianna Russkikh$^\mathrm{\sharp, \flat}$
}

\thanks{\textsc{
${}^\mathrm{\sharp}$ 
Section de Math\'ematiques, Universit\'e de Gen\`eve.
2-4 rue du Li\`evre, Case postale~64, 1211 Gen\`eve 4, Suisse.
}}

\thanks{\textsc{
${}^\mathrm{\flat}$
Chebyshev Laboratory, Department of Mathematics and Mechanics, St. Petersburg
State University. 14th Line, 29b, 199178 St. Petersburg, Russia}}

\thanks{{\it E-mail addresses:} \texttt{Marianna.Russkikh@unige.ch}}

\begin{abstract} 
We study the large-scale behavior of the height function in the dimer model on the square lattice. Richard Kenyon has shown that the fluctuations of the height function on Temperleyan discretizations of a planar domain converge in the scaling limit (as the mesh size tends to zero) to the Gaussian Free Field with Dirichlet boundary conditions. We extend Kenyon's result to a more general class of discretizations.

Moreover, we introduce a new factorization of the coupling function of the double-dimer model into two discrete holomorphic functions, which are similar to discrete fermions defined in~\cite{Stas, Stas07}.For Temperleyan discretizations with appropriate boundary modifications, the results of Kenyon imply that the expectation of the double-dimer height function converges to a harmonic function in the scaling limit. We use the above factorization to extend this result to the class of all polygonal discretizations, that are not necessarily Temperleyan. Furthermore, we show that, quite surprisingly, the expectation of the double-dimer height function in the Temperleyan case is exactly discrete harmonic (for an appropriate choice of Laplacian) even before taking the scaling limit.
\end{abstract}

\maketitle

\tableofcontents

\section{Introduction}
\noindent{\bf Dimer model.}
The dimer model is one of the best known models of statistical physics, first introduced to model a gas of diatomic molecules~\cite{FT}. By modifying the underlying graph, it can be used to study the Ising model (see Fisher's approach~\cite{Ff}). Under the name ``perfect matchings'', it prominently appears in theoretical computer science and combinatorics.

A dimer covering (or perfect matching) of a graph is a subset of edges that covers every vertex exactly once. The dimer model is a random covering of a given graph by dimers. We will be interested in uniform random coverings, that is, those chosen from the distribution in which all dimer configurations are equally weighted.

In this paper, we work with dimers on finite subgraphs (also called domains) of the square lattice. Such a dimer covering  
may be viewed as a random tiling of a domain of the dual lattice by dominos $2\times 1$.
Thurston introduced the height function of a domino tiling which uniquely assigns integer values to all vertices of the dual lattice. 
Moreover, a domino tiling can be reconstructed from the values of the height function. Thus, one can think of a random domino tiling as a random height function on the vertex set of the domain. 

The key question in the dimer model concerns the large-scale behavior of the expectation of the height function and of its fluctuations. 
We are interested in studying the scaling limit of the dimer model on planar graphs as the mesh tends to zero.
One of the main interesting features lies in the conformal invariance of such scaling limit. The scaling limit is conformal invariant if its image under any conformal mapping has the same distribution as an analogous object in the image of the domain.

For planar graphs, Kasteleyn~\cite{Kast} showed that the partition function of the dimer model can be evaluated as the determinant of a signed adjacency matrix, the {\it Kasteleyn matrix}. The local statistics for the uniform measure on dimer configurations can be computed using the inverse Kasteleyn matrix, see~\cite{Klocstat}. The latter can be viewed as a two-point function, called the {\it coupling function}~\cite{Kdom}. 
The coupling function is a complex-valued discrete holomorphic function. As such, its real and imaginary parts are discrete harmonic, and the study of the local statistics of random tilings can be reduced to the study of the convergence of discrete harmonic functions.

A Temperleyan discretization (see Fig.~\ref{Temp}) is a discrete domain with special boundary conditions. It is defined in Section~\ref{hf_and_Td}. Temperleyan domains correspond to Dirichlet boundary conditions or Neumann boundary conditions for the discrete harmonic components of the coupling function. Kenyon~\cite{Kdom, KGff} used this approach to prove the conformal invariance of the limiting distribution  of the height function in the case of Temperleyan discretizations.

More precisely, if one considers Temperleyan discretizations of a given domain $\om$, Kenyon~\cite{Kdom} showed that the limit of the expected height function is a harmonic function with boundary values depending on the direction (the argument of the tangent vector) of the boundary.
 In~\cite{KGff} Kenyon proved that, in the case of Temperleyan discretizations, the fluctuations of the height function converge (as the mesh size tends to zero) to the Gaussian Free Field~\cite{Scott} on $\om$ with Dirichlet boundary conditions. 
One of the main results of the present paper is an extension of Kenyon's result to a class of {\it Piecewise Temperleyan} discretizations defined in Section~\ref{pwd}. 
Note that for more general discretizations, with domains that are not necessarily Temperleyan, the large-scale behavior of the expectation of the height function and its fluctuations is much more complicated, see~\cite{BG, CKP, K-O, K-O-Sh, P}. The exact nature of fluctuations is not established yet, but they expected to be given by a Gaussian free field in appropriate coordinates, obtained from solving the complex Burgers equation in~\cite{K-O}. In the particular case of a sequence of domains whose boundary height functions are bounded by some constant, the new coordinates coincide with the usual ones, and the fluctuations are expected to be given by the Gaussian free field on the limiting domain with Dirichlet boundary conditions. 


A different approach to showing the convergence of the fluctuations of the height function to the Gaussian Free Field was introduced in~\cite{BLR}. The main tool here is the Uniform Spanning Tree and the winding of its branches, which coincides with the dimer model height function. In particular, this approach covers the case of Temperleyan discretizations, but not  the case of Piecewise Temperleyan discretizations.

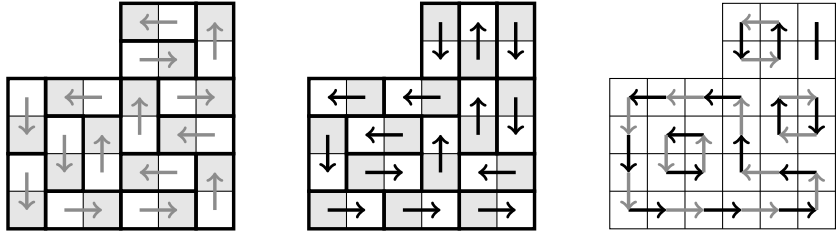
\begin{figure}
\begin{center}
\begin{tikzpicture}[x={(0.5cm,0cm)}, y={(0cm,0.5cm)}]
\begin{scope}
\draw (0,0) -- (0,4) -- (3,4) -- (3,6) -- (6,6) -- (6,0)-- cycle;

\draw  [draw, fill=gray!20](0,0) rectangle (1,1);
\draw  [draw, fill=gray!20](1,1) rectangle (2,2);
\draw  [draw, fill=gray!20](2,2) rectangle (3,3);
\draw  [draw, fill=gray!20](3,3) rectangle (4,4);
\draw  [draw, fill=gray!20](4,4) rectangle (5,5);
\draw  [draw, fill=gray!20](5,5) rectangle (6,6);

\draw  [draw, fill=gray!20](0,2) rectangle (1,3);
\draw  [draw, fill=gray!20](1,3) rectangle (2,4);
\draw  [draw, fill=gray!20](3,5) rectangle (4,6);

\draw  [draw, fill=gray!20](2,0) rectangle (3,1);
\draw  [draw, fill=gray!20](3,1) rectangle (4,2);
\draw  [draw, fill=gray!20](4,2) rectangle (5,3);
\draw  [draw, fill=gray!20](5,3) rectangle (6,4);

\draw  [draw, fill=gray!20](4,0) rectangle (5,1);
\draw  [draw, fill=gray!20](5,1) rectangle (6,2);

\draw[line width=1.3pt,gray!90][<-] (0.5,0.5) -- (0.5,1.5);
\draw[line width=1.3pt,gray!90][<-] (0.5,2.5) -- (0.5,3.5);
\draw[line width=1.3pt,gray!90][<-] (1.5,3.5) -- (2.5,3.5);
\draw[line width=1.3pt,gray!90][<-] (3.5,3.5) -- (3.5,2.5);
\draw[line width=1.3pt,gray!90][<-] (3.5,1.5) -- (4.5,1.5);
\draw[line width=1.3pt,gray!90][<-] (5.5,1.5) -- (5.5,0.5);
\draw[line width=1.3pt,gray!90][->] (3.5,0.5) -- (4.5,0.5);
\draw[line width=1.3pt,gray!90][->] (1.5,0.5) -- (2.5,0.5);

\draw[line width=1.3pt,gray!90][<-] (1.5,1.5)--(1.5,2.5);
\draw[line width=1.3pt,gray!90][->] (2.5,1.5) -- (2.5,2.5);

\draw[line width=1.3pt,gray!90][<-] (4.5,2.5) --(5.5,2.5);
\draw[line width=1.3pt,gray!90][->] (4.5,3.5) --(5.5,3.5);

\draw[line width=1.3pt,gray!90][->] (3.5,4.5) --(4.5,4.5);
\draw[line width=1.3pt,gray!90][<-] (3.5,5.5) --(4.5,5.5);
\draw[line width=1.3pt,gray!90][->] (5.5,4.5) --(5.5,5.5);

\draw[line width=1.3pt] (0,0) rectangle (1,2);
\draw[line width=1.3pt] (0,2) rectangle (1,4);
\draw[line width=1.3pt] (1,0) rectangle (3,1);
\draw[line width=1.3pt] (1,3) rectangle (3,4);
\draw[line width=1.3pt] (3,0) rectangle (5,1);
\draw[line width=1.3pt] (3,1) rectangle (5,2);
\draw[line width=1.3pt] (4,2) rectangle (6,3);
\draw[line width=1.3pt] (4,3) rectangle (6,4);
\draw[line width=1.3pt] (3,4) rectangle (5,5);
\draw[line width=1.3pt] (3,5) rectangle (5,6);
\draw[line width=1.3pt] (1,1) rectangle (2,3);
\draw[line width=1.3pt] (5,0) rectangle (6,2);
\draw[line width=1.3pt] (3,2) rectangle (4,4);
\draw[line width=1.3pt] (5,4) rectangle (6,6);
\end{scope}

\begin{scope}[xshift=4cm]

\draw  [draw, fill=gray!20](0,0) rectangle (1,1);
\draw  [draw, fill=gray!20](1,1) rectangle (2,2);
\draw  [draw, fill=gray!20](2,2) rectangle (3,3);
\draw  [draw, fill=gray!20](3,3) rectangle (4,4);
\draw  [draw, fill=gray!20](4,4) rectangle (5,5);
\draw  [draw, fill=gray!20](5,5) rectangle (6,6);

\draw  [draw, fill=gray!20](0,2) rectangle (1,3);
\draw  [draw, fill=gray!20](1,3) rectangle (2,4);
\draw  [draw, fill=gray!20](3,5) rectangle (4,6);

\draw  [draw, fill=gray!20](2,0) rectangle (3,1);
\draw  [draw, fill=gray!20](3,1) rectangle (4,2);
\draw  [draw, fill=gray!20](4,2) rectangle (5,3);
\draw  [draw, fill=gray!20](5,3) rectangle (6,4);

\draw  [draw, fill=gray!20](4,0) rectangle (5,1);
\draw  [draw, fill=gray!20](5,1) rectangle (6,2);

\draw[line width=1.3pt,black][->] (0.5,2.5) -- (0.5,1.5);
\draw[line width=1.3pt,black][->] (1.5,3.5) -- (0.5,3.5);
\draw[line width=1.3pt,black][->] (3.5,3.5) -- (2.5,3.5);
\draw[line width=1.3pt,black][->] (3.5,1.5) -- (3.5,2.5);
\draw[line width=1.3pt,black][->] (5.5,1.5) -- (4.5,1.5);
\draw[line width=1.3pt,black][->] (4.5,0.5) -- (5.5,0.5);
\draw[line width=1.3pt,black][<-] (3.5,0.5) -- (2.5,0.5);
\draw[line width=1.3pt,black][<-] (1.5,0.5) -- (0.5,0.5);

\draw[line width=1.3pt,black][->] (2.5,2.5)--(1.5,2.5);
\draw[line width=1.3pt,black][->] (1.5,1.5) -- (2.5,1.5);

\draw[line width=1.3pt,black][->] (4.5,2.5) --(4.5,3.5);
\draw[line width=1.3pt,black][->] (5.5,3.5) --(5.5,2.5);

\draw[line width=1.3pt][<-] (5.5,4.5) --(5.5,5.5);
\draw[line width=1.3pt,black][->] (4.5,4.5) --(4.5,5.5);
\draw[line width=1.3pt,black][<-] (3.5,4.5) --(3.5,5.5);

\draw (0,0) -- (0,4) -- (3,4) -- (3,6) -- (6,6) -- (6,0)-- cycle;

\draw[line width=1.3pt] (0,0) rectangle (2,1);
\draw[line width=1.3pt] (2,0) rectangle (4,1);
\draw[line width=1.3pt] (4,0) rectangle (6,1);
\draw[line width=1.3pt] (4,1) rectangle (6,2);
\draw[line width=1.3pt] (4,4) rectangle (5,6);
\draw[line width=1.3pt] (5,4) rectangle (6,6);
\draw[line width=1.3pt] (4,2) rectangle (5,4);
\draw[line width=1.3pt] (5,2) rectangle (6,4);
\draw[line width=1.3pt] (3,4) rectangle (4,6);
\draw[line width=1.3pt] (0,3) rectangle (2,4);
\draw[line width=1.3pt] (2,3) rectangle (4,4);
\draw[line width=1.3pt] (1,1) rectangle (3,2);
\draw[line width=1.3pt] (1,2) rectangle (3,3);
\draw[line width=1.3pt] (0,1) rectangle (1,3);

\end{scope}

\begin{scope}[xshift=8cm]

\draw (0,0) -- (0,4) -- (3,4) -- (3,6) -- (6,6) -- (6,0)-- cycle;
\draw (0,1) -- (6,1);
\draw (0,2) -- (6,2);
\draw (0,3) -- (6,3);
\draw (3,4) -- (6,4);
\draw (3,5) -- (6,5);

\draw (1,0) -- (1,4);
\draw (2,0) -- (2,4);
\draw (3,0) -- (3,4);
\draw (4,0) -- (4,6);
\draw (5,0) -- (5,6);

\draw[line width=1.3pt,black][->] (0.5,2.5) -- (0.5,1.5);
\draw[line width=1.3pt,black][->] (1.5,3.5) -- (0.5,3.5);
\draw[line width=1.3pt,black][->] (3.5,3.5) -- (2.5,3.5);
\draw[line width=1.3pt,black][->] (3.5,1.5) -- (3.5,2.5);
\draw[line width=1.3pt,black][->] (5.5,1.5) -- (4.5,1.5);
\draw[line width=1.3pt,black][->] (4.5,0.5) -- (5.5,0.5);
\draw[line width=1.3pt,black][<-] (3.5,0.5) -- (2.5,0.5);
\draw[line width=1.3pt,black][<-] (1.5,0.5) -- (0.5,0.5);

\draw[line width=1.3pt,black][->] (2.5,2.5)--(1.5,2.5);
\draw[line width=1.3pt,black][->] (1.5,1.5) -- (2.5,1.5);

\draw[line width=1.3pt,black][->] (4.5,2.5) --(4.5,3.5);
\draw[line width=1.3pt,black][->] (5.5,3.5) --(5.5,2.5);

\draw[line width=1.3pt] (5.5,4.5) --(5.5,5.5);
\draw[line width=1.3pt,black][->] (4.5,4.5) --(4.5,5.5);
\draw[line width=1.3pt,black][<-] (3.5,4.5) --(3.5,5.5);

\draw[line width=1.3pt,gray!90][<-] (0.5,0.5) -- (0.5,1.5);
\draw[line width=1.3pt,gray!90][<-] (0.5,2.5) -- (0.5,3.5);
\draw[line width=1.3pt,gray!90][<-] (1.5,3.5) -- (2.5,3.5);
\draw[line width=1.3pt,gray!90][<-] (3.5,3.5) -- (3.5,2.5);
\draw[line width=1.3pt,gray!90][<-] (3.5,1.5) -- (4.5,1.5);
\draw[line width=1.3pt,gray!90][<-] (5.5,1.5) -- (5.5,0.5);
\draw[line width=1.3pt,gray!90][->] (3.5,0.5) -- (4.5,0.5);
\draw[line width=1.3pt,gray!90][->] (1.5,0.5) -- (2.5,0.5);

\draw[line width=1.3pt,gray!90][<-] (1.5,1.5)--(1.5,2.5);
\draw[line width=1.3pt,gray!90][->] (2.5,1.5) -- (2.5,2.5);

\draw[line width=1.3pt,gray!90][<-] (4.5,2.5) --(5.5,2.5);
\draw[line width=1.3pt,gray!90][->] (4.5,3.5) --(5.5,3.5);

\draw[line width=1.3pt,gray!90][->] (3.5,4.5) --(4.5,4.5);
\draw[line width=1.3pt,gray!90][<-] (3.5,5.5) --(4.5,5.5);
\draw[line width=1.3pt] (5.5,4.5) --(5.5,5.5);

\end{scope}
\end{tikzpicture} \caption{Two different domino tilings of the same domain can be combined into a collection of loops and double edges. Orienting the edges of the first covering from white to black, and the edges of the second one from black to white, one gets an orientation of the resulting loops. 
}\label{petli}\end{center}
\end{figure}

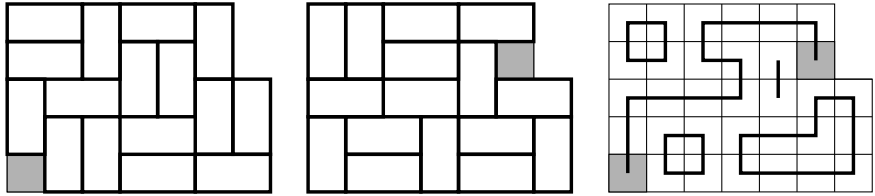
\begin{figure}
\begin{center}
\begin{tikzpicture}[x={(0.5cm,0cm)}, y={(0cm,0.5cm)}]
\begin{scope}
\draw (0,0) -- (0,5) -- (6,5) -- (6,3) -- (7,3) -- (7,0)-- cycle;

\draw  [draw, fill=gray!60](0,0) rectangle (1,1);

\draw[line width=1.3pt] (0,1) rectangle (1,3);
\draw[line width=1.3pt] (0,3) rectangle (2,4);
\draw[line width=1.3pt] (0,4) rectangle (2,5);
\draw[line width=1.3pt] (2,3) rectangle (3,5);
\draw[line width=1.3pt] (3,4) rectangle (5,5);
\draw[line width=1.3pt] (5,3) rectangle (6,5);
\draw[line width=1.3pt] (1,0) rectangle (2,2);
\draw[line width=1.3pt] (2,0) rectangle (3,2);
\draw[line width=1.3pt] (3,2) rectangle (4,4);
\draw[line width=1.3pt] (4,2) rectangle (5,4);
\draw[line width=1.3pt] (5,1) rectangle (6,3);
\draw[line width=1.3pt] (6,1) rectangle (7,3);
\draw[line width=1.3pt] (5,0) rectangle (7,1);
\draw[line width=1.3pt] (3,0) rectangle (5,1);
\end{scope}

\begin{scope}[xshift=4cm]
\draw (0,0) -- (0,5) -- (6,5) -- (6,3) -- (7,3) -- (7,0)-- cycle;

\draw  [draw, fill=gray!60](5,3) rectangle (6,4);

\draw[line width=1.3pt] (0,0) rectangle (1,2);
\draw[line width=1.3pt] (3,0) rectangle (4,2);
\draw[line width=1.3pt] (6,0) rectangle (7,2);
\draw[line width=1.3pt] (0,3) rectangle (1,5);
\draw[line width=1.3pt] (1,3) rectangle (2,5);
\draw[line width=1.3pt] (4,2) rectangle (5,4);
\draw[line width=1.3pt] (1,0) rectangle (3,1);
\draw[line width=1.3pt] (1,1) rectangle (3,2);
\draw[line width=1.3pt] (4,0) rectangle (6,1);
\draw[line width=1.3pt] (4,1) rectangle (6,2);
\draw[line width=1.3pt] (5,3) rectangle (7,2);
\draw[line width=1.3pt] (0,2) rectangle (2,3);
\draw[line width=1.3pt] (2,3) rectangle (4,4);
\draw[line width=1.3pt] (2,4) rectangle (4,5);
\draw[line width=1.3pt] (4,4) rectangle (6,5);
\end{scope}

\begin{scope}[xshift=8cm]
\draw (0,0) -- (0,5) -- (6,5) -- (6,3) -- (7,3) -- (7,0)-- cycle;
\draw (0,1) -- (7,1);
\draw (0,2) -- (7,2);
\draw (0,3) -- (7,3);
\draw (0,4) -- (6,4);

\draw (1,0) -- (1,5);
\draw (2,0) -- (2,5);
\draw (3,0) -- (3,5);
\draw (4,0) -- (4,5);
\draw (5,0) -- (5,5);
\draw (6,0) -- (6,3);

\draw  [draw, fill=gray!60](0,0) rectangle (1,1);
\draw  [draw, fill=gray!60](5,3) rectangle (6,4);
\draw[line width=1.3pt] (0.5,0.5) -- (0.5,2.5) -- (3.5,2.5) -- (3.5,3.5) -- (2.5,3.5) -- (2.5,4.5)--(5.5,4.5)--(5.5,3.5);
\draw[line width=1.3pt] (4.5,3.5) -- (4.5,2.5);
\draw[line width=1.3pt] (0.5,3.5) rectangle (1.5,4.5);
\draw[line width=1.3pt] (1.5,0.5) rectangle (2.5,1.5);
\draw[line width=1.3pt] (3.5,0.5) -- (3.5,1.5) -- (5.5,1.5) -- (5.5,2.5) -- (6.5,2.5) --(6.5,0.5)-- cycle;

\end{scope}
\end{tikzpicture} \caption{Left and center: the coverings of the domains that differ on two squares. Right: the interface between these two squares and the collection of loops and double edges is the result of the composition of the coverings.
}\label{petli2}\end{center}
\end{figure}  

\bigskip

\noindent{\bf Double dimers.}
Let us now come to the second series of results of our paper, which deal with the double-dimer model. Recall that a double-dimer configuration is a union of two dimer coverings, or equivalently a set of even-length simple loops and double edges with the property that every vertex is the endpoint of exactly two edges, see Fig. \ref{petli}.  Note that there are two ways to obtain a given loop (on the dual graph). This can be interpreted as a choice of orientation of the loop, see Fig.~\ref{petli}. Thus, the double-dimer model can be represented as a random covering of the dual graph by oriented loops and double edges~\cite{Percus}. The height function in the double-dimer model, which is the difference of height functions for two dimer configurations, has a simple geometric representation: if we cross a loop, then the height function changes by $+1$ or $-1$, depending on the orientation of the loop.

There is a prediction that the loop ensemble of the double-dimer model converges to the conformal loop ensemble $\mathit{CLE}(4)$, see~\cite{CLE1, CLE2}. In the case of discretizations by Temperleyan domains Kenyon~\cite{Kloop} and Dubedat~\cite{Ddim} obtained results confirming this prediction.
 The loop ensemble $\mathit{CLE}(4)$ is a conformally invariant object. 
It corresponds to level lines of the Gaussian Free Field. There is a gap of $\pm2\lambda=\pm \sqrt{\pi/2}$ between the values of the Gaussian Free Field on the interior and the exterior side of each $\mathit{CLE}(4)$ loop~\cite{ScSh, WW}. This is similar to loops in the double-dimer model outlining the discontinuities of the double-dimer height function, the gap being $\pm1$.

We will consider coverings of a pair of domains that differ by two squares, see~Fig.~\ref{petli2}. In this case, in addition to a collection of loops and double edges, the superposition of the coverings contains an ``interface'' (a simple path between these two squares).
It is expected that the interface converges to a conformally invariant random curve $\mathit{SLE}(4)$ as the mesh size tends to zero, see~\cite{ICM2006}.

The coupling function plays an important role in the proof of convergence of height functions. We define the double-dimer coupling function as a difference of single-dimer coupling functions of a pair of domains that differ by two squares. Similarly to the single-dimer coupling function the double-dimer coupling function can be used to compute the expectation of the double-dimer height function. However the single-dimer coupling function is also the kernel which allows to compute multi-edge correlations, see \cite{Klocstat}. Therefore it allows us to compute all moments of the single-dimer height function, see \cite{KGff}. This is not the case for the double-dimer model.

\bigskip

\noindent{\bf Main results.} Let us now summarise the main results. 
We will show that in the double-dimer model the coupling function $C(u,v)$ has a factorization into a product of two discrete holomorphic functions $\F(u)$ and $\G(v)$ described in Corollary~\ref{main-lemma}. 
Moreover, we will describe the construction of the discrete integral of this product of two discrete holomorphic functions. Then for any discrete domain the expectation of the height function of the double-dimer model can be interpreted as an integral of two discrete holomorphic functions. Due to Kenyon~\cite{Kdom}, for the single-dimer model, the expectation of the height function is harmonic in the limit for approximations by Temperleyan domains. Using the above-mentioned factorization of the double-dimer coupling function we will show that the expectation of the double-dimer height function is harmonic already at the discrete level, with respect to the leap-frog Laplacian, see~(\ref{lfH}). In other words, we have the following result. 

\begin{Th}\label{leap-frog harmonicity} The expectation of the double-dimer height function on an odd Temperleyan domain (see Section~\ref{hf_and_Td} for a precise definition) is exactly discrete leap-frog harmonic.

\end{Th}
\noindent Note that the exact discrete harmonicity does not hold for the single-dimer model.  

Also, we will prove the convergence of the expectation of the (integer-valued) height function in the double-dimer model to the harmonic measure under discretization by polygonal domains. More precisely, we have the following result.

\begin{Th}\label{main-th} 
Let $\om$ be a polygon with $n$ sides parallel to the axes and two marked points $\bbb$ and $\www$ on straight parts of the boundary of $\om$.  Suppose that a sequence of discrete $n$-gons $\om^\delta$ on a grid with mesh size $\delta$ approximates the polygon $\om$ in a proper way, and that each polygon $\om^\delta$ has at least one domino tiling. Assume that some black and white squares $\bbb^\delta$ and $\www^\delta$ of the domain $\om^\delta$ tend to the boundary points $\bbb$ and $\www$ of the domain $\om$. Let $h^\delta$ be the height function of a uniform double-dimer configuration on $\om$. Then $\mathbb{E}h^\delta$ converges to the harmonic measure $\operatorname{hm}_{\om}(\,\cdot\,, (\bbb\www))$ of the boundary arc~$(\bbb\www)$ on the domain $\om$.
 \end{Th}
 
 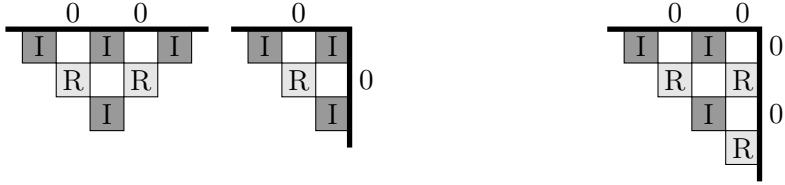
\begin{figure}
 \begin{center}
\begin{tikzpicture}[x={(0.45cm,0cm)}, y={(0cm,0.45cm)}]
\begin{scope}
\path (0,0) node[name=l1, shape=coordinate]{};
\path (1,0) node[name=l2, shape=coordinate]{};
\path (1,1) node[name=l3, shape=coordinate]{};
\path (0,1) node[name=l4, shape=coordinate]{};
\path[draw, fill=gray!80] (l1)--(l2)--(l3)--(l4)--cycle;
\path ($1/2*(l1)+1/2*(l3)$) node[]{I};
\path (1.5,1.5) node[]{$0$};
\path (-0.5,1.5) node[]{$0$};

\path (0,-2) node[name=l1, shape=coordinate]{};
\path (1,-2) node[name=l2, shape=coordinate]{};
\path (1,-1) node[name=l3, shape=coordinate]{};
\path (0,-1) node[name=l4, shape=coordinate]{};
\path[draw, fill=gray!80] (l1)--(l2)--(l3)--(l4)--cycle;
\path ($1/2*(l1)+1/2*(l3)$) node[]{I};

\path (-2,0) node[name=l1, shape=coordinate]{};
\path (-1,0) node[name=l2, shape=coordinate]{};
\path (-1,1) node[name=l3, shape=coordinate]{};
\path (-2,1) node[name=l4, shape=coordinate]{};
\path[draw, fill=gray!80] (l1)--(l2)--(l3)--(l4)--cycle;
\path ($1/2*(l1)+1/2*(l3)$) node[]{I};

\path[draw, fill=gray!20] (-1,0)--(-1,-1)--(0,-1)--(0,0)--cycle;
\path[draw, fill=gray!20] (1,0)--(1,-1)--(2,-1)--(2,0)--cycle;
\path (-0.5,-0.5) node[]{R};
\path (1.5,-0.5) node[]{R};

\path (2,0) node[name=l1, shape=coordinate]{};
\path (3,0) node[name=l2, shape=coordinate]{};
\path (3,1) node[name=l3, shape=coordinate]{};
\path (2,1) node[name=l4, shape=coordinate]{};
\path[draw, fill=gray!80] (l1)--(l2)--(l3)--(l4)--cycle;
\path ($1/2*(l1)+1/2*(l3)$) node[]{I};

\path[draw][line width=2pt] (-2.5,1)--(3.5,1);

\end{scope}

\begin{scope}[xshift=3cm]
\path (0,0) node[name=l1, shape=coordinate]{};
\path (1,0) node[name=l2, shape=coordinate]{};
\path (1,1) node[name=l3, shape=coordinate]{};
\path (0,1) node[name=l4, shape=coordinate]{};
\path[draw, fill=gray!80] (l1)--(l2)--(l3)--(l4)--cycle;
\path ($1/2*(l1)+1/2*(l3)$) node[]{I};
\path (1.5,-0.5) node[]{$0$};
\path (-0.5,1.5) node[]{$0$};

\path[draw, fill=gray!20] (0,0)--(-1,0)--(-1,-1)--(0,-1)--cycle;
\path (-0.5,-0.5) node[]{R};

\path (0,-2) node[name=l1, shape=coordinate]{};
\path (1,-2) node[name=l2, shape=coordinate]{};
\path (1,-1) node[name=l3, shape=coordinate]{};
\path (0,-1) node[name=l4, shape=coordinate]{};
\path[draw, fill=gray!80] (l1)--(l2)--(l3)--(l4)--cycle;
\path ($1/2*(l1)+1/2*(l3)$) node[]{I};

\path (-2,0) node[name=l1, shape=coordinate]{};
\path (-1,0) node[name=l2, shape=coordinate]{};
\path (-1,1) node[name=l3, shape=coordinate]{};
\path (-2,1) node[name=l4, shape=coordinate]{};
\path[draw, fill=gray!80] (l1)--(l2)--(l3)--(l4)--cycle;
\path ($1/2*(l1)+1/2*(l3)$) node[]{I};

\path[draw][line width=2pt] (-2.5,1)--(1,1)--(1,-2.5);

\end{scope}

\begin{scope}[xshift=8cm]

\path (0,0) node[name=l1, shape=coordinate]{};
\path (1,0) node[name=l2, shape=coordinate]{};
\path (1,1) node[name=l3, shape=coordinate]{};
\path (0,1) node[name=l4, shape=coordinate]{};
\path[draw, fill=gray!80] (l1)--(l2)--(l3)--(l4)--cycle;
\path ($1/2*(l1)+1/2*(l3)$) node[]{I};
\path (-0.5,1.5) node[]{0};
\path (1.5,1.5) node[]{0};
\path (2.5,0.5) node[]{0};
\path (2.5,-1.5) node[]{0};

\path (0,-2) node[name=l1, shape=coordinate]{};
\path (1,-2) node[name=l2, shape=coordinate]{};
\path (1,-1) node[name=l3, shape=coordinate]{};
\path (0,-1) node[name=l4, shape=coordinate]{};
\path[draw, fill=gray!80] (l1)--(l2)--(l3)--(l4)--cycle;
\path ($1/2*(l1)+1/2*(l3)$) node[]{I};

\path (-2,0) node[name=l1, shape=coordinate]{};
\path (-1,0) node[name=l2, shape=coordinate]{};
\path (-1,1) node[name=l3, shape=coordinate]{};
\path (-2,1) node[name=l4, shape=coordinate]{};
\path[draw, fill=gray!80] (l1)--(l2)--(l3)--(l4)--cycle;
\path ($1/2*(l1)+1/2*(l3)$) node[]{I};

\path[draw, fill=gray!20] (-1,-1)--(0,-1)--(0,0)--(-1,0)--cycle;
\path[draw, fill=gray!20] (1,-1)--(2,-1)--(2,0)--(1,0)--cycle;
\path[draw, fill=gray!20] (1,-3)--(2,-3)--(2,-2)--(1,-2)--cycle;
\path (-0.5,-0.5) node[]{R};
\path (1.5,-0.5) node[]{R};
\path (1.5,-2.5) node[]{R};

\path[draw][line width=2pt] (-2.5,1)--(2,1)--(2,-3.5);
\end{scope}
\end{tikzpicture}
\end{center}
\caption{The coupling function $C(u,\www)$ with fixed $\www$ is real on the set of light grey squares, and it is pure imaginary on the set of dark grey squares. On the left pictures the coupling function restricted to the light grey squares satisfies the Dirichlet boundary conditions, and the coupling function restricted to the dark grey squares obeys Neumann boundary conditions. 
The picture on the right corresponds to mixed Dirichlet and Neumann boundary conditions for the coupling function.}\label{b_c}
\end{figure}

Furthermore, we will show the convergence of the dimer coupling function in the case of approximations by {\it black-piecewise Temperleyan domains} (see Fig.~\ref{pwTemp}), domains which correspond to mixed Dirichlet and Neumann boundary conditions for the coupling function (see Fig.~\ref{b_c}). For a more precise statement, see Theorem~\ref{main-th2_1}.
Note that the coupling function $C(u,\www)$ with fixed $\www$ coincides with a discrete holomorphic function $\F(u)$. 
Let $\F^\delta$ be equal $\frac{1}{\delta}\F$ on a domain $\om^\delta$ of mesh size $\delta$, we have the following result (for a more precise statement, see Theorem~\ref{convF}).

\begin{Th}\label{main-convF}
Let $\om^\delta$ be a sequence of discrete $2k$-black-piecewise Temperleyan domains of mesh size $\delta$ approximating a continuous domain $\om$. Suppose that each $\om^\delta$ admits a domino tiling. Assume that white square $\www^\delta$ of the domain $\om^\delta$ tends to the boundary point $\www$ of the domain $\om$. Then $\F^\delta$ converges uniformly on compact subsets of  $\om$ to a continuous holomorphic function $f$ with a singularity at $\www$, as $\delta$ tends to~$0$. 
\end{Th}
Similarly, one can show the convergence of $\G^\delta=\frac{1}{\delta}\G$ for approximations by white-piecewise Temperleyan domains. Note that a polygonal domain $\om^{\delta}$ as in Theorem~\ref{main-th} is black-piecewise Temperleyan and also white-piecewise Temperleyan. Thus, we obtain the convergence of the double-dimer coupling function for any polygonal domain.

It is known that 
all moments of the scaling limit of the height function can be written in terms of the scaling limit of the coupling function, see~\cite{KGff}. Thus, adopting the proof of~\cite[Theorem 1.1]{KGff}, we obtain the convergence of the dimer height function to the Gaussian Free Field in the setup of Theorem~\ref{main-convF}. More precisely, we have the following result.

 \begin{cor}\label{main-cor2}
 Let $\om$ be a Jordan domain with smooth boundary in $\mathbb{R}^2$. Let $\om^\delta$ be a black-piecewise Temperleyan domain approximating $\om$. Let $h^\delta$ be the height function of $\om^\delta$. 
Then $h^\delta - \mathbb{E}{h^\delta}$ converges weakly in distribution to the Gaussian Free Field on $\om$ with Dirichlet boundary conditions, as $\delta$ tends to $0$.
 \end{cor}
 
As it was shown in~\cite{KGff} it is enough to compute all the limit moments of the fluctuations $h^\delta - \mathbb{E}{h^\delta}$ of the height function to prove that their limit is the Gaussian Free Field. The main tool to compute these moments is the coupling function. The scaling limit of the coupling function is very sensible to the boundary conditions, in particular the limits of the coupling function in the Temperleyan case and the piecewise Temperleyan case are different. However all the limits of $h^\delta - \mathbb{E}{h^\delta}$ turns out to be the same. 
 
 \bigskip
 
\noindent{\bf Organization of the paper.} The rest of the paper is organized as follows. In Section~\ref{2} we recall some basic facts and definitions. Section~\ref{S3} contains the construction of the primitive of the product of two discrete holomorphic functions. Also, we show that for an appropriate choice of the boundary conditions of discrete holomorphic functions the primitive of their product coincides with the expectation of the height function in the double-dimer model and we prove Theorem~\ref{leap-frog harmonicity}. In Section~\ref{4} we show that the continuos analogue of the above-mentioned primitive is the harmonic measure $\operatorname{hm}_{\om}(\,\cdot\,, (\bbb\www))$ of the boundary arc~$(\bbb\www)$ on the domain $\om$ in the setup of Theorem~\ref{main-th}. In Section~\ref{shodimost} we prove Theorem~\ref{main-convF}. Finally, Section~\ref{6} contains results about the
single dimer model. 
 
\bigskip
\medskip

\noindent{\bf Acknowledgements.} 
The author thanks Stanislav Smirnov for valuable insights, and Dmitry Chelkak for numerous useful discussions. 
The author also thanks the reviewers for their helpful comments and suggestions. 
Research of Theorem~\ref{main-th}  and Proposition~\ref{C=FG}
is supported by the Russian Science Foundation grant 14-21-00035. 
The author also received partial support from the NCCR SwissMAP of the SNSF and ERC AG COMPASP.

\medskip

\setcounter{equation}{0}
\section{Definitions and basic facts}\label{2}

\subsection{Height function and Temperleyan domain}\label{hf_and_Td}
Consider a checkerboard tiling of a discrete domain $\om$ with unit squares. 
We will use grey color for the black squares in our figures.
Sometimes for convenience we will distinguish between two types of black squares, in this case in the figures black squares in even rows will be represented by a 
light grey and those in odd rows will be dark grey (see Fig.~\ref{Temp}). 
 A domain where all corner squares are dark grey is called an {\it odd Temperleyan domain}. To obtain the {\it Temperleyan domain} one removes one dark grey square adjacent to the boundary from an odd Temperleyan domain. 

W. P. Thurston \cite{Ter} defines the {\it height function} $h$ (which is a real-valued function on the vertices of $\om$) as follows. 
Fix a vertex $z_0$ and set $h(z_0)=0$.
For every other vertex $z$ in the tiling, take an edge-path $\gamma$ from $z_0$ to $z$. The height along $\gamma$ changes by $\pm \frac14$ if the traversed edge does not cross a domino from the tiling or by $\mp \frac34$ otherwise: if the traversed edge has a black square on its left then the height increases by $\frac14$ or decreases by $\frac34$; if it has a white
square on its left then it decreases by $\frac14$ or increases by $\frac34$, see Fig.~\ref{hf}. Note that for a simply connected domain, the height is independent of the choice of $\gamma$. 
The height function in the double-dimer model is defined as the difference of the height functions of the two corresponding dimer coverings.

\begin{figure}
\begin{center}
\begin{tikzpicture}
\begin{scope}
\draw (0,0) -- (2,0) -- (2,1) -- (3,1) -- (3,3) -- (0,3)-- cycle;

\draw  [draw, fill=gray!20](1,0) rectangle (2,1);
\draw  [draw, fill=gray!20](2,1) rectangle (3,2);
\draw  [draw, fill=gray!20](1,2) rectangle (2,3);
\draw  [draw, fill=gray!20](0,1) rectangle (1,2);

\draw[line width=1.3pt] (0,0) rectangle (2,1);
\draw[line width=1.3pt] (1,1) rectangle (3,2);
\draw[line width=1.3pt] (1,2) rectangle (3,3);
\draw[line width=1.3pt] (0,1) rectangle (1,3);

\fill[black] (0,0)  node[below,font=\tiny]{$z_0$};
\fill[black] (2.2,2) node[above,font=\tiny]{$z$};

\fill[black] (0,0) circle (1.5pt);
\fill[black] (2,2) circle (1.5pt);
\end{scope}

\begin{scope}[xshift=6cm]
\draw (0,0) -- (2,0) -- (2,1) -- (3,1) -- (3,3) -- (0,3)-- cycle;

\draw  [draw, fill=gray!20](1,0) rectangle (2,1);
\draw  [draw, fill=gray!20](2,1) rectangle (3,2);
\draw  [draw, fill=gray!20](1,2) rectangle (2,3);
\draw  [draw, fill=gray!20](0,1) rectangle (1,2);

\draw[line width=1.3pt] (0,0) rectangle (2,1);
\draw[line width=1.3pt] (1,1) rectangle (3,2);
\draw[line width=1.3pt] (1,2) rectangle (3,3);
\draw[line width=1.3pt] (0,1) rectangle (1,3);

\draw[gray!80, ultra thick, dashed] (0,0) -- (0,1) -- (2,1) -- (2,2);

\fill[black] (0,0)  node[below,font=\tiny]{$0$};
\fill[black] (0,1) node[left,font=\tiny]{$-\frac14$};
\fill[black] (0.7,1) node[below,font=\tiny]{$-\frac24$};
\fill[black] (1.7,1) node[above,font=\tiny]{$-\frac14$};
\fill[black] (2.2,2) node[above,font=\tiny]{$-\frac44$};

\fill[black] (0,0) circle (1.5pt);
\fill[black] (0,1) circle (1.5pt);
\fill[black] (1,1) circle (1.5pt);
\fill[black] (2,1) circle (1.5pt);
\fill[black] (2,2) circle (1.5pt);
\end{scope}

\end{tikzpicture} \caption{Left: a domino tiling of the domain, vertices $z_0$ and $z$. Right: an edge-path from $z_0$ to $z$ and the height along this path: 
$h^\delta(z_0)=0$, $h^\delta(z)=-1$.
}\label{hf}\end{center}
\end{figure}
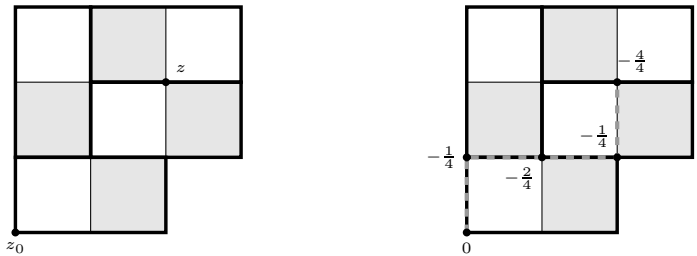

\subsection{Kasteleyn weights and discrete holomorphic functions}\label{coupl}
Let $G$ be a bipartite graph with $n$ black and $n$ white vertices. A {\it Kasteleyn matrix} $K_G$ is an $n\times n$ weighted adjacency matrix whose rows index the black vertices and columns index the white vertices. Let us denote by $\tau(u,v)$ an element of this matrix, where $u$ and $v$ are adjacent  black and white vertices. For finite planar bipartite graphs Kasteleyn~\cite{Kast} proved that if the edge-weights are Kasteleyn, i.e. the alternating product of the weights along any simple face of degree $p$ is equal to $(-1)^{(p+2)/2}$, then the absolute value of the determinant of the Kasteleyn matrix is equal to the number of perfect matchings of the graph.

Kenyon showed how to compute local statistics for the uniform measure on dimer configurations on a planar graph, using the inverse of the Kasteleyn matrix.
Let $E$ be a finite collection of disjoint edges of $\om$. Let $\mu$ be the uniform probability measure on  perfect matchings of $\om$. Let $b_1,\ldots,b_k$ and $w_1,\ldots,w_k$ be the black and white vertices of the edges belonging to $E$ correspondingly.

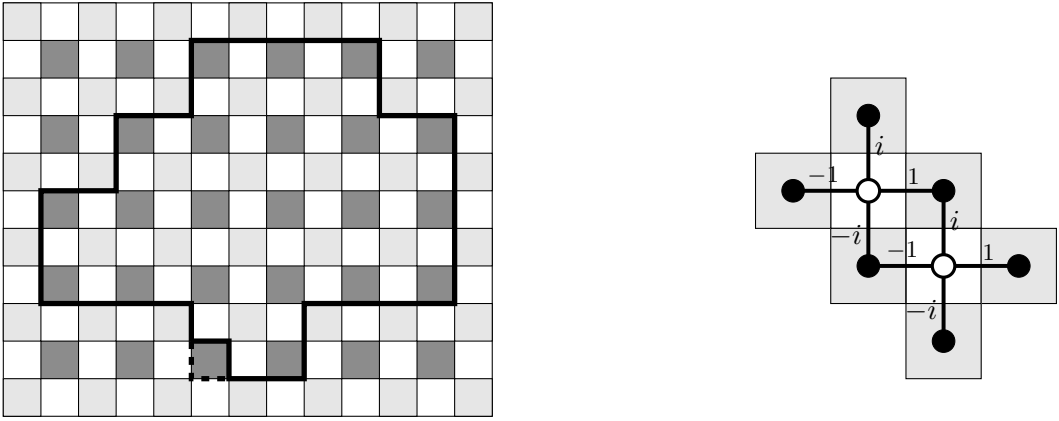
\begin{figure}
\begin{center}
\begin{tikzpicture}[x={(0.5cm,0cm)}, y={(0cm,0.5cm)}]

\begin{scope}
\draw (0,0) rectangle (13,11);

\draw  [draw, fill=gray!20](0,0) rectangle (1,1);
\draw  [draw, fill=gray!20](0,2) rectangle (1,3);
\draw  [draw, fill=gray!20](0,4) rectangle (1,5);
\draw  [draw, fill=gray!20](0,6) rectangle (1,7);
\draw  [draw, fill=gray!20](0,8) rectangle (1,9);
\draw  [draw, fill=gray!20](0,10) rectangle (1,11);

\draw  [draw, fill=gray!20](2,0) rectangle (3,1);
\draw  [draw, fill=gray!20](2,2) rectangle (3,3);
\draw  [draw, fill=gray!20](2,4) rectangle (3,5);
\draw  [draw, fill=gray!20](2,6) rectangle (3,7);
\draw  [draw, fill=gray!20](2,8) rectangle (3,9);
\draw  [draw, fill=gray!20](2,10) rectangle (3,11);

\draw  [draw, fill=gray!20](4,0) rectangle (5,1);
\draw  [draw, fill=gray!20](4,2) rectangle (5,3);
\draw  [draw, fill=gray!20](4,4) rectangle (5,5);
\draw  [draw, fill=gray!20](4,6) rectangle (5,7);
\draw  [draw, fill=gray!20](4,8) rectangle (5,9);
\draw  [draw, fill=gray!20](4,10) rectangle (5,11);

\draw  [draw, fill=gray!20](6,0) rectangle (7,1);
\draw  [draw, fill=gray!20](6,2) rectangle (7,3);
\draw  [draw, fill=gray!20](6,4) rectangle (7,5);
\draw  [draw, fill=gray!20](6,6) rectangle (7,7);
\draw  [draw, fill=gray!20](6,8) rectangle (7,9);
\draw  [draw, fill=gray!20](6,10) rectangle (7,11);

\draw  [draw, fill=gray!20](8,0) rectangle (9,1);
\draw  [draw, fill=gray!20](8,2) rectangle (9,3);
\draw  [draw, fill=gray!20](8,4) rectangle (9,5);
\draw  [draw, fill=gray!20](8,6) rectangle (9,7);
\draw  [draw, fill=gray!20](8,8) rectangle (9,9);
\draw  [draw, fill=gray!20](8,10) rectangle (9,11);

\draw  [draw, fill=gray!20](10,0) rectangle (11,1);
\draw  [draw, fill=gray!20](10,2) rectangle (11,3);
\draw  [draw, fill=gray!20](10,4) rectangle (11,5);
\draw  [draw, fill=gray!20](10,6) rectangle (11,7);
\draw  [draw, fill=gray!20](10,8) rectangle (11,9);
\draw  [draw, fill=gray!20](10,10) rectangle (11,11);

\draw  [draw, fill=gray!20](12,0) rectangle (13,1);
\draw  [draw, fill=gray!20](12,2) rectangle (13,3);
\draw  [draw, fill=gray!20](12,4) rectangle (13,5);
\draw  [draw, fill=gray!20](12,6) rectangle (13,7);
\draw  [draw, fill=gray!20](12,8) rectangle (13,9);
\draw  [draw, fill=gray!20](12,10) rectangle (13,11);

\draw  [draw, fill=gray!90](1,1) rectangle (2,2);
\draw  [draw, fill=gray!90](1,3) rectangle (2,4);
\draw  [draw, fill=gray!90](1,5) rectangle (2,6);
\draw  [draw, fill=gray!90](1,7) rectangle (2,8);
\draw  [draw, fill=gray!90](1,9) rectangle (2,10);

\draw  [draw, fill=gray!90](3,1) rectangle (4,2);
\draw  [draw, fill=gray!90](3,3) rectangle (4,4);
\draw  [draw, fill=gray!90](3,5) rectangle (4,6);
\draw  [draw, fill=gray!90](3,7) rectangle (4,8);
\draw  [draw, fill=gray!90](3,9) rectangle (4,10);

\draw  [draw, fill=gray!90](5,1) rectangle (6,2);
\draw  [draw, fill=gray!90](5,3) rectangle (6,4);
\draw  [draw, fill=gray!90](5,5) rectangle (6,6);
\draw  [draw, fill=gray!90](5,7) rectangle (6,8);
\draw  [draw, fill=gray!90](5,9) rectangle (6,10);

\draw  [draw, fill=gray!90](7,1) rectangle (8,2);
\draw  [draw, fill=gray!90](7,3) rectangle (8,4);
\draw  [draw, fill=gray!90](7,5) rectangle (8,6);
\draw  [draw, fill=gray!90](7,7) rectangle (8,8);
\draw  [draw, fill=gray!90](7,9) rectangle (8,10);

\draw  [draw, fill=gray!90](9,1) rectangle (10,2);
\draw  [draw, fill=gray!90](9,3) rectangle (10,4);
\draw  [draw, fill=gray!90](9,5) rectangle (10,6);
\draw  [draw, fill=gray!90](9,7) rectangle (10,8);
\draw  [draw, fill=gray!90](9,9) rectangle (10,10);

\draw  [draw, fill=gray!90](11,1) rectangle (12,2);
\draw  [draw, fill=gray!90](11,3) rectangle (12,4);
\draw  [draw, fill=gray!90](11,5) rectangle (12,6);
\draw  [draw, fill=gray!90](11,7) rectangle (12,8);
\draw  [draw, fill=gray!90](11,9) rectangle (12,10);

\draw[draw, line width=2pt,dashed] (5,2)--(5,1)--(6,1);
\draw[draw, line width=2pt]  
(5,2)--(6,2)--(6,1)--(8,1)--(8,3)--(12,3)--(12,8)--(10,8)--(10,10)--(5,10)--(5,8)--(3,8)--(3,6)--(1,6)--(1,3)--(5,3)--(5,2);

\end{scope}

\begin{scope}[xshift=12cm]
\draw  (-2,7) rectangle (0,5) node (v3) {};

\draw  [draw, fill=gray!20](-2,9) rectangle (0,7) node (v1) {};

\draw [draw, fill=gray!20] (v1) rectangle (2,5) node (v4) {};

\draw  [draw, fill=gray!20](-4,7) rectangle (-2,5) node (v2) {};

\draw  [draw, fill=gray!20](v2) rectangle (0,3) node (v5) {};

\draw  (v3) rectangle (2,3);
\draw  [draw, fill=gray!20](v4) rectangle (4,3);
\draw  [draw, fill=gray!20](v5) rectangle (2,1);

\draw[draw, line width=1.5pt]  (-1,8) edge (-1,4);
\draw (-1.1,7.2) node[anchor=west]{$i$};
\draw (-0.9,4.8) node[anchor=east]{$-i$};
\draw[draw, line width=1.5pt]  (-3,6) edge (1,6);
\draw (0.2,5.9) node[anchor=south,font=\small]{$1$};
\draw (-2.2,5.9) node[anchor=south,font=\small]{$-1$};

\draw[draw, line width=1.5pt]  (1,6) edge (1,2);
\draw (0.9,5.2) node[anchor=west]{$i$};
\draw (1.1,2.8) node[anchor=east]{$-i$};
\draw[draw, line width=1.5pt]  (-1,4) edge (3,4);
\draw (2.2,3.9) node[anchor=south,font=\small]{$1$};
\draw (-0.1,3.9) node[anchor=south,font=\small]{$-1$};

\path[draw, fill=black] (1,6) node (v9) {} circle[radius=0.15cm];
\path[draw, fill=black] (-3,6) node (v8) {} circle[radius=0.15cm];
\path[draw, fill=black] (-1,8) node (v6) {} circle[radius=0.15cm];
\path[draw, fill=black] (-1,4) node (v7) {} circle[radius=0.15cm];
\path[draw,line width=1.2pt, fill=white] (-1,6) node () {} circle[radius=0.15cm];

\path[draw,line width=1.2pt, fill=white] (1,4) node () {} circle[radius=0.15cm];
\path[draw, fill=black] (1,2) node (v6) {} circle[radius=0.15cm];
\path[draw, fill=black] (3,4) node (v7) {} circle[radius=0.15cm];
\end{scope}
\end{tikzpicture} \caption{Left: A Temperleyan domain. Right: Weights of the Kasteleyn matrix on the square lattice (proposed by Kenyon in~\cite{Kdom}):  at each white vertex the four edge weights
going counterclockwise from the right-going edge are $1$, $i$, $-1$, $-i$ respectively.}\label{Kasteleyn}\label{Temp}\end{center}
\end{figure}


\begin{Th}[\cite{Klocstat}]\label{locstat}
The $\mu$-probability that the set $E$ occurs in a perfect matching is given by $|\det(K_{E}^{-1})|$, where $K_{E}^{-1}$ is the submatrix of $K_\om^{-1}$ whose rows are indexed by $b_1,\ldots,b_k$ and columns are indexed by $w_1,\ldots,w_k$. More precisely, the probability is $c\cdot(-1)^{\sum p_i+ q_j}\cdot a_E\cdot\det(K_{E}^{-1})$, where $p_i, q_i$ is the index of $b_i$, resp. $w_i$, in a fixed ordering of the vertices, $c=\pm1$ is a constant depending only on that ordering, and $a_E$ is the product of the edge weights of the edges $E$.
\end{Th}

For a given planar graph $G$, there are many ways to choose the edge-weights satisfying the Kasteleyn condition. Let us fix the following ones, which were proposed by Kenyon in~\cite{Kdom}: put $\tau(e)=\pm1$ for horizontal edges and $\tau(e)=\pm i$ if $e$ is a vertical edge, see Fig.~\ref{Kasteleyn}. It is easy to check that these weights are Kasteleyn weights.

Let $\om$ be a discrete domain on a square lattice that has at least one domino tiling. Let $K_\om$ be a Kasteleyn matrix of this domain. 
Let us denote by $C_\om(u,v)$ the elements of the inverse matrix $K_\om^{-1}$, where $u$ and $v$ are black and white squares of $\om$. The main advantage of choosing Kasteleyn weights as shown in Fig.~\ref{Kasteleyn} is the following: with this choice of weights the function $C_\om(u,v)$ is discrete holomorphic on the domain. Thus its limiting behavior can be studied using the methods of discrete complex analysis, see~\cite{Kdom}. Following~\cite{Kdom}, we call $C_\om(u,v)$ the {\it coupling function}.

Let $F$ be a function defined on the set of black squares of the domain $\om$. Recall that the function $F$ is called {\it discrete holomorphic} on $\om$ if for any white square $v\in\om$ it satisfies a discrete analogue of the Cauchy-Riemann equation (see Fig.~\ref{C-R}), and at the same time the values of the function $F$ on the set of light grey squares are real, while on the set of dark grey squares they are purely imaginary. 
Note that the real and imaginary parts of a holomorphic function are harmonic functions. 
It is also true on a discrete level: consider the discrete Cauchy-Riemann equations at four white neighbours of a black square $u$, then it is easy to show that $F(u)=\frac{1}{4}\sum_{i=1}^4\F(u_i)$. 
Therefore the {\it discrete leap-frog Laplacian} of $F$ at $u$ equals zero (see. Fig.~\ref{LL}). In other words, real and imaginary parts of discrete holomorphic functions are discrete harmonic functions.

\begin{figure}
\begin{center}
\begin{tikzpicture}[x={(0.5cm,0cm)}, y={(0cm,0.5cm)}]

\begin{scope}
\path (-1,0) node[name=l1, shape=coordinate]{};
\path (0,0) node[name=l2, shape=coordinate]{};
\path (0,1) node[name=l3, shape=coordinate]{};
\path (-1,1) node[name=l4, shape=coordinate]{};
\path[draw, fill=gray!20] (l1)--(l2)--(l3)--(l4)--cycle;
\path ($1/2*(l1)+1/2*(l3)$) node[]{$a$};

\path (0,1) node[name=u1, shape=coordinate]{};
\path (1,1) node[name=u2, shape=coordinate]{};
\path (1,2) node[name=u3, shape=coordinate]{};
\path (0,2) node[name=u4, shape=coordinate]{};
\path[draw,fill=gray!20] (u1)--(u2)--(u3)--(u4)--cycle;
\path ($1/2*(u1)+1/2*(u3)$) node[]{$d$};

\path (1,0) node[name=r1, shape=coordinate]{};
\path (2,0) node[name=r2, shape=coordinate]{};
\path (2,1) node[name=r3, shape=coordinate]{};
\path (1,1) node[name=r4, shape=coordinate]{};
\path[draw, fill=gray!20] (r1)--(r2)--(r3)--(r4)--cycle;
\path ($1/2*(r1)+1/2*(r3)$) node[]{$c$};

\path (0,-1) node[name=d1, shape=coordinate]{};
\path (1,-1) node[name=d2, shape=coordinate]{};
\path (1,0) node[name=d3, shape=coordinate]{};
\path (0,0) node[name=d4, shape=coordinate]{};
\path[draw, fill=gray!20] (d1)--(d2)--(d3)--(d4)--cycle;
\path ($1/2*(d1)+1/2*(d3)$) node[]{$b$};

\path ($1/2*(u2)+1/2*(d4)$) node[]{$v$};

\path (0,-3.5) node[]{$F(c)-F(a)=-i\cdot(F(d)-F(b))$};
\end{scope}

\begin{scope}[xshift=6cm]
\path (1,-1) node[name=dd, shape=coordinate]{};
\path (4,0) node[name=ddd, shape=coordinate]{};

\path (-1,0) node[name=l1, shape=coordinate]{};
\path (0,0) node[name=l2, shape=coordinate]{};
\path (0,1) node[name=l3, shape=coordinate]{};
\path (-1,1) node[name=l4, shape=coordinate]{};
\path[draw, fill=gray!20] (l1)--(l2)--(l3)--(l4)--cycle;
\path ($1/2*(l1)+1/2*(l3)$) node[]{$u_2$};
\path[draw,fill=gray!20] ($(l1)+(ddd)$)--($(l2)+(ddd)$)--($(l3)+(ddd)$)--($(l4)+(ddd)$)--cycle;
\path ($1/2*(l1)+1/2*(l3)+(ddd)$) node[]{$u_4$};

\path (0,1) node[name=u1, shape=coordinate]{};
\path (1,1) node[name=u2, shape=coordinate]{};
\path (1,2) node[name=u3, shape=coordinate]{};
\path (0,2) node[name=u4, shape=coordinate]{};
\path[draw,fill=gray!70] (u1)--(u2)--(u3)--(u4)--cycle;

\path[draw,fill=gray!20] ($(u1)+(u2)$)--($(u2)+(u2)$)--($(u3)+(u2)$)--($(u4)+(u2)$)--cycle;
\path ($1/2*(u1)+1/2*(u3)+(u2)$) node[]{$u_1$};

\path (1,0) node[name=r1, shape=coordinate]{};
\path (2,0) node[name=r2, shape=coordinate]{};
\path (2,1) node[name=r3, shape=coordinate]{};
\path (1,1) node[name=r4, shape=coordinate]{};
\path[draw, fill=gray!20] (r1)--(r2)--(r3)--(r4)--cycle;
\path ($1/2*(r1)+1/2*(r3)$) node[]{$u$};

\path[draw,fill=gray!70] ($(r1)+(u2)$)--($(r2)+(u2)$)--($(r3)+(u2)$)--($(r4)+(u2)$)--cycle;

\path[draw,fill=gray!70] ($(r1)+(dd)$)--($(r2)+(dd)$)--($(r3)+(dd)$)--($(r4)+(dd)$)--cycle;

\path (0,-1) node[name=d1, shape=coordinate]{};
\path (1,-1) node[name=d2, shape=coordinate]{};
\path (1,0) node[name=d3, shape=coordinate]{};
\path (0,0) node[name=d4, shape=coordinate]{};
\path[draw, fill=gray!70] (d1)--(d2)--(d3)--(d4)--cycle;
\path[draw,fill=gray!20] ($(d1)+(dd)$)--($(d2)+(dd)$)--($(d3)+(dd)$)--($(d4)+(dd)$)--cycle;
\path ($1/2*(d1)+1/2*(d3)+(dd)$) node[]{$u_3$};

\path (2,-3.5) node[]{$[\Delta\F](u)=\frac14\sum_{s=1}^4(\F(u_s)-\F(u))
$};

\end{scope}

\end{tikzpicture}\caption{Left: Discrete Cauchy-Riemann equation. Right: Discrete leap-frog Laplacian on the light grey lattice. The function $\F$ is called discrete harmonic at $u$ if $[\Delta \F](u)=0$.}\label{LL}\label{C-R}\end{center}
\end{figure}

The coupling function $C_\om(u,v)=K_\om^{-1}(u,v)$ can be extended to be zero on all boundary black squares, see Fig.~\ref{b_c}. We know that $K_\om^{-1}\cdot K_\om=I$, which is equivalent to the following relation:
\begin{equation}\label{wcr} 
\begin{split}
1\cdot C_\om (v+1, \www) - 1\cdot C_\om (v-1,\www)+
 i\cdot C_\om (v+i,\www)-i\cdot C_\om (v-i,\www)=\mathbb{1}_{\{v=\www\}}.
\end{split}
\end{equation} 
Note that this relation is the discrete Cauchy-Riemann equation for the coupling function $C_\om(\cdot,\www)$, so for any white square $\www\in\om$ the function $C_\om(u,\www)$ considered as a function of $u\in\om$ is discrete holomorphic on $\om\smallsetminus\{\www\}$, for more details see~\cite{Kdom}.
Therefore, the restriction of $C_{\om}(u,\www)$ to one type of black squares is a discrete harmonic function everywhere except the two squares adjacent to $\www$.

Moreover, the function $C_{\om}(u,v)$ satisfies the following property: 
 \begin{enumerate}
\item[$\rhd$] if $u$ and $v$ are adjacent squares, then $|C_\om(u,v)|$ is equal to the probability that the domino $[uv]$ is contained in a random domino tiling of~$\om$, see~\cite{Kdom}.
\end{enumerate}

For Temperleyan domains, each of the two discrete harmonic components of the function $C_{\om}(u,\www)$ has the following boundary conditions: the restriction of the coupling function to the light grey squares (see Fig.~\ref{b_c}), satisfies the Dirichlet boundary conditions, and coupling function restricted to the dark grey squares obeys Neumann boundary conditions.

\subsection{Even/Odd double dimers}A double-dimer configuration is the union of two dimer coverings.
We will consider coverings of a pair of domains $\om_1$, $\om_2$ that differ by two squares, i.e. $|\om_1 \bigtriangleup \om_2|=2$. Note that there are two different situations depending on whether $\om := \om_1 \cup\om_2$ contains an odd or an even number of squares. In the {\it odd} case, assume that $\om$ has one more black square than white squares. Then the domains $\om_1$ and $\om_2$ are obtained from $\om$ by removing black squares $u_1$ and $u_2$ adjacent to the boundary (see Fig.~\ref{I}).
In the {\it even} case, let $\om_1=\om$ and $\om_2$ be obtained from $\om$ by removing black and white squares $u_0$ and $v_0$, which are adjacent to the boundary. One can modify a domain in the odd case to reduce it to the even case, see Fig.~\ref{odd_even} and Remark~\ref{odd=even}.

Let us define the {\it double-dimer coupling function} on $\om=\om_1 \cup \om_2$ as the difference of the two dimer coupling functions on domains $\om_1$ and $\om_2$
\[C_{\operatorname{dbl-d}, \om}(u, v) := C_{\om_1}(u, v) - C_{\om_2}(u, v).\]

Recall that the absolute value of the coupling function is the probability that the corresponding domino is contained in a random tiling, and the determinant of the Kasteleyn matrix is equal to the number of domino tilings of our domain, so, $|C_{\om\smallsetminus\{\bbb, \www\}}(u, v)|=\left|\frac{\det(K_{\om\smallsetminus\{\bbb, \www, u, v\}})}{\det(K_{\om\smallsetminus\{\bbb, \www\}})} \right|.$  
Note that
\[
 \frac{\det(K_{\om\smallsetminus\{\bbb, \www\}})}{\det(K_\om)} = \pm K^{-1}_\om(\bbb, \www) \quad \text{and} \quad  \frac{\det(K_{\om\smallsetminus\{u, v\}})}{\det(K_\om)} = \pm K^{-1}_\om(u, v),
\]
and also
\[
 \frac{\det(K_{\om\smallsetminus\{\bbb, \www, u, v\}})}{\det(K_\om)} = \pm
 \det\begin{pmatrix} K^{-1}_\om(\bbb, \www) & K^{-1}_\om(\bbb, v) \\ K^{-1}_\om(u, \www) & K^{-1}_\om(u, v) \end{pmatrix}.
\]
Therefore, 
\[
C_{\operatorname{dbl-d}, \om}(u, v)=C_\om(u, v) - C_{\om\smallsetminus\{\bbb, \www\}}(u, v)= \pm\frac{K^{-1}_\om(\bbb, v)\cdot K^{-1}_\om(u, \www)}{K^{-1}_\om(\bbb, \www)}.
\]
Recall that for a fixed $\www$ the function $K_\om^{-1}(u, \www)$ is a discrete holomorphic function of $u$. Let us denote it by $F_{v_0}(u)$ and similarly let us define a function $G_{\bbb}(v) := K^{-1}_\om(\bbb, v)$. So, we obtain 
\[
C_{\operatorname{dbl-d}, \om}(u, v) = const_{\bbb,\www}\cdot F_{\www}(u) \cdot G_{\bbb}(v),
\]
where $const_{\bbb,\www}=\pm1/K^{-1}_\om(\bbb, \www)$. 

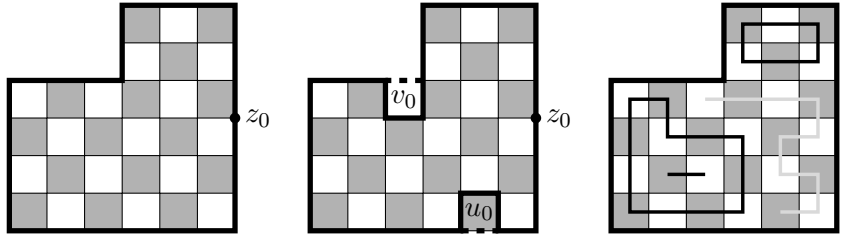
\begin{figure}
\begin{center}
\begin{tikzpicture}[x={(0.5cm,0cm)}, y={(0cm,0.5cm)}]
\begin{scope}
\draw  [draw, fill=gray!60](0,0) rectangle (1,1);
\draw  [draw, fill=gray!60](1,1) rectangle (2,2);
\draw  [draw, fill=gray!60](2,2) rectangle (3,3);
\draw  [draw, fill=gray!60](3,3) rectangle (4,4);
\draw  [draw, fill=gray!60](4,4) rectangle (5,5);
\draw  [draw, fill=gray!60](5,5) rectangle (6,6);

\draw  [draw, fill=gray!60](0,2) rectangle (1,3);
\draw  [draw, fill=gray!60](1,3) rectangle (2,4);
\draw  [draw, fill=gray!60](3,5) rectangle (4,6);

\draw  [draw, fill=gray!60](2,0) rectangle (3,1);
\draw  [draw, fill=gray!60](3,1) rectangle (4,2);
\draw  [draw, fill=gray!60](4,2) rectangle (5,3);
\draw  [draw, fill=gray!60](5,3) rectangle (6,4);

\draw  [draw, fill=gray!60](4,0) rectangle (5,1);
\draw  [draw, fill=gray!60](5,1) rectangle (6,2);

\draw[line width=2pt] (0,0) -- (0,4) -- (3,4) -- (3,6) -- (6,6) -- (6,0)-- cycle;

\fill[black] (6,3) circle (2pt) node[right]{$z_0$};
\end{scope}

\begin{scope}[xshift=4cm]

\draw  [draw, fill=gray!60](0,0) rectangle (1,1);
\draw  [draw, fill=gray!60](1,1) rectangle (2,2);
\draw  [draw, fill=gray!60](2,2) rectangle (3,3);
\draw  [draw, fill=gray!60](3,3) rectangle (4,4);
\draw  [draw, fill=gray!60](4,4) rectangle (5,5);
\draw  [draw, fill=gray!60](5,5) rectangle (6,6);

\draw  [draw, fill=gray!60](0,2) rectangle (1,3);
\draw  [draw, fill=gray!60](1,3) rectangle (2,4);
\draw  [draw, fill=gray!60](3,5) rectangle (4,6);

\draw  [draw, fill=gray!60](2,0) rectangle (3,1);
\draw  [draw, fill=gray!60](3,1) rectangle (4,2);
\draw  [draw, fill=gray!60](4,2) rectangle (5,3);
\draw  [draw, fill=gray!60](5,3) rectangle (6,4);

\draw  [draw, fill=gray!60](4,0) rectangle (5,1);
\draw  [draw, fill=gray!60](5,1) rectangle (6,2);

\draw[line width=2pt] (0,0) -- (0,4) --(2,4)--(2,3)--(3,3)-- (3,4)-- (3,6) -- (6,6) -- (6,0)--(5,0)--(5,1)--(4,1)--(4,0)-- cycle;

\draw[draw, line width=2pt,dashed] (3,4)--(2,4);
\draw[draw, line width=2pt,dashed] (4,0)--(5,0);

\fill[black] (6,3) circle (2pt) node[right]{$z_0$};

\fill[black] (4.5,0.5)  node[]{$u_0$};
\fill[black] (2.5,3.5)  node[]{$v_0$};
\end{scope}

\begin{scope}[xshift=8cm]
\draw  [draw, fill=gray!60](0,0) rectangle (1,1);
\draw  [draw, fill=gray!60](1,1) rectangle (2,2);
\draw  [draw, fill=gray!60](2,2) rectangle (3,3);
\draw  [draw, fill=gray!60](3,3) rectangle (4,4);
\draw  [draw, fill=gray!60](4,4) rectangle (5,5);
\draw  [draw, fill=gray!60](5,5) rectangle (6,6);

\draw  [draw, fill=gray!60](0,2) rectangle (1,3);
\draw  [draw, fill=gray!60](1,3) rectangle (2,4);
\draw  [draw, fill=gray!60](3,5) rectangle (4,6);

\draw  [draw, fill=gray!60](2,0) rectangle (3,1);
\draw  [draw, fill=gray!60](3,1) rectangle (4,2);
\draw  [draw, fill=gray!60](4,2) rectangle (5,3);
\draw  [draw, fill=gray!60](5,3) rectangle (6,4);

\draw  [draw, fill=gray!60](4,0) rectangle (5,1);
\draw  [draw, fill=gray!60](5,1) rectangle (6,2);

\draw[line width=2pt] (0,0) -- (0,4) -- (3,4) -- (3,6) -- (6,6) -- (6,0)-- cycle;

\draw[line width=1.3pt] (0.5,0.5) -- (0.5,3.5) -- (1.5,3.5) -- (1.5,2.5) -- (3.5,2.5) -- (3.5,0.5)-- cycle;

\draw[line width=1.3pt,gray!30] (2.5,3.5) -- (5.5,3.5) -- (5.5,2.5) -- (4.5,2.5) -- (4.5,1.5) -- (5.5,1.5)--(5.5,0.5)--(4.5,0.5);

\draw[line width=1.3pt] (1.5,1.5) -- (2.5,1.5);
\draw[line width=1.3pt] (3.5,4.5) rectangle (5.5,5.5);

\end{scope}
\end{tikzpicture} \caption{{\it Even case.} The vertex $z_0$ is the vertex on the boundary where  $h_1(z_0)=h_2(z_0)=h(z_0)=0$. The squares $u_0$ and $v_0$ are the difference between the  domains $\om_1=\om$ and $\om_2=\om\smallsetminus \{u_0, v_0\}$. On the right: an example of the interface (grey) from $u_0$ to $v_0$ and the set of loops and double edges in $\om$.
}\label{II}\end{center}
\end{figure}

\begin{figure}
\begin{center}
\begin{tikzpicture}[x={(0.5cm,0cm)}, y={(0cm,0.5cm)}]
\begin{scope}
\draw  [draw, fill=gray!60](0,0) rectangle (1,1);
\draw  [draw, fill=gray!60](1,1) rectangle (2,2);
\draw  [draw, fill=gray!60](2,2) rectangle (3,3);
\draw  [draw, fill=gray!60](3,3) rectangle (4,4);
\draw  [draw, fill=gray!60](4,4) rectangle (5,5);

\draw  [draw, fill=gray!60](0,2) rectangle (1,3);
\draw  [draw, fill=gray!60](1,3) rectangle (2,4);
\draw  [draw, fill=gray!60](2,4) rectangle (3,5);

\draw  [draw, fill=gray!60](0,4) rectangle (1,5);

\draw  [draw, fill=gray!60](2,0) rectangle (3,1);
\draw  [draw, fill=gray!60](3,1) rectangle (4,2);
\draw  [draw, fill=gray!60](4,2) rectangle (5,3);
\draw  [draw, fill=gray!60](5,3) rectangle (6,4);

\draw  [draw, fill=gray!60](4,0) rectangle (5,1);
\draw  [draw, fill=gray!60](5,1) rectangle (6,2);
\draw  [draw, fill=gray!60](6,2) rectangle (7,3);

\draw  [draw, fill=gray!60](6,0) rectangle (7,1);

\draw[line width=2pt] (1,0)--(1,1)--(0,1) -- (0,5) -- (6,5) -- (6,3) -- (7,3) -- (7,0)-- cycle;

\draw[draw, line width=2pt,dashed] (1,0)--(0,0)--(0,1);
\fill[black] (2,5) circle (2pt) node[above]{$z_0$};

\fill[black] (0.5,0.5)  node[]{$u_1$};

\end{scope}

\begin{scope}[xshift=4cm]
\draw  [draw, fill=gray!60](0,0) rectangle (1,1);
\draw  [draw, fill=gray!60](1,1) rectangle (2,2);
\draw  [draw, fill=gray!60](2,2) rectangle (3,3);
\draw  [draw, fill=gray!60](3,3) rectangle (4,4);
\draw  [draw, fill=gray!60](4,4) rectangle (5,5);

\draw  [draw, fill=gray!60](0,2) rectangle (1,3);
\draw  [draw, fill=gray!60](1,3) rectangle (2,4);
\draw  [draw, fill=gray!60](2,4) rectangle (3,5);

\draw  [draw, fill=gray!60](0,4) rectangle (1,5);

\draw  [draw, fill=gray!60](2,0) rectangle (3,1);
\draw  [draw, fill=gray!60](3,1) rectangle (4,2);
\draw  [draw, fill=gray!60](4,2) rectangle (5,3);
\draw  [draw, fill=gray!60](5,3) rectangle (6,4);

\draw  [draw, fill=gray!60](4,0) rectangle (5,1);
\draw  [draw, fill=gray!60](5,1) rectangle (6,2);
\draw  [draw, fill=gray!60](6,2) rectangle (7,3);

\draw  [draw, fill=gray!60](6,0) rectangle (7,1);

\draw[line width=2pt] (0,0) -- (0,5) -- (6,5) -- (6,4)--(5,4)--(5,3) -- (7,3) -- (7,0)-- cycle;

\draw[draw, line width=2pt,dashed] (6,4)--(6,3);

\fill[black] (2,5) circle (2pt) node[above]{$z_0$};
\fill[black] (5.5,3.5)  node[]{$u_2$};
\end{scope}

\begin{scope}[xshift=8cm]
\draw  [draw, fill=gray!60](0,0) rectangle (1,1);
\draw  [draw, fill=gray!60](1,1) rectangle (2,2);
\draw  [draw, fill=gray!60](2,2) rectangle (3,3);
\draw  [draw, fill=gray!60](3,3) rectangle (4,4);
\draw  [draw, fill=gray!60](4,4) rectangle (5,5);

\draw  [draw, fill=gray!60](0,2) rectangle (1,3);
\draw  [draw, fill=gray!60](1,3) rectangle (2,4);
\draw  [draw, fill=gray!60](2,4) rectangle (3,5);

\draw  [draw, fill=gray!60](0,4) rectangle (1,5);

\draw  [draw, fill=gray!60](2,0) rectangle (3,1);
\draw  [draw, fill=gray!60](3,1) rectangle (4,2);
\draw  [draw, fill=gray!60](4,2) rectangle (5,3);
\draw  [draw, fill=gray!60](5,3) rectangle (6,4);

\draw  [draw, fill=gray!60](4,0) rectangle (5,1);
\draw  [draw, fill=gray!60](5,1) rectangle (6,2);
\draw  [draw, fill=gray!60](6,2) rectangle (7,3);

\draw  [draw, fill=gray!60](6,0) rectangle (7,1);

\draw[line width=2pt] (0,0) -- (0,5) -- (6,5) -- (6,3) -- (7,3) -- (7,0)-- cycle;

\draw[line width=1.3pt, gray!30] (0.5,0.5) -- (0.5,2.5) -- (3.5,2.5) -- (3.5,3.5) -- (2.5,3.5) -- (2.5,4.5)--(5.5,4.5)--(5.5,3.5);
\draw[line width=1.3pt] (4.5,3.5) -- (4.5,2.5);
\draw[line width=1.3pt] (0.5,3.5) rectangle (1.5,4.5);
\draw[line width=1.3pt] (1.5,0.5) rectangle (2.5,1.5);
\draw[line width=1.3pt] (3.5,0.5) -- (3.5,1.5) -- (5.5,1.5) -- (5.5,2.5) -- (6.5,2.5) --(6.5,0.5)-- cycle;
\end{scope}
\end{tikzpicture} \caption{{\it Odd case.} The vertex $z_0$ is the vertex on the boundary where  \mbox{}$h_1(z_0)=h_2(z_0)=h(z_0)=0$. The squares $u_1$ and $u_2$ are the difference between the  domains $\om_1=\om\smallsetminus \{u_1\}$ and $\om_2=\om\smallsetminus \{u_2\}$. On the right: an example of the interface (grey) from $u_1$ to $u_2$ and the set of loops and double edges in $\om$.
}\label{I}\end{center}
\end{figure}
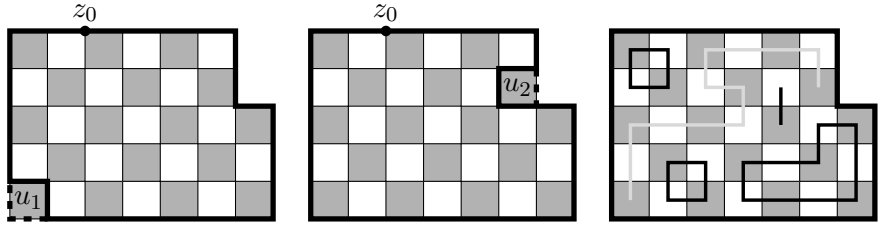

Below we discuss the above factorization of the double-dimer coupling function from the viewpoint of discrete holomorphic solutions to appropriate boundary value problems, the framework which we use later to prove the main convergence theorems.
%

\setcounter{equation}{0}
\section{Expectation of the height function in the double-dimer model and the proof of Theorem~\ref{leap-frog harmonicity}}\label{S3}

\subsection{Notation}
Put $\lambda = e^{i\frac{\pi}{4}}$ and $\bar{\lambda} = e^{-i\frac{\pi}{4}}$.

Consider a checkerboard tiling $\mathbb{C}^\delta$ of $\mathbb{R}^2$ with 
squares, each square has side $\delta$ and is centered at a lattice point of $$\{(\frac{\delta n}{\sqrt{2}},\frac{\delta m}{\sqrt{2}})\, | \,n, m \in \mathbb{Z}; n+m \in 2\mathbb{Z}\}$$
(see~Fig.~\ref{rotation}). The pair $(n,m)$ is called the coordinates of a point on this lattice. Let $\om^\delta$ be a simply connected discrete domain composed of a finite number of 
squares of $\mathbb{C}^\delta$ bounded by a disjoint simple closed lattice path.   Let $\V{}^\delta$ be the vertex set of $\om^\delta$. 
We will denote by $\bb{}^\delta$ the set of black squares and by $\ww{}^\delta$ the set of white squares of $\om^\delta$. So, $\om^\delta = \bb{}^\delta \sqcup \ww{}^\delta$. Let the coordinates of a square be the coordinates of its center. Then we can define the sets $\bb{0}^\delta$ and $\bb{1}^\delta$ of black squares of $\om^\delta$ and the sets $\ww{0}^\delta$ and $\ww{1}^\delta$  of white squares by the following properties:
\begin{enumerate}
\item[($\bb{0}^\delta$)] both coordinates are even and the sum of coordinates is divisible by~$4$;
\item[($\bb{1}^\delta$)] both coordinates are even and the sum of coordinates is not divisible by~$4$;
\item[($\ww{0}^\delta$)] both coordinates are odd and the sum of coordinates is not divisible by~$4$;
\item[($\ww{1}^\delta$)] both coordinates are odd and the sum of coordinates is divisible by~$4$.
\end{enumerate}

Define $\dV^\delta$ to be the set of vertices on the boundary. Let $\dom^\delta$ be the set of faces adjacent to $\om^\delta$ but not in $\om^\delta$. Let $\db ^\delta$ and $\dw^\delta$ be the sets of black and white faces of $\dom^\delta$ correspondingly. Let $\diom^\delta$ be the set of interior faces that have a common edge with boundary of $\om^\delta$. Similarly define sets $\dib^\delta$ and $\diw^\delta$ ($\diom^\delta=\dib^\delta\sqcup\diw^\delta$). Let us denote by $\clom^\delta$ the set $\om^\delta\sqcup\dom^\delta$, define also sets $\clb^\delta$ and $\clw^\delta$, to be exact: $\clb^\delta=\bb{}^\delta\sqcup\db^\delta$ and $\clw^\delta=\ww{}^\delta\sqcup\dw^\delta$.   
In the same way we define sets $\dbb{0,1}^\delta$, 
$\dww{0,1}^\delta$, 
$\dintb{0,1}^\delta$, 
$\dintw{0,1}^\delta$,
$\clbb{0,1}^\delta$ and $\clww{0,1}^\delta$.

\begin{figure}
\begin{center}

\begin{tikzpicture}[x={(0.5cm,0.5cm)}, y={(-0.5cm,0.5cm)}]

\path (11,0) node[name=l1, shape=coordinate]{};
\path (12,0) node[name=l2, shape=coordinate]{};
\path (12,1) node[name=l3, shape=coordinate]{};
\path (11,1) node[name=l4, shape=coordinate]{};
\path[draw,fill=gray!20] (l1)--(l2)--(l3)--(l4)--cycle;
\path ($1/2*(l1)+1/2*(l3)$) node[]{$u_3$};

\path (12,1) node[name=l1, shape=coordinate]{};
\path (13,1) node[name=l2, shape=coordinate]{};
\path (13,2) node[name=l3, shape=coordinate]{};
\path (12,2) node[name=l4, shape=coordinate]{};
\path[draw,fill=gray!20] (l1)--(l2)--(l3)--(l4)--cycle;
\path ($1/2*(l1)+1/2*(l3)$) node[]{$u_2$};

\path (13,0) node[name=l1, shape=coordinate]{};
\path (13,1) node[name=l2, shape=coordinate]{};
\path (14,1) node[name=l3, shape=coordinate]{};
\path (14,0) node[name=l4, shape=coordinate]{};
\path[draw,fill=gray!20] (l1)--(l2)--(l3)--(l4)--cycle;
\path ($1/2*(l1)+1/2*(l3)$) node[]{$u_1$};

\path (12,0) node[name=l1, shape=coordinate]{};
\path (13,0) node[name=l2, shape=coordinate]{};
\path (13,1) node[name=l3, shape=coordinate]{};
\path (12,1) node[name=l4, shape=coordinate]{};
\path[draw] (l1)--(l2)--(l3)--(l4)--cycle;
\path ($1/2*(l1)+1/2*(l3)$) node[]{$\bf\www$};

\path (12.9,-0.1) node[anchor= west]{$z_+$};
\path (12.1,-0.1) node[anchor= north]{$z_\flat$};
\path (12,1) node[anchor= east]{$z_-$};
\path (13,1) node[anchor= south]{$z_\sharp$};


\path (16,5) node[name=l1, shape=coordinate]{};
\path (17,5) node[name=l2, shape=coordinate]{};
\path (17,6) node[name=l3, shape=coordinate]{};
\path (16,6) node[name=l4, shape=coordinate]{};
\path[draw,fill=gray!20] (l1)--(l2)--(l3)--(l4)--cycle;
\path ($1/2*(l1)+1/2*(l3)$) node[]{$\bf\nbvupugl{s}$};

\path (4,13) node[name=l1, shape=coordinate]{};
\path (5,13) node[name=l2, shape=coordinate]{};
\path (5,12) node[name=l3, shape=coordinate]{};
\path (4,12) node[name=l4, shape=coordinate]{};
\path[draw] (l1)--(l2)--(l3)--(l4)--cycle;
\path ($1/2*(l1)+1/2*(l3)$) node[]{$\bf\nwvupugl{k}$};

\path (1,9) node[name=u1, shape=coordinate]{};
\path (2,9) node[name=u2, shape=coordinate]{};
\path (2,8) node[name=u3, shape=coordinate]{};
\path (1,8) node[name=u4, shape=coordinate]{};
\path[draw,fill=gray!20] (u1)--(u2)--(u3)--(u4)--cycle;
\path ($1/2*(u1)+1/2*(u3)$) node[]{$\bf\bbb$};

\path (12,6) node[name=l1, shape=coordinate]{};
\path (11,6) node[name=l2, shape=coordinate]{};
\path (11,5) node[name=l3, shape=coordinate]{};
\path (12,5) node[name=l4, shape=coordinate]{};
\path[draw] (l1)--(l2)--(l3)--(l4)--cycle;
\path ($1/2*(l1)+1/2*(l3)$) node[]{$\bf\nwvpugl{k}$};

\path (4,10) node[name=u1, shape=coordinate]{};
\path (5,10) node[name=u2, shape=coordinate]{};
\path (5,9) node[name=u3, shape=coordinate]{};
\path (4,9) node[name=u4, shape=coordinate]{};
\path[draw,fill=gray!20] (u1)--(u2)--(u3)--(u4)--cycle;
\path ($1/2*(u1)+1/2*(u3)$) node[]{$\bf\nbvpugl{s}$};


\path (15,2) node[name=u1, shape=coordinate]{};
\path (16,2) node[name=u2, shape=coordinate]{};
\path (16,3) node[name=u3, shape=coordinate]{};
\path (15,3) node[name=u4, shape=coordinate]{};
\path[draw,fill=gray!20] (u1)--(u2)--(u3)--(u4)--cycle;
\path (15.45, 2.45) node[]{$u$};

\path (15,3) node[name=u1, shape=coordinate]{};
\path (16,3) node[name=u2, shape=coordinate]{};
\path (16,4) node[name=u3, shape=coordinate]{};
\path (15,4) node[name=u4, shape=coordinate]{};
\path[draw] (u1)--(u2)--(u3)--(u4)--cycle;
\path ($1/2*(u1)+1/2*(u3)$) node[]{$v$};

\draw[line width=1pt][->] (u1) --(15.9,3);

\path[draw, fill=white] (u1) circle[radius=0.05cm];
\path[draw, fill=white] (u2) circle[radius=0.05cm];

\path (u1) node[anchor=north]{$z$};
\path (16.07,3) node[anchor=south]{$z'$};


\path (5,13) node[name=u1, shape=coordinate]{};
\path (6,13) node[name=u2, shape=coordinate]{};
\path (6,12) node[name=u3, shape=coordinate]{};
\path (5,12) node[name=u4, shape=coordinate]{};
\path[draw,fill=gray!20] (u1)--(u2)--(u3)--(u4)--cycle;
\path ($1/2*(u1)+1/2*(u3)$) node[]{$\bb{0}$};

\path (7,13) node[name=u1, shape=coordinate]{};
\path (8,13) node[name=u2, shape=coordinate]{};
\path (8,12) node[name=u3, shape=coordinate]{};
\path (7,12) node[name=u4, shape=coordinate]{};
\path[draw,fill=gray!20] (u1)--(u2)--(u3)--(u4)--cycle;
\path ($1/2*(u1)+1/2*(u3)$) node[]{$\bb{0}$};

\path[draw](6,12)--(7,12);

\path (8,14) node[name=u1, shape=coordinate]{};
\path (9,14) node[name=u2, shape=coordinate]{};
\path (9,13) node[name=u3, shape=coordinate]{};
\path (8,13) node[name=u4, shape=coordinate]{};
\path[draw,fill=gray!20] (u1)--(u2)--(u3)--(u4)--cycle;

\path (6,14) node[name=u1, shape=coordinate]{};
\path (7,14) node[name=u2, shape=coordinate]{};
\path (7,13) node[name=u3, shape=coordinate]{};
\path (6,13) node[name=u4, shape=coordinate]{};
\path[draw,fill=gray!20] (u1)--(u2)--(u3)--(u4)--cycle;

\path (4,14) node[name=u1, shape=coordinate]{};
\path (5,14) node[name=u2, shape=coordinate]{};
\path (5,13) node[name=u3, shape=coordinate]{};
\path (4,13) node[name=u4, shape=coordinate]{};
\path[draw,fill=gray!20] (u1)--(u2)--(u3)--(u4)--cycle;

\path (6.5, 12.5) node[]{$\ww{0}$};
\path (8.5, 12.5) node[]{$\ww{0}$};
\path[draw](8,12)--(9,12)--(9,13);

\path[draw](5,14)--(6,14);
\path[draw](7,14)--(8,14);

\path (4.5, 13.5) node[]{$0$};
\path (5.5, 13.5) node[]{$0$};
\path (6.5, 13.5) node[]{$0$};
\path (7.5, 13.5) node[]{$0$};
\path (8.5, 13.5) node[]{$0$};

\path (4,12) node[name=u1, shape=coordinate]{};
\path (5,12) node[name=u2, shape=coordinate]{};
\path (5,11) node[name=u3, shape=coordinate]{};
\path (4,11) node[name=u4, shape=coordinate]{};
\path[draw,fill=gray!20] (u1)--(u2)--(u3)--(u4)--cycle;
\path ($1/2*(u1)+1/2*(u3)$) node[]{$\bb{1}$};

\path[draw](5,10)--(5,11);

\path (3,13) node[name=u1, shape=coordinate]{};
\path (4,13) node[name=u2, shape=coordinate]{};
\path (4,12) node[name=u3, shape=coordinate]{};
\path (3,12) node[name=u4, shape=coordinate]{};
\path[draw,fill=gray!20] (u1)--(u2)--(u3)--(u4)--cycle;

\path (3,11) node[name=u1, shape=coordinate]{};
\path (4,11) node[name=u2, shape=coordinate]{};
\path (4,10) node[name=u3, shape=coordinate]{};
\path (3,10) node[name=u4, shape=coordinate]{};
\path[draw,fill=gray!20] (u1)--(u2)--(u3)--(u4)--cycle;

\path (4.5, 10.5) node[]{$\ww{0}$};

\path[draw](3,11)--(3,12);

\path (3.5, 12.5) node[]{$0$};
\path (3.5, 11.5) node[]{$0$};
\path (3.5, 10.5) node[]{$0$};

\path (2,10) node[name=u1, shape=coordinate]{};
\path (3,10) node[name=u2, shape=coordinate]{};
\path (3,9) node[name=u3, shape=coordinate]{};
\path (2,9) node[name=u4, shape=coordinate]{};
\path[draw,fill=gray!20] (u1)--(u2)--(u3)--(u4)--cycle;
\path ($1/2*(u1)+1/2*(u3)$) node[]{$\bb{1}$};

\path[draw](3,9)--(4,9);

\path (1,11) node[name=u1, shape=coordinate]{};
\path (2,11) node[name=u2, shape=coordinate]{};
\path (2,10) node[name=u3, shape=coordinate]{};
\path (1,10) node[name=u4, shape=coordinate]{};
\path[draw,fill=gray!20] (u1)--(u2)--(u3)--(u4)--cycle;

\path (3.5, 9.5) node[]{$\ww{1}$};
\path (1.5, 9.5) node[]{$\ww{1}$};
\path[draw](1,10)--(1,9)--(2,9);

\path[draw](2,11)--(3,11);

\path (1.5, 10.5) node[]{$0$};
\path (2.5, 10.5) node[]{$0$};


\path (7,5) node[name=u1, shape=coordinate]{};
\path (6,5) node[name=u2, shape=coordinate]{};
\path (6,6) node[name=u3, shape=coordinate]{};
\path (7,6) node[name=u4, shape=coordinate]{};
\path[draw,fill=gray!20] (u1)--(u2)--(u3)--(u4)--cycle;
\path ($1/2*(u1)+1/2*(u3)$) node[]{$\bb{1}$};

\path (8,5) node[name=u1, shape=coordinate]{};
\path (7,5) node[name=u2, shape=coordinate]{};
\path (7,6) node[name=u3, shape=coordinate]{};
\path (8,6) node[name=u4, shape=coordinate]{};
\path[draw] (u1)--(u2)--(u3)--(u4)--cycle;
\path (7.45, 5.45) node[]{$\ww{1}$};

\path (9,5) node[name=u1, shape=coordinate]{};
\path (8,5) node[name=u2, shape=coordinate]{};
\path (8,6) node[name=u3, shape=coordinate]{};
\path (9,6) node[name=u4, shape=coordinate]{};
\path[draw,fill=gray!20] (u1)--(u2)--(u3)--(u4)--cycle;
\path ($1/2*(u1)+1/2*(u3)$) node[]{$\bb{1}$};

\path (7,6) node[name=u1, shape=coordinate]{};
\path (6,6) node[name=u2, shape=coordinate]{};
\path (6,7) node[name=u3, shape=coordinate]{};
\path (7,7) node[name=u4, shape=coordinate]{};
\path[draw] (u1)--(u2)--(u3)--(u4)--cycle;
\path ($1/2*(u1)+1/2*(u3)$) node[]{$\ww{0}$};

\path (8,6) node[name=u1, shape=coordinate]{};
\path (7,6) node[name=u2, shape=coordinate]{};
\path (7,7) node[name=u3, shape=coordinate]{};
\path (8,7) node[name=u4, shape=coordinate]{};
\path[draw,fill=gray!20] (u1)--(u2)--(u3)--(u4)--cycle;
\path ($1/2*(u1)+1/2*(u3)$) node[]{$\bb{0}$};

\path (9,6) node[name=u1, shape=coordinate]{};
\path (8,6) node[name=u2, shape=coordinate]{};
\path (8,7) node[name=u3, shape=coordinate]{};
\path (9,7) node[name=u4, shape=coordinate]{};
\path[draw] (u1)--(u2)--(u3)--(u4)--cycle;
\path (8.4,6.5) node[]{$\ww{0}$};

\path (10,6) node[name=u1, shape=coordinate]{};
\path (9,6) node[name=u2, shape=coordinate]{};
\path (9,7) node[name=u3, shape=coordinate]{};
\path (10,7) node[name=u4, shape=coordinate]{};
\path[draw,fill=gray!20] (u1)--(u2)--(u3)--(u4)--cycle;
\path ($1/2*(u1)+1/2*(u3)$) node[]{$\bb{0}$};

\path (7,7) node[name=u1, shape=coordinate]{};
\path (6,7) node[name=u2, shape=coordinate]{};
\path (6,8) node[name=u3, shape=coordinate]{};
\path (7,8) node[name=u4, shape=coordinate]{};
\path[draw,fill=gray!20] (u1)--(u2)--(u3)--(u4)--cycle;
\path (6.4,7.5) node[]{$\bb{1}$};
\path (8,7) node[name=u1, shape=coordinate]{};
\path (7,7) node[name=u2, shape=coordinate]{};
\path (7,8) node[name=u3, shape=coordinate]{};
\path (8,8) node[name=u4, shape=coordinate]{};
\path[draw] (u1)--(u2)--(u3)--(u4)--cycle;
\path ($1/2*(u1)+1/2*(u3)$) node[]{$\ww{1}$};

\path (9,7) node[name=u1, shape=coordinate]{};
\path (8,7) node[name=u2, shape=coordinate]{};
\path (8,8) node[name=u3, shape=coordinate]{};
\path (9,8) node[name=u4, shape=coordinate]{};
\path[draw,fill=gray!20] (u1)--(u2)--(u3)--(u4)--cycle;
\path ($1/2*(u1)+1/2*(u3)$) node[]{$\bb{1}$};

\path (10,7) node[name=u1, shape=coordinate]{};
\path (9,7) node[name=u2, shape=coordinate]{};
\path (9,8) node[name=u3, shape=coordinate]{};
\path (10,8) node[name=u4, shape=coordinate]{};
\path[draw] (u1)--(u2)--(u3)--(u4)--cycle;
\path ($1/2*(u1)+1/2*(u3)$) node[]{$\ww{1}$};

\path (8,8) node[name=u1, shape=coordinate]{};
\path (7,8) node[name=u2, shape=coordinate]{};
\path (7,9) node[name=u3, shape=coordinate]{};
\path (8,9) node[name=u4, shape=coordinate]{};
\path[draw,fill=gray!20] (u1)--(u2)--(u3)--(u4)--cycle;
\path ($1/2*(u1)+1/2*(u3)$) node[]{$\bb{0}$};

\path (9,8) node[name=u1, shape=coordinate]{};
\path (8,8) node[name=u2, shape=coordinate]{};
\path (8,9) node[name=u3, shape=coordinate]{};
\path (9,9) node[name=u4, shape=coordinate]{};
\path[draw] (u1)--(u2)--(u3)--(u4)--cycle;
\path (8.45,8.45) node[]{$\ww{0}$};

\path (10,8) node[name=u1, shape=coordinate]{};
\path (9,8) node[name=u2, shape=coordinate]{};
\path (9,9) node[name=u3, shape=coordinate]{};
\path (10,9) node[name=u4, shape=coordinate]{};
\path[draw,fill=gray!20] (u1)--(u2)--(u3)--(u4)--cycle;
\path ($1/2*(u1)+1/2*(u3)$) node[]{$\bb{0}$};

\path (8,9) node[name=u1, shape=coordinate]{};
\path (7,9) node[name=u2, shape=coordinate]{};
\path (7,10) node[name=u3, shape=coordinate]{};
\path (8,10) node[name=u4, shape=coordinate]{};
\path[draw] (u1)--(u2)--(u3)--(u4)--cycle;
\path ($1/2*(u1)+1/2*(u3)$) node[]{$\ww{1}$};

\path (9,9) node[name=u1, shape=coordinate]{};
\path (8,9) node[name=u2, shape=coordinate]{};
\path (8,10) node[name=u3, shape=coordinate]{};
\path (9,10) node[name=u4, shape=coordinate]{};
\path[draw,fill=gray!20] (u1)--(u2)--(u3)--(u4)--cycle;
\path ($1/2*(u1)+1/2*(u3)$) node[]{$\bb{1}$};

\draw[line width=1pt][->] (8,9) --(8.9,9);
\draw[line width=1pt][->] (7,6) --(7.9,6);
\draw[line width=1pt][->] (7,8) --(7,7.1);
\draw[line width=1pt][->] (9,7) --(9,6.1);

\path[draw, fill=white] (7,5) circle[radius=0.05cm];
\path[draw, fill=white] (6,5) circle[radius=0.05cm];
\path[draw, fill=white] (6,6) circle[radius=0.05cm];
\path[draw, fill=white] (7,6) circle[radius=0.05cm];

\path[draw, fill=white] (9,5) circle[radius=0.05cm];
\path[draw, fill=white] (8,5) circle[radius=0.05cm];
\path[draw, fill=white] (8,6) circle[radius=0.05cm];
\path[draw, fill=white] (9,6) circle[radius=0.05cm];

\path[draw, fill=white] (8,6) circle[radius=0.05cm];
\path[draw, fill=white] (7,6) circle[radius=0.05cm];
\path[draw, fill=white] (7,7) circle[radius=0.05cm];
\path[draw, fill=white] (8,7) circle[radius=0.05cm];

\path[draw, fill=white] (10,6) circle[radius=0.05cm];
\path[draw, fill=white] (9,6) circle[radius=0.05cm];
\path[draw, fill=white] (9,7) circle[radius=0.05cm];
\path[draw, fill=white] (10,7) circle[radius=0.05cm];

\path[draw, fill=white] (7,7) circle[radius=0.05cm];
\path[draw, fill=white] (6,7) circle[radius=0.05cm];
\path[draw, fill=white] (6,8) circle[radius=0.05cm];
\path[draw, fill=white] (7,8) circle[radius=0.05cm];

\path[draw, fill=white] (9,7) circle[radius=0.05cm];
\path[draw, fill=white] (8,7) circle[radius=0.05cm];
\path[draw, fill=white] (8,8) circle[radius=0.05cm];
\path[draw, fill=white] (9,8) circle[radius=0.05cm];

\path[draw, fill=white] (8,8) circle[radius=0.05cm];
\path[draw, fill=white] (7,8) circle[radius=0.05cm];
\path[draw, fill=white] (7,9) circle[radius=0.05cm];
\path[draw, fill=white] (8,9) circle[radius=0.05cm];

\path[draw, fill=white] (10,8) circle[radius=0.05cm];
\path[draw, fill=white] (9,8) circle[radius=0.05cm];
\path[draw, fill=white] (9,9) circle[radius=0.05cm];
\path[draw, fill=white] (10,9) circle[radius=0.05cm];

\path[draw, fill=white] (9,9) circle[radius=0.05cm];
\path[draw, fill=white] (8,9) circle[radius=0.05cm];
\path[draw, fill=white] (8,10) circle[radius=0.05cm];
\path[draw, fill=white] (9,10) circle[radius=0.05cm];

\path[draw][line width=1.5pt](13,0)--(17,0)--(17,6)--(12,6)--(12,11)--(9,11)--(9,13)--(4,13)--(4,10)--(1,10)--(1,9);
\path[draw][line width=1.5pt][line width=3pt](1,8)--(1,6)--(4,6)--(4,3)--(7,3)--(7,0)--(12,0);
\path[draw][white](1,8)--(1,6)--(4,6)--(4,3)--(7,3)--(7,0)--(12,0);
\path[draw][line width=1.5pt,dashed] (12,0)--(13,0);
\path[draw][line width=1.5pt,dashed] (1,8)--(1,9);

\path[draw, fill=white] (12,0) circle[radius=0.05cm];
\path[draw, fill=white] (13,0) circle[radius=0.05cm];
\path[draw, fill=white] (12,1) circle[radius=0.05cm];
\path[draw, fill=white] (13,1) circle[radius=0.05cm];
\path[draw, fill=white] (1,8) circle[radius=0.05cm];
\path[draw, fill=white] (1,9) circle[radius=0.05cm];


\path[draw, fill=white] (7,10) circle[radius=0.05cm];

\path[draw, fill=white] (1,10) circle[radius=0.05cm];
\path[draw, fill=white] (2,10) circle[radius=0.05cm];
\path[draw, fill=white] (3,10) circle[radius=0.05cm];
\path[draw, fill=white] (4,10) circle[radius=0.05cm];
\path[draw, fill=white] (4,10) circle[radius=0.05cm];
\path[draw, fill=white] (4,11) circle[radius=0.05cm];
\path[draw, fill=white] (4,12) circle[radius=0.05cm];
\path[draw, fill=white] (4,13) circle[radius=0.05cm];
\path[draw, fill=white] (5,13) circle[radius=0.05cm];
\path[draw, fill=white] (6,13) circle[radius=0.05cm];
\path[draw, fill=white] (7,13) circle[radius=0.05cm];
\path[draw, fill=white] (8,13) circle[radius=0.05cm];
\path[draw, fill=white] (9,13) circle[radius=0.05cm];

\draw (8.8,8.67) node[font=\tiny]{$\bf\delta\lambda$};
\draw (7.8,5.67) node[font=\tiny]{$\bf\delta\lambda$};

\draw (15.8,2.67) node[font=\tiny]{$\bf\delta\lambda$};

\draw (6.75,7.5) node[font=\tiny]{$\delta\bar\lambda$};
\draw (8.75,6.5) node[font=\tiny]{$\delta\bar\lambda$};

\path (14,7) node[]{$\bf{(\www\bbb)}$};
\path (10,-1) node[]{$\bf{(\bbb\www)}$};

\end{tikzpicture}
\caption{
The domain $\om^\delta$, the sets ${\color{gray}\blacklozenge}^\delta$ and $\lozenge^\delta$ of this domain and the set $\mathcal{V}^\delta$ is the vertex set of 
$\om^\delta$, the squares 
$\bbb \in \protect\dintb{0}$ 
and $\www \in \protect\dintw{0}$, 
and the elements of the sets of square corners 
$\{\protect\nbvupugl{s}\}^{m+1}_{s=1}$, $\{\protect\nwvupugl{k}\}^{n+1}_{k=1}$, $\{\protect\nbvpugl{s}\}^{m-1}_{s=1}$ and  $\{\protect\nwvpugl{k}\}^{n-1}_{k=1}$. 
We call a corner of $\om^\delta$ a convex corner if the interior angle is $\pi/2$, and concave if the interior angle is $3\pi/2$. A corner is called white if there is a white square in the corner, and black if there is a black square in this corner.
The discrete holomorphic function $\F^\delta$ is real on~$\protect\clbb{0}^\delta$ and purely imaginary on~$\protect\clbb{1}^\delta$; the discrete holomorphic function $\G^\delta$ belongs to~$\lambda\mathbb{R}$ on $\protect\clww{0}^\delta$ and belongs to~$\bar{\lambda}\mathbb{R}$ on $\protect\clww{1}^\delta$. The discrete primitive $\HH^\delta$ is defined on vertices and is purely real: it is easy to check that in all possible positions (according to the types of the squares) the difference of values of function $\HH^\delta$ at two adjacent vertices is real.  For all $u\in {\db}^\delta$, either $\re[\F^\delta(u)]=0$ or $\im[\F^\delta(u)]=0$; for all $v\in \dw^\delta$, either $\re[\bar{\lambda}\G^\delta(v)]=0$ or $\re[\lambda\G^\delta(v)]=0$. The function $\F^\delta$ (resp., $\G^\delta$) changes boundary conditions only at white (resp., black) corners of $\om^\delta$.
Let $(\bbb\www)$ be a part of the boundary starting at the middle of the boundary side of the square $\bbb$ and going to the middle of the boundary side of square $\www$ in the positive direction. Note that two segments of the boundary $(\bbb\www)$ and $(\www\bbb)$ form the whole boundary.
The function $\HH^\delta$ is a constant on  $(\bbb\www)$ and $(\www\bbb)$. 
The difference of the values on these segments is nonzero: 
$H^\delta(z_\flat)-H^\delta(z_+)=4i\delta^2\G^\delta(\www)[\bar{\partial}^\delta\F^\delta](\www)$.
}
\label{colors}
\label{rotation}
\label{z1z2}
\label{d-b = c-a v_0}
\end{center}
\end{figure}
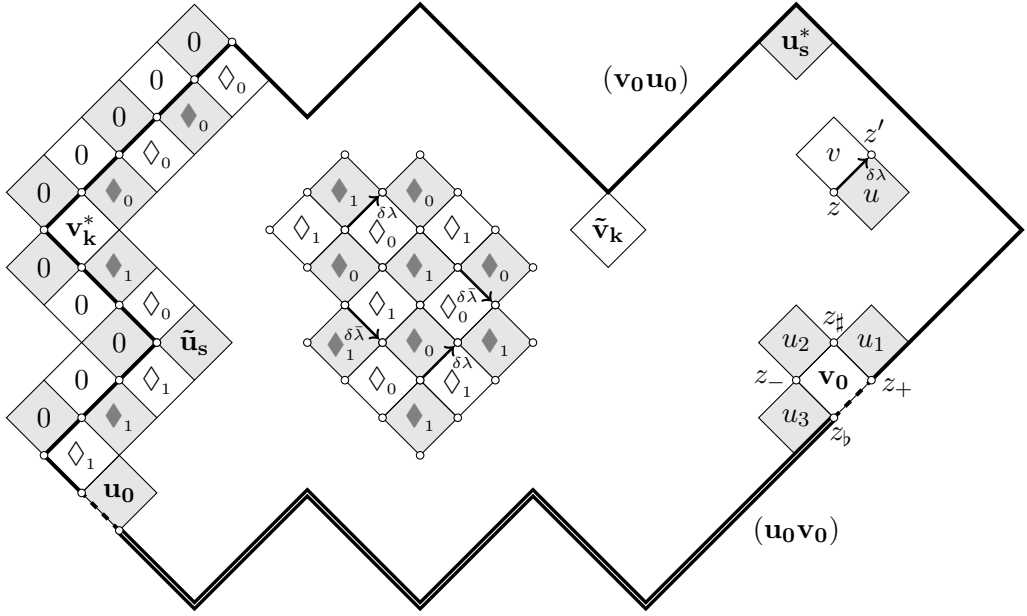

Let $\F^\delta~\colon\clb^\delta~\to~\mathbb{C}$ be a function.  Let us define discrete operators $\partial^\delta$ and $\bar{\partial}^\delta$ by the formulas: 
\[
[\partial^\delta \F^\delta](v)=\frac12\left(\frac{\F^\delta(v+\delta\lambda)-
 \F^\delta(v-\delta\lambda)}{2\delta\lambda}+\frac{\F^\delta(v+\delta\bar{\lambda})-
 \F^\delta(v-\delta\bar{\lambda})}{2\delta\bar{\lambda}}\right), \]

\[ [\bar{\partial}^\delta \F^\delta](v)=\frac12\left(\frac{\F^\delta(v+\delta\lambda)-
 \F^\delta(v-\delta\lambda)}{2\delta\bar{\lambda}}+\frac{\F^\delta(v+\delta\bar{\lambda})-
 \F^\delta(v-\delta\bar{\lambda})}{2\delta\lambda}\right),
\]
  where $v\in\ww{}^\delta$. Note that, If $[\bar{\partial}^\delta \F^\delta](v)=0$, then the two terms involved into the definition of $[\partial \F^\delta]$ are equal to each other.

We can similarly define these operators for a function $\G^\delta\colon\clw^\delta\to\mathbb{C}$.

\begin{definition}
A function $\F^\delta\colon\clb^\delta\to\mathbb{C}$ is called discrete holomorphic in $\om^\delta$ if 
$[\bar{\partial}^\delta\F^\delta](v) = 0$ for all $v\in\ww{}^\delta$. Also, we always assume that $\F^\delta$ is real on~$\clbb{0}^\delta$ and purely imaginary on~$\clbb{1}^\delta$. 

A function $\G^\delta\colon\clw^\delta\to\mathbb{C}$ is called discrete holomorphic in $\om^\delta$ if $[\bar{\partial}^\delta\G^\delta](u) = 0$ for all $u\in\bb{}^\delta$.
Also, we always assume that $\G^\delta$ belongs to~$\lambda\mathbb{R}$ (resp., $\bar{\lambda}\mathbb{R}$) on $\clww{0}^\delta$ (resp., on $\clww{1}^\delta$).
\end{definition}

\begin{remark}
If a function $\F^\delta\colon\clb^\delta\to\mathbb{C}$ is discrete holomorphic in~$\om^\delta$ then
\mbox{$[i\cdot\partial^\delta \F^\delta]\colon\ww{}^\delta\to\mathbb{C}$} is a discrete holomorphic function in~$\om^\delta\smallsetminus\diom^\delta$. 
Similarly, if $\G^\delta\colon\clw^\delta\to\mathbb{C}$ is discrete holomorphic in~$\om^\delta$ then $\partial^\delta \G^\delta\colon\bb{}^\delta\to\mathbb{C}$ is discrete holomorphic in~$\om^\delta\smallsetminus\diom^\delta$.
\end{remark}
 
 Define the discrete Laplacian of $\F^\delta$ by
$$\Delta^\delta\F^\delta(u)=\frac {\F^\delta(u+2\delta\lambda)+\F^\delta(u+2\delta\bar{\lambda})+\F^\delta(u-2\delta\lambda)+\F^\delta(u-2\delta\bar{\lambda})-4\F^\delta(u)}{4\delta^2},$$ where $u\in\bb{}^\delta$. Note that $\Delta^\delta \F^\delta(u)=4[\partial^\delta\bar{\partial}^\delta\F^\delta](u)=4[\bar{\partial}^\delta\partial^\delta\F^\delta](u).$

A function  $\F^\delta\colon\clb^\delta\to\mathbb{C}$ (resp., $\G^\delta\colon\clw^\delta\to\mathbb{C}$) is called a discrete harmonic function in $\om^\delta$ if it satisfies $\Delta^\delta \F^\delta(u)=0$ for all $u\in\bb{}^\delta$ (resp., $\Delta^\delta \G^\delta(v)=0$ for all $v\in\ww{}^\delta$).

It is easy to see that discrete harmonic functions satisfy the maximum principle:
$$\max_{u\in\om^\delta}\F^\delta(u)=\max_{u\in\dom^\delta}\F^\delta(u).$$

\subsection{The primitive of the product of two discrete holomorphic functions}\label{2.2}
 In this section we will define the discrete primitive of the product of two discrete holomorphic functions. This definition is close to the definition of the discrete primitive of the square of the s-holomorphic function~\cite{CS, Stas}. Also, there is a straightforward generalization of this construction on isoradial graphs, see Appendix~\ref{B}.

\begin{definition}\label{defH}
Let  $\F^\delta\colon\clb^\delta\to\mathbb{C}$ and $\G^\delta\colon\clw^\delta\to\mathbb{C}$ be discrete holomorphic functions.
Let us define a discrete primitive $\HH^\delta\colon \V{}^\delta \to \mathbb{R}$ by the equality
\begin{equation} \label{def H}
\HH^\delta(z'\,)-\HH^\delta(z)=(z' - z)\F^\delta(u)\G^\delta(v),
\end{equation}
where $u$, $v$ are adjacent black and white squares (correspondingly); and $z$, $z'$ are their common vertices, see Fig.~\ref{colors}.
\end{definition}

\begin{remark} It is easy to see that, if $\om^\delta$ is simply connected, then $\HH^\delta$ is well defined (see Fig.~\ref{lapH}). %
\end{remark}

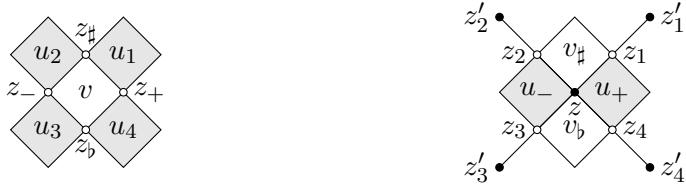
\begin{figure}\begin{center}
\begin{tikzpicture}[x={(0.5cm,0.5cm)}, y={(-0.5cm,0.5cm)}]
\begin{scope}
\path (-2,0) node[name=l1, shape=coordinate]{};
\path (-1,0) node[name=l2, shape=coordinate]{};
\path (-1,1) node[name=l3, shape=coordinate]{};
\path (-2,1) node[name=l4, shape=coordinate]{};
\path[draw, fill=gray!20] (l1)--(l2)--(l3)--(l4)--cycle;
\path ($1/2*(l1)+1/2*(l3)$) node[]{$u_3$};

\path (-1,1) node[name=u1, shape=coordinate]{};
\path (0,1) node[name=u2, shape=coordinate]{};
\path (0,2) node[name=u3, shape=coordinate]{};
\path (-1,2) node[name=u4, shape=coordinate]{};
\path[draw,fill=gray!20] (u1)--(u2)--(u3)--(u4)--cycle;
\path ($1/2*(u1)+1/2*(u3)$) node[]{$u_2$};

\path (0,0) node[name=r1, shape=coordinate]{};
\path (1,0) node[name=r2, shape=coordinate]{};
\path (1,1) node[name=r3, shape=coordinate]{};
\path (0,1) node[name=r4, shape=coordinate]{};
\path[draw, fill=gray!20] (r1)--(r2)--(r3)--(r4)--cycle;
\path ($1/2*(r1)+1/2*(r3)$) node[]{$u_1$};

\path (-1,-1) node[name=d1, shape=coordinate]{};
\path (0,-1) node[name=d2, shape=coordinate]{};
\path (0,0) node[name=d3, shape=coordinate]{};
\path (-1,0) node[name=d4, shape=coordinate]{};
\path[draw, fill=gray!20] (d1)--(d2)--(d3)--(d4)--cycle;
\path ($1/2*(d1)+1/2*(d3)$) node[]{$u_4$};

\path ($1/2*(u2)+1/2*(d4)$) node[]{$v$};

\path[draw, fill=white] (l3) circle[radius=0.05cm];
\path ($(l3)$) node[anchor= east]{$z_-$};
\path[draw, fill=white] (u2) circle[radius=0.05cm];
\path ($(u2)$) node[anchor= south]{$z_\sharp$};
\path[draw, fill=white] (r1) circle[radius=0.05cm];
\path ($(r1)$) node[anchor= west]{$z_+$};
\path[draw, fill=white] (d4) circle[radius=0.05cm];
\path ($(d4)$) node[anchor= north]{$z_\flat$};
\end{scope}

\begin{scope}[xshift=6cm]
\path (-1,0) node[name=l1, shape=coordinate]{};
\path (0,0) node[name=l2, shape=coordinate]{};
\path (0,1) node[name=l3, shape=coordinate]{};
\path (-1,1) node[name=l4, shape=coordinate]{};
\path[draw, fill=gray!20] (l1)--(l2)--(l3)--(l4)--cycle;
\path ($1/2*(l1)+1/2*(l3)$) node[]{$u_-$};

\path (0,0) node[name=u1, shape=coordinate]{};
\path (1,0) node[name=u2, shape=coordinate]{};
\path (1,1) node[name=u3, shape=coordinate]{};
\path (0,1) node[name=u4, shape=coordinate]{};
\path ($1/2*(u1)+1/2*(u3)$) node[]{$v_\sharp$};
\path[draw] (u1)--(u2)--(u3)--(u4)--cycle;

\path (0,-1) node[name=r1, shape=coordinate]{};
\path (1,-1) node[name=r2, shape=coordinate]{};
\path (1,0) node[name=r3, shape=coordinate]{};
\path (0,0) node[name=r4, shape=coordinate]{};
\path[draw, fill=gray!20] (r1)--(r2)--(r3)--(r4)--cycle;
\path ($1/2*(r1)+1/2*(r3)$) node[]{$u_+$};

\path (-1,-1) node[name=d1, shape=coordinate]{};
\path (0,-1) node[name=d2, shape=coordinate]{};
\path (0,0) node[name=d3, shape=coordinate]{};
\path (-1,0) node[name=d4, shape=coordinate]{};
\path[draw] (d1)--(d2)--(d3)--(d4)--cycle;
\path ($1/2*(d1)+1/2*(d3)$) node[]{$v_{\tiny\flat}$};


\path ($2*(l3)$) node[name=lu, shape=coordinate]{};
\path ($2*(u2)$) node[name=ru, shape=coordinate]{};
\path ($2*(r1)$) node[name=ld, shape=coordinate]{};
\path ($2*(d4)$) node[name=rd, shape=coordinate]{};
\path[draw] (l3)--(lu);
\path[draw] (u2)--(ru);
\path[draw] (r1)--(ld);
\path[draw] (d4)--(rd);

\path[draw, fill=black] (u1) circle[radius=0.05cm];
\path[draw, fill=black] (lu) circle[radius=0.05cm];
\path[draw, fill=black] (ru) circle[radius=0.05cm];
\path[draw, fill=black] (ld) circle[radius=0.05cm];
\path[draw, fill=black] (rd) circle[radius=0.05cm];
\path (l2) node[anchor=north]{$z$};
\path (ru) node[anchor=west]{${z'_{1}}$};
\path (lu) node[anchor=east]{$z'_{2}$};
\path (rd) node[anchor=east]{$z'_{3}$};
\path (ld) node[anchor=west]{$z'_{4}$};

\path[draw, fill=white] (l3) circle[radius=0.05cm];
\path[draw, fill=white] (r1) circle[radius=0.05cm];
\path[draw, fill=white] (u2) circle[radius=0.05cm];
\path[draw, fill=white] (d4) circle[radius=0.05cm];

\path (l3) node[anchor= east]{$z_2$};
\path (u2) node[anchor= west]{$z_1$};
\path (r1) node[anchor= west]{$z_4$};
\path (d4) node[anchor= east]{$z_3$};
\end{scope}
\end{tikzpicture}\end{center}\caption{Left: The discrete holomorphic condition $[\bar{\partial}^\delta\F^\delta](v)=0$ guarantees that $(\HH^\delta(z_\flat) - \HH^\delta(z_-)) + (\HH^\delta(z_+) - \HH^\delta(z_\flat)) + (\HH^\delta(z_\sharp) - \HH^\delta(z_+)) + (\HH^\delta(z_-) - \HH^\delta(z_\sharp)) = 0.$
Right: The discrete leap-frog Laplacian of $\HH^\delta$ is defined by ($\ref{lfH}$). 
The function $\HH^\delta$ has no saddle points: a value at an interior vertex cannot be strictly greater than values at two of its neighbouring vertices and strictly smaller than values at two other neighbouring vertices at the same time.}
\label{d-b = c-a}
\label{lapH}\label{min-max}
\end{figure}

Let us define the discrete leap-frog Laplacian of $\HH^\delta$ at $z\in\intV^\delta$ by 
\begin{equation} \label{lfH}
[\Delta^\delta\HH^\delta](z)=\frac{1}{4\delta^2}\sum_{z'_s\sim z}(\HH^\delta(z'_s)-\HH^\delta(z)),
\end{equation} 
where $s \in \lbrace 1, 2, 3, 4\rbrace$, and $z'_s$ are defined as shown in Fig.~\ref{lapH}.

\begin{prop}\label{formula lap}
Let $\Aa$, $\cc$, $\pp$, $\qq$, $z$ be as shown in Fig.~\ref{lapH}, then

\begin{equation} \label{Delta H}
\begin{split}
[\Delta^\delta\HH^\delta](z) =\delta\cdot(
&{\lambda}[\partial^\delta\F^\delta](\pp)[\partial^\delta\G^\delta](\Aa)-\bar{\lambda}[\partial^\delta\F^\delta](\pp)[\partial^\delta\G^\delta](\cc)\\
+&\bar{\lambda}[\partial^\delta\F^\delta](\qq)[\partial^\delta\G^\delta](\Aa)-
{\lambda}[\partial^\delta\F^\delta](\qq)[\partial^\delta\G^\delta](\cc)).
\end{split}
\end{equation}
\end{prop}

\begin{proof}
Note that
\begin{align*}
4\delta^2[\Delta^\delta\HH^\delta](z) =
&{\delta\lambda}[\F^\delta(\Aa)+2\delta{\lambda}[\partial^\delta\F^\delta](\pp)]\cdot[\G^\delta(\qq)+2\delta{\lambda}[\partial^\delta\G^\delta](\cc)]+\delta{\lambda}\F^\delta(\cc)\G^\delta(\pp)\\
-&\delta\bar{\lambda}[\F^\delta(\cc)-2\delta\bar{\lambda}[\partial^\delta\F^\delta](\pp)]\cdot
[\G^\delta(\qq)-2\delta\bar{\lambda}[\partial^\delta\G^\delta](\Aa)]-\delta\bar{\lambda}\F^\delta(\Aa)\G^\delta(\pp)\\
-&\delta{\lambda}[\F^\delta(\cc)-2\delta{\lambda}[\partial^\delta\F^\delta](\qq)]\cdot[\G^\delta(\pp)-2\delta{\lambda}[\partial^\delta\G^\delta](\Aa)]-\delta{\lambda}\F^\delta(\Aa)\G^\delta(\qq)\\
+&\delta\bar{\lambda}[\F^\delta(\Aa)+2\delta\bar{\lambda}[\partial^\delta\F^\delta](\qq)]\cdot
[\G^\delta(\pp)+2\delta\bar{\lambda}[\partial^\delta\G^\delta](\cc)]+\delta\bar{\lambda}\F^\delta(\cc)\G^\delta(\qq).
\end{align*}
One can rewrite the above formula  in the following form
\begin{align*}
&\F^\delta(\Aa)\cdot\underbrace{[\delta{\lambda}\G^\delta(\qq)+2\delta^2\lambda^2[\partial^\delta\G^\delta](\cc)
-\delta\bar{\lambda}\G^\delta(\pp)-\delta{\lambda}\G^\delta(\qq)+\delta\bar{\lambda}\G^\delta(\pp)+2\delta^2\bar{\lambda}^2[\partial^\delta\G^\delta](\cc)]}_{=0}\\
+&\F^\delta(\cc)\cdot\underbrace{[\delta{\lambda}\G^\delta(\pp)+2\delta^2{\lambda}^2[\partial^\delta\G^\delta](\Aa)
-\delta\bar{\lambda}\G^\delta(\qq)-\delta{\lambda}\G^\delta(\pp)+\delta\bar{\lambda}\G^\delta(\qq)+2\delta^2\bar{\lambda}^2[\partial^\delta\G^\delta](\Aa)]}_{=0}\\
+&\G^\delta(\qq)\cdot\underbrace{[2\delta^2\bar\lambda^2[\partial^\delta\F^\delta](\pp)+2\delta^2{\lambda}^2[\partial^\delta\F^\delta](\pp)]}_{=0}+
\G^\delta(\pp)\cdot\underbrace{[2\delta^2\bar\lambda^2[\partial^\delta\F^\delta](\qq)+2\delta^2{\lambda}^2[\partial^\delta\F^\delta](\qq)]}_{=0}\\
+&4\cdot[\delta^3{\lambda}^3[\partial^\delta\F^\delta](\pp)[\partial^\delta\G^\delta](\cc)-\delta^3\bar{\lambda}^3[\partial^\delta\F^\delta](\pp)[\partial^\delta\G^\delta](\Aa)\\
- &\delta^3{\lambda}^3[\partial^\delta\F^\delta](\qq)[\partial^\delta\G^\delta](\Aa)
+\delta^3\bar{\lambda}^3[\partial^\delta\F^\delta](\qq)[\partial^\delta\G^\delta](\cc)].
\end{align*}

Finally, note that $\bar\lambda^3=-{\lambda}$ and ${\lambda}^3=-\bar\lambda.$
\end{proof}

\begin{prop}\label{sedla}
The function $\HH^\delta$ has no local maxima or minima. Moreover, a value at an interior vertex cannot be strictly greater than values at two of its neighbouring vertices and strictly smaller than values at two other neighbouring vertices at the same time.
\end{prop}

\begin{proof} 
It is enough to show that the product of all the differences is non-positive (see~Fig.~\ref{min-max}):
\begin{align*}
&(\HH^\delta(z)-\HH^\delta(z_1))\cdot
(\HH^\delta(z)-\HH^\delta(z_2))\cdot
(\HH^\delta(z)-\HH^\delta(z_3))\cdot
(\HH^\delta(z)-\HH^\delta(z_4))\\
=&(-\delta{\lambda})\F^\delta(u_+)\G^\delta(v_\sharp)
\cdot\delta\bar{\lambda}\G^\delta(v_\sharp)\F^\delta(u_-)
\cdot\delta{\lambda}\F^\delta(u_-)\G^\delta(v_\flat)
\cdot(-\delta\bar{\lambda})\G^\delta(v_\flat)\F^\delta(u_+)\\
=&\delta^4\cdot(\F^\delta(u_+)\cdot\G^\delta(v_\sharp)\cdot\F^\delta(u_-)\cdot\G^\delta(v_\flat))^2\leq 0,
\end{align*}  
since $\F^\delta(u_+)\cdot\F^\delta(u_-)\in i\mathbb{R}$ and $\G^\delta(v_\sharp)\cdot\G^\delta(v_\flat)\in \mathbb{R}$. %
\end{proof}

\begin{remark}
\begin{enumerate}
\item[1.] The function $\HH^\delta$ satisfies the maximum principle:
$$\max_{z\in\V{}^\delta}\HH^\delta(z)=\max_{z\in\dV{}^\delta}\HH^\delta(z).$$ 

\item[2.]Also, it is easy to see that $\HH^\delta$ satisfies the following non-linear equation:
\[(\HH^\delta(z)-\HH^\delta(z_1))\cdot
(\HH^\delta(z)-\HH^\delta(z_3)) +
(\HH^\delta(z)-\HH^\delta(z_2))\cdot
(\HH^\delta(z)-\HH^\delta(z_4))=0,\]
where $z$, $z_1$, $z_2$, $z_3$, $z_4$ are defined as shown in Fig.~\ref{min-max}.
\end{enumerate}
\end{remark}

It is worth noting that Definition~\ref{defH} coincides with the definition of a primitive of the product of two s-holomorphic functions used in ~\cite{Stas07}. To see this let us divide the vertex set $\mathcal{V}$ into two sets $\mathcal{V}_\circ$ and $\mathcal{V}_\bullet$ as it shown on Fig.~\ref{s-h}. On the set $\mathcal{V}_\bullet$ the function $H_{\operatorname{s-hol}}$ defined below as a discrete integral of the product of two discrete s-holomorphic functions coincides with the function $\HH$ defined above.

Let $\F\colon \bar{\bb{}}\to\mathbb{C}$ and $\G\colon\bar{\ww{}}\to\mathbb{C}$ be discrete holomorphic functions defined above. Let $F_{\operatorname{s-hol}}$ be a function defined as follows:
\[
\begin{cases} 
\begin{array}{llll}
F_{\operatorname{s-hol}}(u)=\F(u)  \quad & \operatorname{if}\, u\in \bb{}; 
& F_{\operatorname{s-hol}}(\vl)=\frac{\lambda}{\sqrt{2}}\cdot(\F(\uR)-i\F(\uI)) \quad & \operatorname{if}\, \vl\in\ww{0}; \\
F_{\operatorname{s-hol}}(z)=\F(\uR)+\F(\uI)  \quad & \operatorname{if}\, z\in \mathcal{V}_\circ; 
&  F_{\operatorname{s-hol}}(\vlbr)=\frac{\bar\lambda}{\sqrt{2}}\cdot(\F(\uR)+i\F(\uI)) \quad & \operatorname{if}\, \vlbr\in\ww{1},\\
\end{array}
\end{cases}
\] where $z\in \mathcal{V}_\circ$ and $\uI, \vlbr, \uR, \vl$ are adjacent to the vertex $z$ squares (see Fig.~\ref{s-h}).

Let us similarly define a function  $G_{\operatorname{s-hol}}$:
\[
\begin{cases} 
\begin{array}{llll}
G_{\operatorname{s-hol}}(v)=\G(v)  \quad & \operatorname{if}\, v\in \ww{}; & G_{\operatorname{s-hol}}(\uR)=\left(\frac{\bar\lambda\G(\vl)+\lambda\G(\vlbr)}{\sqrt{2}}\right) \quad & \operatorname{if}\, \uR\in\bb{0};\\
G_{\operatorname{s-hol}}(z)=\F(\vl)+\F(\vlbr)  \quad & \operatorname{if}\, z\in \mathcal{V}_\circ; & G_{\operatorname{s-hol}}(\uI)=i\cdot\left(\frac{\bar\lambda\G(\vl)-\lambda\G(\vlbr)}{\sqrt{2}}\right) \quad & \operatorname{if}\, \uI\in\bb{1}.\\
\end{array}
\end{cases}
\]

Note that functions $\F|_{\mathcal{V}_\circ}$ and $\G|_{\mathcal{V}_\circ}$ are s-holomorphic functions on $\mathcal{V}_\circ$, i.e. for each pair of white vertices $z_1^\circ$, $z_2^\circ$ of the same square $a$
\[\mathrm{Proj}_{\tau(a)} [F(z_1)] = \mathrm{Proj}_{\tau(a)} [F(z_2)],\]
where ${\mathrm{Proj}}_{\tau(a)}[z]=\tau(a)\cdot\re\left[z\cdot\overline{\tau(a)}\right]$ and $\tau(a)$ is $1$, $i$, $\lambda$ or $\bar{\lambda}$ if the square $a$ is a square of type $\bb{0}$, $\bb{1}$, $\ww{0}$ or $\ww{1}$ correspondingly.%

\begin{figure}
\begin{center}

\begin{tikzpicture}[x={(0.5cm,0.5cm)}, y={(-0.5cm,0.5cm)}]
\begin{scope}
\path (-1,0) node[name=l1, shape=coordinate]{};
\path (0,0) node[name=l2, shape=coordinate]{};
\path (0,1) node[name=l3, shape=coordinate]{};
\path (-1,1) node[name=l4, shape=coordinate]{};
\path[draw, fill=gray!20] (l1)--(l2)--(l3)--(l4)--cycle;
\path ($1/2*(l1)+1/2*(l3)$) node[]{$\bb{0}$};

\path (0,1) node[name=u1, shape=coordinate]{};
\path (1,1) node[name=u2, shape=coordinate]{};
\path (1,2) node[name=u3, shape=coordinate]{};
\path (0,2) node[name=u4, shape=coordinate]{};
\path[draw,fill=gray!20] (u1)--(u2)--(u3)--(u4)--cycle;
\path ($1/2*(u1)+1/2*(u3)$) node[]{$\bb{1}$};

\path (1,0) node[name=r1, shape=coordinate]{};
\path (2,0) node[name=r2, shape=coordinate]{};
\path (2,1) node[name=r3, shape=coordinate]{};
\path (1,1) node[name=r4, shape=coordinate]{};
\path[draw, fill=gray!20] (r1)--(r2)--(r3)--(r4)--cycle;
\path ($1/2*(r1)+1/2*(r3)$) node[]{$\bb{0}$};

\path (0,-1) node[name=d1, shape=coordinate]{};
\path (1,-1) node[name=d2, shape=coordinate]{};
\path (1,0) node[name=d3, shape=coordinate]{};
\path (0,0) node[name=d4, shape=coordinate]{};
\path[draw, fill=gray!20] (d1)--(d2)--(d3)--(d4)--cycle;
\path ($1/2*(d1)+1/2*(d3)$) node[]{$\bb{1}$};

\path ($1/2*(u2)+1/2*(d4)$) node[]{$\ww{0}$};

\path[draw] (-1,1)--(-1,2)--(0,2);
\path (-0.5,1.5) node[]{$\ww{1}$};

\path[draw] (-1,0)--(-1,-1)--(0,-1);
\path (-0.5,-0.5) node[]{$\ww{1}$};

\path[draw, fill=gray!20] (2,-1)--(2,0)--(3,0)--(3,-1)--cycle;
\path (2.5,-0.5) node[]{$\bb{1}$};

\path[draw] (2,1)--(3,1)--(3,0);
\path (2.5,0.5) node[]{$\ww{0}$};

\path[draw] (1,-1)--(2,-1);
\path (1.5,-0.5) node[]{$\ww{1}$};
\path[draw, fill=white] (l3) circle[radius=0.05cm];
\path[draw, fill=black] (u2) circle[radius=0.05cm];
\path[draw, fill=white] (r1) circle[radius=0.05cm];
\path[draw, fill=black] (d4) circle[radius=0.05cm];
\path[draw, fill=black] (-1,-1) circle[radius=0.05cm];
\path[draw, fill=white] (-1,0) circle[radius=0.05cm];
\path[draw, fill=black] (-1,1) circle[radius=0.05cm];
\path[draw, fill=white] (-1,2) circle[radius=0.05cm];
\path[draw, fill=black] (0,2) circle[radius=0.05cm];
\path[draw, fill=white] (0,-1) circle[radius=0.05cm];
\path[draw, fill=white] (1,2) circle[radius=0.05cm];
\path[draw, fill=black] (1,-1) circle[radius=0.05cm];

\path[draw, fill=white] (2,1) circle[radius=0.05cm];
\path[draw, fill=black] (2,0) circle[radius=0.05cm];
\path[draw, fill=white] (2,-1) circle[radius=0.05cm];
\path[draw, fill=black] (3,1) circle[radius=0.05cm];
\path[draw, fill=white] (3,0) circle[radius=0.05cm];
\path[draw, fill=black] (3,-1) circle[radius=0.05cm];
\end{scope}

\begin{scope}[xshift=6cm]
\path (-1.3,0) node[name=l1, shape=coordinate]{};
\path (0,0) node[name=l2, shape=coordinate]{};
\path (0,1.3) node[name=l3, shape=coordinate]{};
\path (-1.3,1.3) node[name=l4, shape=coordinate]{};
\path[draw] (l1)--(l2)--(l3)--(l4)--cycle;
\path ($1/2*(l1)+1/2*(l3)$) node[]{$\vlbr$};

\path (0,0) node[name=u1, shape=coordinate]{};
\path (1.3,0) node[name=u2, shape=coordinate]{};
\path (1.3,1.3) node[name=u3, shape=coordinate]{};
\path (0,1.3) node[name=u4, shape=coordinate]{};
\path[draw, fill=gray!20] (u1)--(u2)--(u3)--(u4)--cycle;
\path ($1/2*(u1)+1/2*(u3)$) node[]{$\uI$};

\path (0,-1.3) node[name=r1, shape=coordinate]{};
\path (1.3,-1.3) node[name=r2, shape=coordinate]{};
\path (1.3,0) node[name=r3, shape=coordinate]{};
\path (0,0) node[name=r4, shape=coordinate]{};
\path[draw] (r1)--(r2)--(r3)--(r4)--cycle;
\path ($1/2*(r1)+1/2*(r3)$) node[]{$\vl$};

\path (-1.3,-1.3) node[name=d1, shape=coordinate]{};
\path (0,-1.3) node[name=d2, shape=coordinate]{};
\path (0,0) node[name=d3, shape=coordinate]{};
\path (-1.3,0) node[name=d4, shape=coordinate]{};
\path[draw, fill=gray!20] (d1)--(d2)--(d3)--(d4)--cycle;
\path (-0.7,-0.8) node[]{$\uR$};

\path[draw, fill=white] (u1) circle[radius=0.05cm];
\path[draw, fill=black] (l3) circle[radius=0.05cm];
\path[draw, fill=black] (r1) circle[radius=0.05cm];
\path[draw, fill=black] (u2) circle[radius=0.05cm];
\path[draw, fill=black] (d4) circle[radius=0.05cm];

\path (l2) node[anchor=north]{$z$};
\end{scope}
\end{tikzpicture}\end{center}
\caption{
Left: the set $\mathcal{V}_\circ$ (white vertices), the set  
$\mathcal{V}_\bullet$ (black vertices). 
So, $\mathcal{V} = \mathcal{V}_\circ \sqcup \mathcal{V}_\bullet$. 
Right: adjacent to the vertex $z\in \mathcal{V}_\circ$ squares $\protect\uI, \protect\vlbr, \protect\uR, \protect\vl$. 
S-holomorphic functions defined on $\mathcal{V}_\circ$ and its projections defined 
on $\protect\ww{}\sqcup\protect\bb{}$.
} \label{s-h}
\end{figure}
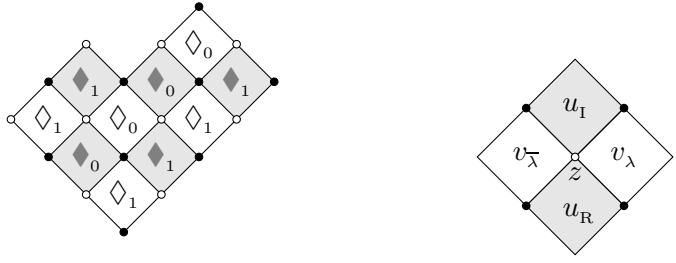

Let $H_{\operatorname{s-hol}}\colon \mathcal{V}_\bullet \to \mathbb{R}$ be a function defined by the equality
\begin{equation*} 
H_{\operatorname{s-hol}}(z_2^\bullet)-H_{\operatorname{s-hol}}(z_1^\bullet)=F_{\operatorname{s-hol}}(a)\cdot G_{\operatorname{s-hol}}(a)\cdot(z_2^\bullet - z_1^\bullet),
\end{equation*}
where $z_1^\bullet$, $z_2^\bullet$ are two black vertices of the same square $a$. It is easy to check that 
\[
H_{\operatorname{s-hol}}(z_2^\bullet)-H_{\operatorname{s-hol}}(z_1^\bullet)=
(H(z_2^\bullet)-H(z^\circ))+(H(z^\circ)-H(z_1^\bullet)),
\]
where $z^\circ$ is one of two white vertices of the square $a$.
Note that the function $H_{\operatorname{s-hol}}(\cdot)$ is defined up to an additive constant. One can choose the additive constant such that the function $H_{\operatorname{s-hol}}$ coincides 
with the function $\HH|_{\mathcal{V}_\bullet}.$

\subsection{The expectation of the double dimer height function}\label{3_3}
In the rest of section~$3$, we will use the square lattice with mesh size $1$ rather than $\delta$. For the simplicity of notations we will not write the index $\delta$. (Later, in section~$4$, we are going to use notations without index for continuous objects.)
We prove that the function $\HH$ defined by formula~$(\ref{def H})$ with an appropriate choice of functions $\F$ and $\G$ described above is the expectation of the height function for double dimers up to a multiplicative constant. 

\begin{lemma}\label{complan}
$1$. Let a domain $\om$ admit a domino tiling. Suppose that a discrete holomorphic function $\F\colon\clb\to\mathbb{C}$ vanishes on~$\db$. Then $\F$ is identically zero.

$2$. Let $\om$ be a domain which contains $m$ white squares and $m+1$ black squares. Let the domain have a domino tiling after removing one black square from $\dib$. Then there exists a nontrivial discrete holomorphic function $\F\colon\clb\to\mathbb{C}$, which is equal to zero on~$\db$. Such a function $\F$ is unique up to a multiplicative constant. Moreover $\F(u)\neq 0$ for all black squares $u \in \dib$ such that $\om\smallsetminus u$ admits a domino tiling.

\end{lemma}
\begin{proof}
1. Consider a system of linear equations with variables that correspond to values of $\F$ in the black faces, and each equation means that the function $\F$ is holomorphic in some white face.
The number of variables is equal to the number of black faces, and the number of equations is equal to the number of white faces. So we have a linear system with a square matrix, the same linear system as~(\ref{wcr}) but with a vanishing right hand side.
To prove that this system has only a trivial solution it is enough to show that the determinant of the matrix is not equal to zero. Note that the absolute value of the determinant is equal to the number of the domino tilings of $\om$, since the matrix is the Kasteleyn matrix of $\om$. Hence, if the domain has a domino tiling then the determinant is not zero. Therefore  $\F\equiv0.$

2. We can consider a system of linear equations in the same way as described above. Note that in this case the number of variables is one more then the number of equations. Hence the system has a non-trivial solution. Let $\F$ have the values which correspond to this solution.
Let  $u'$ be a square in $\dib$ and let the domain  $\om\smallsetminus u'$ have a domino tiling. Let $\F$ be equal to zero at $u'$. 
Note that the function $\F$ satisfies the conditions of the first part of the lemma, therefore $\F\equiv0$ on $\om$. We obtain a contradiction with a non-triviality of the solution of our system.
Similarly to the proof of the first part of the lemma we can show that there is the unique discrete holomorphic function $\F$ such that $\F(u')=1$. 
\end{proof}

\begin{cor}\label{G(v0)} \label{main-lemma} 
Let a domain $\om$ contain the same number of black and white squares, and let $\om$ admit a domino tiling. Fix a black square $\bbb \in \dintb{0}$ and a white square $\www~\in~\dintw{0}$ such that the domain $\om\smallsetminus \{\bbb, \www\}$ admits a domino tiling. Then the following holds:

$1$. There exists the unique function $\F\colon\clb\to\mathbb{C}$ such that $\F|_{\db}=0$ and $\F$ is discrete holomorphic everywhere in $\ww{}$ except at the face $\www$ where one has $[\bar{\partial}\F](\www) = \lambda.$ Moreover, $\F(\bbb)\neq 0$.

$2$. Similarly, there exists the unique function $\G~\colon~\clw~\to~\mathbb{C}$ such that $\G|_{\dw}=0$ and $\G$ is discrete holomorphic everywhere in $\bb{}$ except at the face $\bbb$ where one has $[\bar{\partial}\G](\bbb) = i.$ Moreover, $\G(\www)\neq 0$.
\end{cor}

\begin{proof}
Due to Lemma~\ref{complan} the function $\F$ on $\om\smallsetminus\www$ is unique up to a multiplicative constant. Moreover, $\F(\bbb)\neq 0$ since the domain $\om\smallsetminus \{\bbb, \www\}$ admits a domino tiling. Therefore, $[\bar{\partial}\F](\www) \neq 0$ (otherwise $\F$ is identically zero due to Lemma~\ref{complan}). Finally, the condition $[\bar{\partial}\F](\www) = \lambda$ defined the function $\F$ uniquely.
\end{proof}

In the setup of Corollary~\ref{main-lemma}, we construct the function $\HH$ defined on the vertex set of the domain $\om~\smallsetminus~\{\bbb,\www \}$ as described in Section~\ref{2.2}. So, the formula $(\ref{def H})$ holds for all square edges of the domain $\om$ except boundary edges of the squares $\bbb$, $\www$.
Note that if $\bbb$ and $\www$ are not corner squares of the domain $\om$,
then the vertex set of the domain $\om\smallsetminus\{\bbb,\www \}$ and the vertex set of the domain $\om$ are the same. 
Define $\dom=(\bbb\www)\cup(\www\bbb)$, see Fig.~\ref{colors}.
Note that the product $\F\cdot\G$ along each boundary square edge equals zero, since $\F|_{\db}=0$ and $\G|_{\dw}=0$. Therefore $\HH$ is constant on each of boundary segments.
Recall that $\HH$ is defined up to an additive constant, which can be chosen so that $\HH|_{(\www\bbb)}\equiv 0.$

\begin{lemma}\label{C_H}
The value of the function $\HH$ on the boundary segment $(\bbb\www)$ equals 
\[
\HH|_{(\bbb\www)}=4i\G(\www)[\bar{\partial}\F](\www)=-4i\F(\bbb)[\bar{\partial}\G](\bbb)\neq 0.
\]
\end{lemma}

\begin{proof}
Consider the difference between the values of the function $\HH$ in boundary vertices of the square $\www$:
\begin{align*}
(\HH(z_\flat) - \HH(z_+)) &= (\HH(z_\sharp) - \HH(z_+)) +(\HH(z_-) - \HH(z_\sharp)) + (\HH(z_\flat) - 
\HH(z_-))\\
&=\G(\www)(-\bar{\lambda}\F(u_1)-{\lambda}\F(u_2)+\bar{\lambda}\F(u_3))\\
&=4i\G(\www)[\bar{\partial}\F](\www),
\end{align*}
where $u_1$, $u_2$, $u_3$, $z_+$, $z_-$, $z_\sharp$, $z_\flat$ and $\www$ are defined as shown in Fig.~\ref{d-b = c-a v_0}.

The second expression for $\HH|_{(\bbb\www)}$ can be obtained in a similar way. Finally, $\HH|_{(\bbb\www)}\neq 0$ since $\G(\www)\neq 0.$
\end{proof}

Recall that we can think about the inverse Kasteleyn matrix $C_\om(u,v)$ as a function of two variables $u \in \bb{}$ and $v \in \ww{}$. If $v\in\ww{0}$, then $C_\om(u,v)$ is a discrete holomorphic function of $u$, with a simple pole at $v$: 
\[4\bar{\lambda}\bar{\partial}[C_\om(u,v)](v) =C_{\om} (v+\lambda,v)-C_{\om} (v-\lambda,v)+
 iC_{\om} (v+\bar{\lambda},v)-
 iC_{\om} (v-\bar{\lambda},v)=1,\]
since the product of the Kasteleyn matrix and the inverse Kasteleyn matrix is equal to the identity matrix. 
 
Let functions $\F$ and $\G$ be constructed as in Corollary~\ref{main-lemma}. Let $\omp = \om \smallsetminus \lbrace\bbb,\www\rbrace$. Recall that $C_{\operatorname{dbl-d}, \om}(u, v)=\Cm(u,v) - \Cmm(u,v)$.
 
\begin{prop}[factorization of the double-dimer coupling function]\label{C=FG} 
Let $u\in\bb{}$ and $v \in \ww{}$, then the following identity holds 
\[ C_{\operatorname{dbl-d}, \om}(u, v) = \mathrm{const}\cdot\F(u)\G(v),\] 
where  $\mathrm{const}=\frac{1}{4\G(\www)}$. 
\end{prop}

\begin{proof}
For a fixed $\widetilde{v}\in \ww{}$, consider $\Cm(u,\widetilde{v}) - \Cmm(u,\widetilde{v})$ as a function of~$u$. This function is holomorphic at all faces in $\ww{}\smallsetminus\www$. Moreover $\bar{\partial}[(\Cm-\Cmm)(u,\widetilde{v})](\www) \neq 0$, since otherwise the function $\Cm(u,\widetilde{v}) - \Cmm(u,\widetilde{v})$ is discrete holomorphic everywhere in $\om$ and vanishes on the boundary and then $\Cm(u,\widetilde{v}) - \Cmm(u,\widetilde{v})\equiv 0$ from Lemma~\ref{complan}.
 Hence, for fixed $\widetilde{v}\in \ww{}$ this difference is  equal to $\F(u)$ up to a multiplicative constant. So, $$\Cm(u,\widetilde{v}) - \Cmm(u,\widetilde{v})=k_1\cdot\F(u),$$ where $k_1$ depends on~$\widetilde{v}$.
 
Similarly, for a fixed $\widetilde{u}\in \bb{}$, consider $\Cm(\widetilde{u},v) - \Cmm(\widetilde{u},v)$ as a function of $v$. We obtain that $\Cm(\widetilde{u},v) - \Cmm(\widetilde{u},v)=k_2\cdot{\lambda}\G(v)$, where $k_2$ depends on $\widetilde{u}$. 
 
Therefore 
\[
\Cm(u,v) - \Cmm(u,v) =\mathrm{const}\cdot \F(u)\G(v).
\] 
Consider $\Cm(u,\www) - \Cmm(u,\www)$ as a function of~$u$. Note that 
\[
\Cmm(u,\www)\equiv 0.
\] 
Hence 
\[
\Cm(u,\www) =\mathrm{const}\cdot \F(u)\G(\www).
\] 
Recall that 
\[4\bar{\partial}[\Cm(u,\www)](\www)=\lambda.\]  Thus, $\mathrm{const}=\frac{1}{4\G(\www)}$.
\end{proof}

\begin{cor}\label{h=H}
Let $h$ be the height function in the double-dimer model on the vertices of the domain~$\om$. Then for all $z \in \V{}$ the following equality holds 
\[
\mathbb{E}[h(z)]=\HH(z)\cdot\HH|_{(\bbb\www)}^{-1},
\]
 where the value $\HH|_{(\bbb\www)}$ is given in Lemma~\ref{C_H}.
\end{cor}

\begin{proof} 
Let $h_\om$ and $h_\om'$ be height functions in the dimer model on domains $\om$ and $\om'$, i.e. $h=h_\om-h_\om'$. Recall that the probability that there is a domino $[uv]$ in the domino tiling of $\om$ is equal to $|C_\om(u,v)|$. It is easy to see, that

\[\mathbb{E}[\h_{\om}(z_1)-h_{\om}(z_2)]=\tfrac34\cdot\mathbb{P}[uv]+
(-\tfrac{1}{4})\cdot(1-\mathbb{P}[uv]),\] where $u$, $v$ are adjacent squares; and $z_1$, $z_2$ are their common vertices. 
Therefore, 
\[\mathbb{E}[\h_{\om}(z_1)-h_{\om}(z_2)]=\mathbb{P}[uv]-\tfrac14=|C_{\om} (u,v)|-\tfrac14.\]

Similarly, $\mathbb{E}[\h_{\om'}(z_1)-h_{\om'}(z_2)]=|C_{\om'} (u,v)|-\frac14.$

So, $\mathbb{E}[\h(z_1)-h(z_2)]=|C_{\om} (u,v)|-|C_{\om'} (u,v)|.$

Note that for $u_1$, $u_2$, $u_3$, $u_4$ and $v$ defined as shown on Fig.~\ref{d-b = c-a} the following equality holds:
\begin{align*}
1 &=\mathbb{P}[u_1v]+\mathbb{P}[u_2v]+\mathbb{P}[u_3v]+\mathbb{P}[u_4v]\\
&=|C_{\om} (u_1,v)|+|C_{\om} (u_2,v)|+|C_{\om} (u_3,v)|+|C_{\om} (u_4,v)|.
\end{align*}

Moreover,  
\[C_{\om} (u_2,v)+iC_{\om} (u_3,v)-
 C_{\om} (u_4,v)-
 iC_{\om} (u_1,v)=1,\]
since the product of the Kasteleyn matrix and the inverse Kasteleyn matrix is equal to the identity matrix. Therefore 
\[|C_{\om} (u,v)|-|C_{\om'} (u,v)|={\tau(uv)}\cdot(C_{\om} (u,v)-C_{\om'} (u,v)),\] 
where $\tau(uv)$ is the Kasteleyn weight of the edge $(uv)$. To complete the proof it is enough to apply  Proposition~\ref{C=FG}. 
\end{proof}

\subsection{Proof of Theorem~\ref{leap-frog harmonicity}}
We call a discrete domain an {\it odd Temperleyan domain} if all its corner squares are of type $\bb{0}$. Recall that to obtain a Temperleyan domain one should remove a square of type $\bb{0}$ from the set $\dib$ from an odd Temperleyan domain, see Fig.~\ref{Temp}. A Temperleyan domain always admits a domino tiling.

We need to adjust the notation from the previous section to this setup.
Corollary~\ref{main-lemma} is stated for the case of the domain containing the same number of black and white squares. If we consider a discrete domain in which the number of black squares is greater by one than the number of white squares (see.~Fig.~\ref{I}), then we have some differences in definitions of functions $\F$ and $\G$.
Fix two black squares $u_1, u_2 \in \dib$ in such a way, that after removing one of them the resulting domain admits a domino tiling. Let $u_1\in \bb{0}$.
\begin{enumerate}
\item[1.] There exists the unique function $\F\colon\clb\to\mathbb{C}$ such that $\F|_{\db}=0$, $F(u_1) = 1$ and $\F$ is discrete holomorphic everywhere in $\ww{}$.

\item[2.] There exists the unique function $\G~\colon~\clw~\to~\mathbb{C}$ such that $\G|_{\dw}=0$ and $\G$ is discrete holomorphic everywhere in $\bb{}$ except at faces $u_1$, $u_2$ and one has $[\bar{\partial}\G](u_2) = i.$
\end{enumerate}
The existence and the uniqueness of functions $\F$ and $\G$ follow from Lemma~\ref{complan}.

\begin{proof}[ Proof of Theorem~\ref{leap-frog harmonicity}]
Let $\om$ be an odd Temperleyan domain. Note that Proposition~\ref{C=FG} and Corollary~\ref{h=H} are still true in odd case.
So, it is enough to show that $\HH$ is a discrete leap-frog harmonic function.
This follows directly from Proposition~\ref{formula lap}. In this case $\F$ is a discrete holomorphic function at all white squares of $\om$. So, its imaginary part is a discrete harmonic function with zero boundary conditions. Therefore $\im\F$ is identically zero, and thus the real part of $\F$ is a constant. 
Hence, $\partial\F$ is identically zero. 
\end{proof}

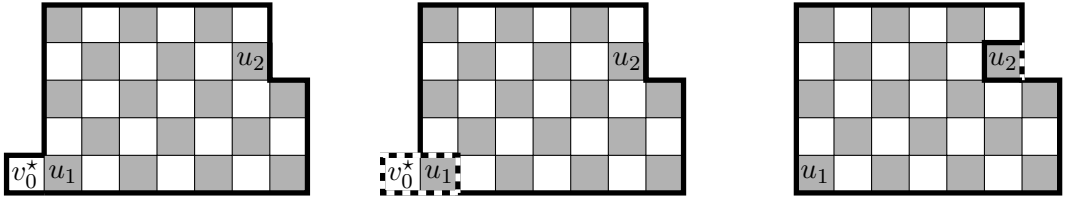
\begin{figure}
\begin{center}
\begin{tikzpicture}[x={(0.5cm,0cm)}, y={(0cm,0.5cm)}]
\begin{scope}
\draw  [draw, fill=gray!60](0,0) rectangle (1,1);
\draw  [draw, fill=gray!60](1,1) rectangle (2,2);
\draw  [draw, fill=gray!60](2,2) rectangle (3,3);
\draw  [draw, fill=gray!60](3,3) rectangle (4,4);
\draw  [draw, fill=gray!60](4,4) rectangle (5,5);

\draw  [draw, fill=gray!60](0,2) rectangle (1,3);
\draw  [draw, fill=gray!60](1,3) rectangle (2,4);
\draw  [draw, fill=gray!60](2,4) rectangle (3,5);

\draw  [draw, fill=gray!60](0,4) rectangle (1,5);

\draw  [draw, fill=gray!60](2,0) rectangle (3,1);
\draw  [draw, fill=gray!60](3,1) rectangle (4,2);
\draw  [draw, fill=gray!60](4,2) rectangle (5,3);
\draw  [draw, fill=gray!60](5,3) rectangle (6,4);

\draw  [draw, fill=gray!60](4,0) rectangle (5,1);
\draw  [draw, fill=gray!60](5,1) rectangle (6,2);
\draw  [draw, fill=gray!60](6,2) rectangle (7,3);

\draw  [draw, fill=gray!60](6,0) rectangle (7,1);

\draw[line width=2pt] (0,1) -- (0,5) -- (6,5) -- (6,3) -- (7,3) -- (7,0)-- (0,0);


\fill[black] (0.5,0.5)  node[]{$u_1$};

\draw[draw, line width=2pt] (0,0)--(-1,0)--(-1,1)--(0,1);
\fill[black] (-0.5,0.5)  node[]{$\www^\star$};
\draw[draw, line width=2pt] (6,4)--(6,3);

\fill[black] (5.5,3.5)  node[]{$u_2$};
\end{scope}

\begin{scope}[xshift=5cm]

\draw  [draw, fill=gray!60](0,0) rectangle (1,1);
\draw  [draw, fill=gray!60](1,1) rectangle (2,2);
\draw  [draw, fill=gray!60](2,2) rectangle (3,3);
\draw  [draw, fill=gray!60](3,3) rectangle (4,4);
\draw  [draw, fill=gray!60](4,4) rectangle (5,5);

\draw  [draw, fill=gray!60](0,2) rectangle (1,3);
\draw  [draw, fill=gray!60](1,3) rectangle (2,4);
\draw  [draw, fill=gray!60](2,4) rectangle (3,5);

\draw  [draw, fill=gray!60](0,4) rectangle (1,5);

\draw  [draw, fill=gray!60](2,0) rectangle (3,1);
\draw  [draw, fill=gray!60](3,1) rectangle (4,2);
\draw  [draw, fill=gray!60](4,2) rectangle (5,3);
\draw  [draw, fill=gray!60](5,3) rectangle (6,4);

\draw  [draw, fill=gray!60](4,0) rectangle (5,1);
\draw  [draw, fill=gray!60](5,1) rectangle (6,2);
\draw  [draw, fill=gray!60](6,2) rectangle (7,3);

\draw  [draw, fill=gray!60](6,0) rectangle (7,1);

\draw[line width=2pt] (1,0)--(1,1)--(0,1) -- (0,5) -- (6,5) -- (6,3) -- (7,3) -- (7,0)-- cycle;

\fill[black] (0.5,0.5)  node[]{$u_1$};

\draw[draw, line width=2pt] (1,0)--(-1,0)--(-1,1)--(1,1);
\draw[draw, line width=2pt,dashed,white] (1,0)--(-1,0)--(-1,1)--(1,1)--(1,0);
\fill[black] (-0.5,0.5)  node[]{$\www^\star$};
\draw[draw, line width=2pt] (6,4)--(6,3);

\fill[black] (5.5,3.5)  node[]{$u_2$};
\end{scope}

\begin{scope}[xshift=10cm]
\draw  [draw, fill=gray!60](0,0) rectangle (1,1);
\draw  [draw, fill=gray!60](1,1) rectangle (2,2);
\draw  [draw, fill=gray!60](2,2) rectangle (3,3);
\draw  [draw, fill=gray!60](3,3) rectangle (4,4);
\draw  [draw, fill=gray!60](4,4) rectangle (5,5);

\draw  [draw, fill=gray!60](0,2) rectangle (1,3);
\draw  [draw, fill=gray!60](1,3) rectangle (2,4);
\draw  [draw, fill=gray!60](2,4) rectangle (3,5);

\draw  [draw, fill=gray!60](0,4) rectangle (1,5);

\draw  [draw, fill=gray!60](2,0) rectangle (3,1);
\draw  [draw, fill=gray!60](3,1) rectangle (4,2);
\draw  [draw, fill=gray!60](4,2) rectangle (5,3);
\draw  [draw, fill=gray!60](5,3) rectangle (6,4);

\draw  [draw, fill=gray!60](4,0) rectangle (5,1);
\draw  [draw, fill=gray!60](5,1) rectangle (6,2);
\draw  [draw, fill=gray!60](6,2) rectangle (7,3);

\draw  [draw, fill=gray!60](6,0) rectangle (7,1);

\draw[line width=2pt] (0,0) -- (0,5) -- (6,5) -- (6,3) -- (7,3) -- (7,0)-- cycle;

\fill[black] (0.5,0.5)  node[]{$u_1$};

\draw[draw, line width=2pt,dashed,white] (6,4)--(6,3);
\draw[draw, line width=2pt] (6,4)--(5,4)--(5,3)--(6,3);

\fill[black] (5.5,3.5)  node[]{$u_2$};

\end{scope}
\end{tikzpicture}\end{center}
\caption{Left: $\om^{\star}=\om\cup\{\www^\star\}$. Center: $\om^{\star}_1=\om_1\cup\{\www^\star, u_1\}$. Note that each domino covering of the domain $\om^{\star}_1$ has domino $[\www^\star u_1]$. Therefore there is a bijection between the sets of domino coverings of domains $\om_1$ and $\om^{\star}_1$. Right: $\om^{\star}_2=\om_2$.
} \label{odd_even}
\end{figure}

Let $\www^\star$ be a white square on $\dom$ adjacent to $u_1$. Let us define domains $\om^{\star}$, $\om^{\star}_1$ and $\om^{\star}_2$ as it shown on Fig.~\ref{odd_even}. Let $\bbb^\star=u_2$. Then there are unique functions $\F^\star$ and $\G^\star$ satisfying Corollary~\ref{main-lemma} on the domain $\om^{\star}$ with marked squares $\www^\star$ and $\bbb^\star$. 

\begin{remark}\label{odd=even}
It is easy to check that the functions $\F$ (resp., $\G$) defined above equals $\F^\star$ (resp., $\G^\star$) on $\om$. 
Hence there is no difference between odd and even cases in terms of functions $\F$ and $\G$.
\end{remark}

\setcounter{equation}{0}

\section{Double-dimer height function in polygonal domains }\label{4}
From now onwards, we will use the square lattice with mesh size $\delta$ rather than $1$. Let $\om$ be a polygon in $\mathbb{C}$ with sides parallel to vectors $\lambda$ and $\bar{\lambda}$. For each sufficiently small $\delta > 0$, let $\om^\delta$ be a discrete polygon approximating $\om$ on the square lattice with mesh size $\delta$.

Let us define functions $\F^\delta$
 and $\G^\delta$ similarly to the previous section:
 \begin{enumerate}
\item[1.] The function $\F^\delta$ is discrete holomorphic everywhere in $\ww{}^\delta$ except at the face $\www^\delta$ where one has $[\bar{\partial}^\delta\F^\delta](\www^\delta) = \frac{\lambda}{\delta^2}.$
\item[2.] Similarly, the function $\G^\delta$ is discrete holomorphic everywhere in $\bb{}^\delta$ except at the face $\bbb^\delta$ where one has $[\bar{\partial}^\delta\G^\delta](\bbb^\delta) = \frac{i}{\delta^2}.$
 \end{enumerate}

Our goal is to prove the convergence of the functions $\HH^{\delta}$ defined by the formula~$(\ref{def H})$. Recall that this definition can be thought of as ``$\HH^{\delta} = \int^\delta\re[\F^\delta\G^\delta dz]$''. We will prove that the functions $\F^{\delta}$ and $\G^{\delta}$ converge individually.

To prove the convergence of the functions $\F^{\delta}$ we will consider approximations by domains $\om^{\delta}$ with fixed colour type of the corners. We will describe this classification below. The limits of the functions  $\F^{\delta}$ and $\G^{\delta}$ depend on the type of the corners. At the same time the limit of the functions $\HH^{\delta}$ does not depend on the type of the corners.

We will call a corner of $\om^\delta$ a \emph{convex} corner if the interior angle is $\pi/2$, and \emph{concave} if the interior angle is $3\pi/2$. A corner is called \emph{white} if there is a white square in the corner, and \emph{black} if there is a black square in this corner, see Fig.~\ref{colors}.

\begin{lemma}\label{combi-lemma}
If a simply connected domain $\om^\delta$ contains the same number of black and white squares then 
\[
\#\lbrace white\,\, convex\,\, corners\rbrace = \#\lbrace white\,\, concave\,\, corners\rbrace + 2,
\]
\[
\#\lbrace black\,\, convex\,\, corners\rbrace = \#\lbrace black\,\, concave\,\, corners\rbrace + 2.
\]
\end{lemma}

\begin{proof}
Note that 
$\pi \cdot( \#\lbrace corners\rbrace - 2) = \frac{\pi}{2}\cdot\#\lbrace convex\,\, corners\rbrace + \frac{3\pi}{2}\cdot \#\lbrace concave\,\, corners\rbrace,$
hence 
\[
\#\lbrace convex\,\, corners\rbrace = \#\lbrace concave\,\, corners\rbrace + 4.
\]

Recall that the height along the boundary changes by $\pm \frac14$: if an edge has a black square on its left then the height increases by $\frac14$; if it has a white square on its left then the height decreases by~$\frac14$. Along each straight segment of the boundary of the domain the height function varies between two values. This pair increases (resp., decreases) by $\frac14$ if the boundary turns left along black (resp., white) convex square, and decreases (resp., increases) by $\frac14$ if it turns right along black (resp., white) concave square.
Then
\begin{align*}
\#\lbrace white\,\, convex\,\, corners\rbrace \, + \, &\#\lbrace black\,\, concave\,\, corners\rbrace = \\
&\#\lbrace white\,\, concave\,\, corners\rbrace \, + \, \#\lbrace black\,\, convex\,\, corners\rbrace,
\end{align*}
since the height function on the boundary is well defined if the domain contains the same number of black and white squares (this is easily proved by induction on the number of black squares, starting from the case of a $2\times1$ rectangle).  
\end{proof}

Let $\om^{\delta}$ admit a domino tiling. Let $\bbb^{\delta}$ and $\www^{\delta}$ be black and white squares in $\diom^\delta$ placed away from the corners of $\om^\delta$ in such a way that the domain $\om^\delta\smallsetminus\{\bbb^{\delta}, \www^{\delta}\}$  admits a domino tiling.
Let $\{\wvpugl{k}\}^{n-1}_{k=1}$ be the set of white squares located in the concave white corners of the domain  $\om^{\delta}$, and let $\{\wvupugl{k}\}^{n+1}_{k=1}$ be the set of white squares located in the convex white corners of the domain  $\om^{\delta}$, see Fig.~\ref{colors}.
 Recall that the cardinality of the latter set is greater by two than the cardinality of the former due to Lemma~\ref{combi-lemma}.
Similarly,  let $\{\bvpugl{s}\}^{m-1}_{s=1}$ be the set of black squares located in the concave black corners of the domain  $\om^{\delta}$, and let  $\{\bvupugl{s}\}^{m+1}_{s=1}$ be the set of black squares located in the convex black corners of the domain  $\om^{\delta}$ (see~Fig.~\ref{colors}).

\subsection{
Discrete boundary value problem for the functions $\F$~and~$\G$}

Note that for all $u^\delta\in {\db}^\delta$ one has $\F^\delta(u^\delta)=0$, which can be thought of as a zero Dirichlet boundary conditions either for $\re[\F^\delta]$ or for $\im[\F^\delta]$. Similarly, for all $v^\delta\in \dw^\delta$, either $\re[\bar{\lambda}\G^\delta]$ or $\re[\lambda\G^\delta]$ has zero Dirichlet boundary conditions.

\begin{remark}\label{change}
The function $\F^\delta$ (resp., $\G^\delta$) changes boundary conditions only at white (resp., black) corners of $\om^\delta$. 
\end{remark}

A function on a discrete domain $\om^\delta$ is called {\it semibounded by its boundary values} in a subdomain $U^\delta\subset\om^\delta$ if either the maximum or the minimum of this function in $U^\delta$ is attained on the boundary of $U^\delta$.
A function on a discrete domain $\om^\delta$ is called {\it bounded by its boundary values} in a subdomain $U^\delta\subset\om^\delta$ if both, the maximum and the minimum of this function in $U^\delta$, are attained on $\partial U^\delta$.

\begin{figure}
\begin{center}
\begin{tikzpicture}[x={(0.37cm,0.37cm)}, y={(-0.37cm,0.37cm)}]
\begin{scope}
\path (0,0) node[name=l1, shape=coordinate]{};
\path (1,0) node[name=l2, shape=coordinate]{};
\path (1,1) node[name=l3, shape=coordinate]{};
\path (0,1) node[name=l4, shape=coordinate]{};
\path[draw, fill=gray!20] (l1)--(l2)--(l3)--(l4)--cycle;
\path ($1/2*(l1)+1/2*(l3)$) node[]{$u$};
\path (1.5,1.5) node[]{$0$};
\path (-0.5,1.5) node[]{$0$};

\path (0,-2) node[name=l1, shape=coordinate]{};
\path (1,-2) node[name=l2, shape=coordinate]{};
\path (1,-1) node[name=l3, shape=coordinate]{};
\path (0,-1) node[name=l4, shape=coordinate]{};
\path[draw, fill=gray!20] (l1)--(l2)--(l3)--(l4)--cycle;
\path ($1/2*(l1)+1/2*(l3)$) node[]{$u_2$};

\path (-2,0) node[name=l1, shape=coordinate]{};
\path (-1,0) node[name=l2, shape=coordinate]{};
\path (-1,1) node[name=l3, shape=coordinate]{};
\path (-2,1) node[name=l4, shape=coordinate]{};
\path[draw, fill=gray!20] (l1)--(l2)--(l3)--(l4)--cycle;
\path ($1/2*(l1)+1/2*(l3)$) node[]{$u_{\tiny{1}}$};

\path (2,0) node[name=l1, shape=coordinate]{};
\path (3,0) node[name=l2, shape=coordinate]{};
\path (3,1) node[name=l3, shape=coordinate]{};
\path (2,1) node[name=l4, shape=coordinate]{};
\path[draw, fill=gray!20] (l1)--(l2)--(l3)--(l4)--cycle;
\path ($1/2*(l1)+1/2*(l3)$) node[]{$u_3$};

\path[draw][line width=2pt] (-3,1)--(4,1);

\path (-2.5,-2.5) node[]{$\F^{\delta}(u)=\frac{\F^{\delta}(u_1)+\F^{\delta}(u_2)+\F^{\delta}(u_3)}{3}$};

\end{scope}

\begin{scope}[xshift=5cm]

\path (0,0) node[name=l1, shape=coordinate]{};
\path (1,0) node[name=l2, shape=coordinate]{};
\path (1,1) node[name=l3, shape=coordinate]{};
\path (0,1) node[name=l4, shape=coordinate]{};
\path[draw, fill=gray!20] (l1)--(l2)--(l3)--(l4)--cycle;
\path ($1/2*(l1)+1/2*(l3)$) node[]{$u$};
\path (1.5,-0.5) node[]{$0$};
\path (-0.5,1.5) node[]{$0$};

\path (0,-2) node[name=l1, shape=coordinate]{};
\path (1,-2) node[name=l2, shape=coordinate]{};
\path (1,-1) node[name=l3, shape=coordinate]{};
\path (0,-1) node[name=l4, shape=coordinate]{};
\path[draw, fill=gray!20] (l1)--(l2)--(l3)--(l4)--cycle;
\path ($1/2*(l1)+1/2*(l3)$) node[]{$u_2$};

\path (-2,0) node[name=l1, shape=coordinate]{};
\path (-1,0) node[name=l2, shape=coordinate]{};
\path (-1,1) node[name=l3, shape=coordinate]{};
\path (-2,1) node[name=l4, shape=coordinate]{};
\path[draw, fill=gray!20] (l1)--(l2)--(l3)--(l4)--cycle;
\path ($1/2*(l1)+1/2*(l3)$) node[]{$u_1$};

\path[draw][line width=2pt] (-3,1)--(1,1)--(1,-3);

\path (-2.5,-2.5) node[]{$\F^{\delta}(u)=\frac{\F^{\delta}(u_1)+\F^{\delta}(u_2)}{2}$};

\end{scope}

\begin{scope}[xshift=10cm]

\path (0,0) node[name=l1, shape=coordinate]{};
\path (1,0) node[name=l2, shape=coordinate]{};
\path (1,1) node[name=l3, shape=coordinate]{};
\path (0,1) node[name=l4, shape=coordinate]{};
\path[draw, fill=gray!20] (l1)--(l2)--(l3)--(l4)--cycle;
\path ($1/2*(l1)+1/2*(l3)$) node[]{$u$};
\path (-0.5,1.5) node[]{$0$};
\path (1.5,1.5) node[]{$0$};
\path (2.5,0.5) node[]{$0$};
\path (2.5,-1.5) node[]{$0$};

\path (0,-2) node[name=l1, shape=coordinate]{};
\path (1,-2) node[name=l2, shape=coordinate]{};
\path (1,-1) node[name=l3, shape=coordinate]{};
\path (0,-1) node[name=l4, shape=coordinate]{};
\path[draw, fill=gray!20] (l1)--(l2)--(l3)--(l4)--cycle;
\path ($1/2*(l1)+1/2*(l3)$) node[]{$u_2$};

\path (-2,0) node[name=l1, shape=coordinate]{};
\path (-1,0) node[name=l2, shape=coordinate]{};
\path (-1,1) node[name=l3, shape=coordinate]{};
\path (-2,1) node[name=l4, shape=coordinate]{};
\path[draw, fill=gray!20] (l1)--(l2)--(l3)--(l4)--cycle;
\path ($1/2*(l1)+1/2*(l3)$) node[]{$u_1$};

\path[draw][line width=2pt] (-3,1)--(2,1)--(2,-3);

\path (-2,-3) node[]{$\F^{\delta}(u)=\frac{\F^{\delta}(u_1)+\F^{\delta}(u_2)}{3}$};
\end{scope}
\end{tikzpicture}\caption{Discrete harmonicity of the function $\F^\delta$ together with the boundary conditions implies the following equations for~$u\in\diom^{\delta}$, see also Fig.~\ref{b_c}.
}
\label{3bound}\end{center}
\end{figure}
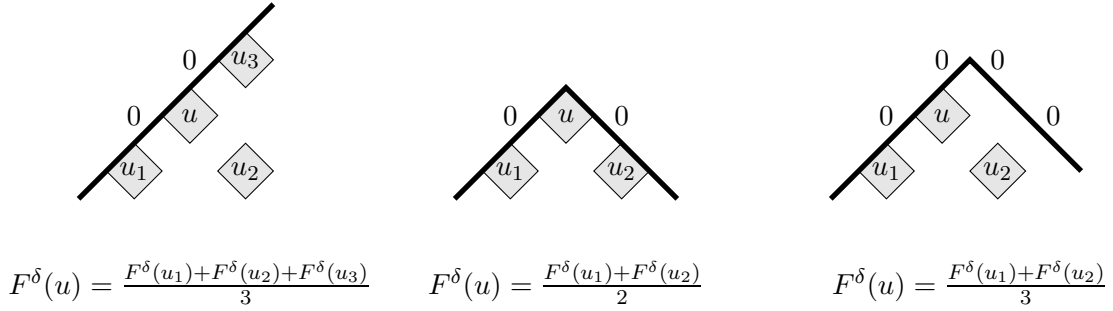

\begin{lemma}\label{bounded}
The function $\F^\delta|_{\bb{0}^\delta}$ is bounded by its boundary values in neighbourhoods of white convex corners and semibounded by its boundary values in neighbourhoods of white concave corners.
\end{lemma}

\begin{proof}
Note that the function $\F^\delta|_{\bb{0}^\delta}$ is discrete harmonic in ${\bb{0}^\delta}$, except at the squares of type $\bb{0}^\delta$ adjacent to $\{\nwvpugl{k}^\delta\}$ and $\www^\delta$, where $\{\nwvpugl{k}^\delta\}$ is the set of white squares in the white concave corners. In particular, the function 
$\F^\delta$ is bounded by its boundary values in vicinities of white convex 
corners $\{\wvupugl{k}\}$, 
see Fig.~\ref{3bound}.

Let us consider a neighbourhood of a corner $\nwvpugl{k}^\delta$. Note that in this neighbourhood the function $\F^\delta|_{\bb{0}^\delta}$ is discrete harmonic everywhere except at the unique black square of type ${\bb{0}^\delta}$ adjacent to $\nwvpugl{k}^\delta$. Note that at this square either the maximum or the minimum of $\F^\delta$ can be reached, thus $\F^\delta|_{\bb{0}^\delta}$ is semi-bounded near $\nwvpugl{k}^\delta$.
\end{proof}

\subsection{The continuous analogue of the functions $\F^\delta$~and~$\G^\delta$} 
In this section we will describe the continuous analogue of the functions $\F^\delta$~and~$\G^\delta$.
Also, we will show that the primitive of their product is the harmonic measure.


\begin{prop}\label{1-4}
Let $\om$ be a simply connected Jordan domain.  Let $\www$ be a boundary point which lies on a straight segment of the boundary of $\om$, and this segment goes to the direction~$\lambda$. Let $\{\nwvupugl{k} \}^{n+1}_{k=1}\cup\{\nwvpugl{k} \}^{n-1}_{k=1}$ be a set of marked points on $\dom\smallsetminus\{\www\}$. Then there exists the unique holomorphic function $f_\om$ on $\om$ such that:

\begin{enumerate}
\item[$\rhd$] $f_\om(z)=\frac{\lambda}{z-\www}+O(1)$ in a vicinity of the point $\www$; 

\item[$\rhd$] $f_\om$ is bounded in vicinities of the points $\nwvupugl{k}$;

\item[$\rhd$] $f_\om$ is semi-bounded (either from above or from below) in vicinities of the points $\nwvpugl{k}$;

\item[$\rhd$] along each boundary arc between marked points $\{\nwvupugl{k} \}^{n+1}_{k=1}\cup\{\nwvpugl{k} \}^{n-1}_{k=1}$, one has either $\re[f_\om]=0$ or $\im[f_\om]=0$;

\item[$\rhd$] aforementioned boundary conditions change at all marked points $\nwvpugl{k}$ and $\, \nwvupugl{s}$.  
\end{enumerate}
\end{prop}

\begin{proof}
Let $\phi$ be a conformal mapping of the domain $\om$ onto the upper half plane $\mathbb{H}$ such that none of the marked points and $\www$ is mapped onto infinity. Then $f_\mathbb{H}:= f_\om\circ \phi^{-1}$ is a holomorphic function on $\mathbb{H}$, which satisfies the following conditions:

\begin{enumerate}
\item \label{pol} $f_\mathbb{H}(w)=\frac{\lambda\cdot\phi'(\www)}{w-\phi(\www)}+O(1)$ in a vicinity of the point $\phi(\www)$;

\item \label{ogr} $f_\mathbb{H}$ is bounded in vicinities of the points $\phi(\nwvupugl{k})$;

\item \label{pologr}$f_\mathbb{H}$ is semi-bounded (either from above or from below) in vicinities of the points $\phi(\nwvpugl{k})$;

\item \label{inf} $f_\mathbb{H}$ is bounded at infinity;

\item \label{RI}  on each segment of the real line between the points of the set $\{\phi(\nwvpugl{k})\}^{n-1}_{k=1}\cup\{\phi(\nwvupugl{k})\}^{n+1}_{k=1}$ one has either $\re[f_\mathbb{H}]=0$ or $\im[f_\mathbb{H}]=0$;

\item \label{smena} the function $f_\mathbb{H}$ changes the boundary conditions at all points $\phi(\nwvpugl{k})$ and $\, \phi(\nwvupugl{k})$, and only at these points.
\end{enumerate}

For a given $k$ let us add a constant to $\phi$ so that $\phi(\nwvpugl{k})=0$.
Let us consider a function $f_\mathbb{H}(w^2)$ in a vicinity of zero. The boundary conditions $(5)$, $(6)$ of the function $f_\mathbb{H}(w^2)$ allow one to extend this function to a punctured vicinity of $0$ by the Schwarz reflection principle. 

Let us show that $f_\mathbb{H}(w^2)= O(1/w)$ as $w \to 0$. Great Picard's Theorem together with the semi-boundedness condition $(3)$ implies that $f_\mathbb{H}(w^2)$  cannot have an essential singularity at zero. So, the function $f_\mathbb{H}(w^2)$ either is regular or has a pole at zero. This pole must be simple due to $(3)$, and hence  $f_\mathbb{H}(w)= O\big((w-\phi(\nwvpugl{k}))^{-\frac12}\big)$ in a vicinity of~$\nwvpugl{k}$.

Similarly, conditions $(2)$, $(5)$ and $(6)$ imply that  $f_\mathbb{H}(w)= O\big((w-\phi(\nwvupugl{k}))^{\frac12}\big)$ in a vicinity of each of the points~$\nwvupugl{k}$

Consider a function 
\[f_\mathbb{H}(w)\cdot(w-\phi(\www))\cdot\prod_{k=1}^{n-1}(w-\phi(\nwvpugl{k}))^{\frac12}\cdot\prod_{k=1}^{n+1}(w-\phi(\nwvupugl{k}))^{-\frac12},
\]
which can be extended to a bounded function in the whole plane by the Schwarz reflection principle. 
Hence it is a constant, and
\[f_\mathbb{H}(w)=\frac{c_\phi}{w-\phi(\www)}\cdot\prod_{k=1}^{n+1}(w-\phi(\nwvupugl{k}))^{\frac12}\cdot\prod_{k=1}^{n-1}(w-\phi(\nwvpugl{k}))^{-\frac12},\]
where the real constant $c_\phi$ can be determined from the condition $(\ref{pol})$.

Since $f_\mathbb{H}= f\circ \phi^{-1}$, we obtain 
\begin{equation}\label{f=}
f_\om(z)=\frac{c_\phi}{(\phi(z)-\phi(\www))} \cdot\prod_{k=1}^{n+1}(\phi(z)-\phi(\nwvupugl{k}))^{\frac12}\cdot\prod_{k=1}^{n-1}(\phi(z)-\phi(\nwvpugl{k}))^{-\frac12},
\end{equation}
where $c_{\phi}$ is a real constant that depends on $\phi$.
\end{proof}

\begin{remark}\label{v_0_inner}
The previous proposition also holds if $\www$ is an inner point of $\om$. In this case 
\[f_\om(z)=c_\phi\cdot\left(\frac{1}{\phi(z)-\phi(\www)}-\frac{1}{\phi(z)-\overline{\phi(\www)}}\right)\cdot\prod_{k=1}^{n+1}(\phi(z)-\phi(\nwvupugl{k}))^{\frac12}\cdot\prod_{k=1}^{n-1}(\phi(z)-\phi(\nwvpugl{k}))^{-\frac12}.\]
\end{remark}

Similarly, for the set of boundary points  $\{\nbvpugl{k}\}^{m-1}_{s=1}\cup\{\nbvupugl{k}\}^{m+1}_{s=1}$ and the point $\bbb$ on a straight segment of the boundary of $\om$ parallel to vector $\bar{\lambda}$, there exists the unique holomorphic function $g$, which satisfies conditions analogous to conditions from Proposition~\ref{1-4}:

\begin{enumerate}
\item[$\rhd$] $g_\om(z)=\frac{i}{z-\bbb}+O(1)$ in a vicinity of the point $\bbb$; 

\item[$\rhd$] $g_\om$ is bounded in vicinities of the points $\nbvupugl{k}$;

\item[$\rhd$] $g_\om$ is semi-bounded in vicinities of the points $\nbvpugl{k}$;

\item[$\rhd$] along each boundary segment between boundary points of the set $\{\nbvupugl{k} \}^{m+1}_{k=1}\cup\{\nbvpugl{k} \}^{m-1}_{k=1}$, one has either $\re[\bar{\lambda}g_\om]=0$ or $\re[\lambda g_\om]=0$;

\item[$\rhd$] aforementioned boundary conditions of the function $g_\om$ change at all points $\nbvpugl{k}$ and $\, \nbvupugl{s}$. 
\end{enumerate}

This function is written as follows
\begin{equation}\label{g=}
g_\om(z)=\frac{\lambda\widetilde{c_{\phi}}}{(\phi(z)-\phi(\bbb))} \cdot \prod_{k=1}^{m+1}(\phi(z)-\phi(\nbvupugl{k}))^{\frac12}\cdot\prod_{k=1}^{m-1}(\phi(z)-\phi(\nbvpugl{k}))^{-\frac12},
\end{equation} 
where $\widetilde{c_{\phi}}$ is a real constant that depends on $\phi$.

It is worth noting that the product of the functions $f_\om(z)$ and $g_\om(z)$ defined by $(\ref{f=})$ and $(\ref{g=})$ , respectively, does not depend on the colours of corners of $\om$ (while each of $f_\om(z)$, $g_\om(z)$ does depend on these colours).

\begin{prop}\label{f*g}
Let $\om$ be a polygon in $\mathbb{C}$ with sides parallel to vectors $\lambda$ and $\bar{\lambda}$. Let $\www$~and~$\bbb$ be the points on the straight part of the boundary of the polygon~$\om$. Let  $\{\nwvupugl{k}\}^{n+1}_{k=1}\cup\{\nbvupugl{s}\}^{m+1}_{s=1}$ be the set of vertices of the convex corners of the polygon~$\om$, and $\{\nwvpugl{k}\}^{n-1}_{k=1}\cup\{\nbvpugl{s}\}^{m-1}_{s=1}$ be the set of vertices of the concave corners of the polygon~$\om$.  Assume that the boundary arc $(\bbb\www)$ contains~$0$.

Let functions $f_\om$ and $g_\om$ be defined as in Proposition~$\ref{1-4}$, then the function 
\[\int^w_0 \re[f_\om(z)g_\om(z)dz]\] 
is proportional to the harmonic measure $\operatorname{hm}_\om(w, (\www\bbb))$ in the domain~$\om$.
\end{prop}

\begin{proof} 
Let us consider the product of functions $f_\om(z)$ and $g_\om(z)$. It equals
\begin{align*}
f_\om(z)\cdot g_\om(z)&=\frac{\lambda c_{\phi}\widetilde{c_{\phi}}}{(\phi(z)-\phi(\www))\cdot(\phi(z)-\phi(\bbb))}\times \\
&\prod_{k=1}^{n+1}(\phi(z)-\phi(\nwvupugl{k}))^{\frac12}
\cdot\prod_{k=1}^{n-1}(\phi(z)-\phi(\nwvpugl{k}))^{-\frac12}
\cdot \prod_{k=1}^{m+1}(\phi(z)-\phi(\nbvupugl{k}))^{\frac12}
\cdot\prod_{k=1}^{m-1}(\phi(z)-\phi(\nbvpugl{k}))^{-\frac12}.
\end{align*}

Let $\psi(w)$ be a conformal transformation of the upper half-plane onto the interior of a simple polygon $\om$, the inverse mapping to $\phi$. The Schwarz–Christoffel mapping theorem implies that

\begin{align*}
\psi'(w)={\lambda}c_{\psi}\cdot 
\prod_{k=1}^{n-1}(w-\phi(\nwvpugl{k}))^{\frac12}
\cdot \prod_{k=1}^{m-1}(w-\phi(\nbvpugl{k}))^{\frac12}
\cdot \prod_{k=1}^{n+1}(w-\phi(\nwvupugl{k}))^{-\frac12}
\cdot \prod_{k=1}^{m+1}(w-\phi(\nbvupugl{k}))^{-\frac12},
\end{align*}
where $c_{\psi}$ is a real constant.

Note that $\phi$ is the inverse mapping to $\psi$, so $\frac{1}{\psi'(\phi(z))}=\phi'(z).$

Therefore
\[
f(z)\cdot g(z)=\frac{\lambda c_{\phi}\widetilde{c_{\phi}}{\lambda}c_{\psi}\cdot\phi'(z)}{(\phi(z)-\phi(\www))(\phi(z)-\phi(\bbb))}=
\frac{i c_{\phi}\widetilde{c_{\phi}}c_{\psi}}{\phi(\www)-\phi(\bbb)}\cdot\left(\log\frac{(\phi(z)-\phi(\www))}{(\phi(z)-\phi(\bbb))}\right)',\]
hence $\int \re[fgdz]$ is proportional to $\frac{1}{\pi}\im \log \left(\frac{(\phi(z)-\phi(\www))}{(\phi(z)-\phi(\bbb))}\right)$ which is the harmonic measure of~$(\www\bbb)$.
\end{proof}

Now to complete the proof of Theorem~\ref{main-th} it is enough to prove convergence of functions $\F^\delta$ and~$\G^\delta$. In Section~\ref{shodimost} we will prove a more general result: the convergence of $\F^\delta$ for approximations by black-piecewise Temperleyan domains. This special type of discrete domains is defined below in Section~\ref{pwd}. Similarly, one can show the convergence of $\G^\delta$ for approximations by white-piecewise Temperleyan domains. In the setup of Proposition~\ref{f*g} 
the polygonal approximations $\om^\delta$ are $2n$-black-piecewise Temperleyan and $2m$-white-piecewise Temperleyan domains at the same time.

\begin{figure}
\begin{center}
\begin{tikzpicture}[x={(0.4cm,0cm)}, y={(-0cm,0.4cm)}]

\draw  [draw, fill=gray!20](2,4) rectangle (3,5);
\draw  [draw, fill=gray!20](2,6) rectangle (3,7);

\draw  [draw, fill=gray!20](4,4) rectangle (5,5);
\draw  [draw, fill=gray!20](4,6) rectangle (5,7);

\draw  [draw, fill=gray!20](6,4) rectangle (7,5);
\draw  [draw, fill=gray!20](6,6) rectangle (7,7);
\draw  [draw, fill=gray!20](6,8) rectangle (7,9);
\draw  [draw, fill=gray!20](6,10) rectangle (7,11);

\draw  [draw, fill=gray!20](8,2) rectangle (9,3);
\draw  [draw, fill=gray!20](8,4) rectangle (9,5);
\draw  [draw, fill=gray!20](8,6) rectangle (9,7);
\draw  [draw, fill=gray!20](8,8) rectangle (9,9);
\draw  [draw, fill=gray!20](8,10) rectangle (9,11);

\draw  [draw, fill=gray!20](10,2) rectangle (11,3);
\draw  [draw, fill=gray!20](10,4) rectangle (11,5);
\draw  [draw, fill=gray!20](10,6) rectangle (11,7);
\draw  [draw, fill=gray!20](10,8) rectangle (11,9);
\draw  [draw, fill=gray!20](10,10) rectangle (11,11);

\draw  [draw, fill=gray!20](12,2) rectangle (13,3);
\draw  [draw, fill=gray!20](12,4) rectangle (13,5);
\draw  [draw, fill=gray!20](12,6) rectangle (13,7);
\draw  [draw, fill=gray!20](12,8) rectangle (13,9);

\draw  [draw, fill=gray!20](14,2) rectangle (15,3);
\draw  [draw, fill=gray!20](14,4) rectangle (15,5);
\draw  [draw, fill=gray!20](14,6) rectangle (15,7);

\draw  [draw, fill=gray!20](16,2) rectangle (17,3);
\draw  [draw, fill=gray!20](16,4) rectangle (17,5);

\draw  [draw, fill=gray!20](18,2) rectangle (19,3);

\draw  [draw, fill=gray!20](20,2) rectangle (21,3);

\draw  [draw, fill=gray!90](1,5) rectangle (2,6);
\draw  [draw, fill=gray!90](1,7) rectangle (2,8);

\draw  [draw, fill=gray!90](3,5) rectangle (4,6);
\draw  [draw, fill=gray!90](3,7) rectangle (4,8);

\draw  [draw, fill=gray!90](5,5) rectangle (6,6);
\draw  [draw, fill=gray!90](5,7) rectangle (6,8);
\draw  [draw, fill=gray!90](5,9) rectangle (6,10);

\draw  [draw, fill=gray!90](7,1) rectangle (8,2);
\draw  [draw, fill=gray!90](7,3) rectangle (8,4);
\draw  [draw, fill=gray!90](7,5) rectangle (8,6);
\draw  [draw, fill=gray!90](7,7) rectangle (8,8);
\draw  [draw, fill=gray!90](7,9) rectangle (8,10);

\draw  [draw, fill=gray!90](9,1) rectangle (10,2);
\draw  [draw, fill=gray!90](9,3) rectangle (10,4);
\draw  [draw, fill=gray!90](9,5) rectangle (10,6);
\draw  [draw, fill=gray!90](9,7) rectangle (10,8);
\draw  [draw, fill=gray!90](9,9) rectangle (10,10);

\draw  [draw, fill=gray!90](11,1) rectangle (12,2);
\draw  [draw, fill=gray!90](11,3) rectangle (12,4);
\draw  [draw, fill=gray!90](11,5) rectangle (12,6);
\draw  [draw, fill=gray!90](11,7) rectangle (12,8);

\draw  [draw, fill=gray!90](13,1) rectangle (14,2);
\draw  [draw, fill=gray!90](13,3) rectangle (14,4);
\draw  [draw, fill=gray!90](13,5) rectangle (14,6);

\draw  [draw, fill=gray!90](15,3) rectangle (16,4);
\draw  [draw, fill=gray!90](15,1) rectangle (16,2);

\draw  [draw, fill=gray!90](17,1) rectangle (18,2);

\draw  [draw, fill=gray!90](19,1) rectangle (20,2);

\draw[draw, line width=2pt,gray!20] (1,4)--(7,4);
\draw[draw, line width=2pt,gray!20] (5,11)--(11,11)--(11,9)--(13,9)--
(13,7)--(15,7)--(15,5)--(17,5)--(17,3)--(21,3)--(21,1);
\draw[draw, line width=2pt, gray!90] (1,4)--(1,8)--(5,8)--(5,11);
\draw[draw, line width=2pt, gray!90] (21,1)--(7,1)--(7,4);

\path[draw, fill=black] (1,4) circle[radius=0.05cm];
\path[draw, fill=black] (7,4) circle[radius=0.05cm];
\path[draw, fill=black] (5,11) circle[radius=0.05cm];
\path[draw, fill=black] (21,1) circle[radius=0.05cm];
\end{tikzpicture} \caption{A $4$-black-piecewise Temperleyan domain.
}\label{pwTemp}\end{center}
\end{figure}
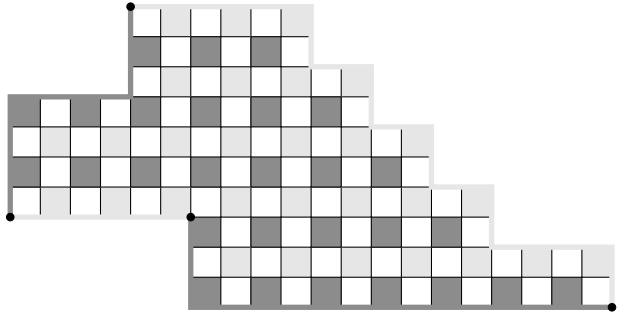

\setcounter{equation}{0}

\section{Convergence of $\F^\delta$ in black-piecewise Temperleyan domains}\label{shodimost}

\subsection{Black-piecewise Temperleyan domains}\label{pwd} 
Let us fix a natural number $n$.
A discrete domain is called a {\it $2n$-black-piecewise Temperleyan domain} if it is a domain with $n+1$ convex white corners and $n-1$ concave white corners. Consider a segment of the boundary between two neighbouring white corners; we will call such a segment a  {\it black Temperleyan segment}. Note that all black squares on this part of the boundary are of the same type: either they all are of type $\bb{0}^\delta$ or all of type $\bb{1}^\delta$ (see Fig.~\ref{pwTemp}).

Let $\om$ be a bounded, simply connected Jordan domain with a piecewise-smooth boundary and $2n$ boundary marked points $\nwvupugl{1}, \ldots, \nwvupugl{n+1}$, $\nwvpugl{1}, \ldots, \nwvpugl{n-1}$. 
For sufficiently small $\delta$,we say that a $2n$-black-piecewise Temperleyan domain $\om^\delta$ approximates $\om$ if the boundaries of the $2n$-black-piecewise Temperleyan domain are within $O(\delta)$ of the boundaries of $\om$, and if furthermore, all convex white corners  $\wvupugl{k}$ are within $O(\delta)$ of the set of marked points $\nwvupugl{k}$ and all concave white corners  $\wvpugl{j}$ are within $O(\delta)$ of the set of marked points~$\nwvpugl{j}$.

\subsection{Proof of the convergence
}\label{convergence}
Let $u^\delta$ be a square on the square lattice with mesh size $\delta$. By $B_r^\delta(u^\delta)$ we denote the set of squares on this lattice such that the distance from them to $u^\delta$ is less then or equal to $r$. Let $\partial B_r^\delta(u^\delta)$ be the set of boundary squares of the set $B_r^\delta(u^\delta)$.

Consider a discrete domain $\om^\delta$. Let $E^\delta$ be a subset of the set $\dom^\delta$. Let $\operatorname{hm}_{\om^\delta}(x^\delta,E^\delta)$ be a discrete harmonic function in $\om^\delta$ such that it is equal to $\chi_{E^\delta}$ on the boundary of $\om^\delta$, where $\chi_{E^\delta}$ is the characteristic function of the set $E^\delta$. The function $\operatorname{hm}_{\om^\delta}(x^\delta,E^\delta)$ is called the harmonic measure. Note that the harmonic measure is a probabilistic measure for any fixed $x^\delta\in\om^\delta$. Note also that the value of $\operatorname{hm}_{\om^\delta}(x^\delta,E^\delta)$ equals to the probability that a simple random walk starting at $x$ first hits the boundary of the domain $\om^\delta$ on the set $E^\delta$.

Let $F_{\mathrm{harm}}^\delta$ be a discrete harmonic function in $\om^\delta$ defined on the set $\om^\delta\cup\dom^\delta$. Then it is easy to see that 
\[
F_{\mathrm{harm}}^\delta(x^\delta)=\sum\limits_{y^\delta\in\dom^\delta}F_{\mathrm{harm}}^\delta(y^\delta)\cdot\operatorname{hm}_{\om^\delta}(x^\delta,\{y^\delta\}).
\]

\begin{remark}
From now on we assume that $\delta > 0$ and $r > 0$ are chosen so that 
the discrete punctured vicinity 
$B_{r}^{\delta}(\wvpugl{k})\smallsetminus\{\wvpugl{k}\}$ contains neither $\www^\delta$ nor white corner squares of $\om^\delta$ for all $k\in\{1,\ldots,n-1\}$.  
\end{remark}

\begin{figure}
\begin{center}
\begin{tikzpicture}[x={(0.33cm,-0.33cm)}, y={(0.33cm,0.33cm)}]
\path (6,8) node[name=l1, shape=coordinate]{};
\path (7,8) node[name=l2, shape=coordinate]{};
\path (7,9) node[name=l3, shape=coordinate]{};
\path (6,9) node[name=l4, shape=coordinate]{};
\path[draw, fill=gray!20] (l1)--(l2)--(l3)--(l4)--cycle;

\path (8,8) node[name=l1, shape=coordinate]{};
\path (9,8) node[name=l2, shape=coordinate]{};
\path (9,9) node[name=l3, shape=coordinate]{};
\path (8,9) node[name=l4, shape=coordinate]{};
\path[draw, fill=gray!20] (l1)--(l2)--(l3)--(l4)--cycle;

\path (8,6) node[name=l1, shape=coordinate]{};
\path (9,6) node[name=l2, shape=coordinate]{};
\path (9,7) node[name=l3, shape=coordinate]{};
\path (8,7) node[name=l4, shape=coordinate]{};
\path[draw, fill=gray!20] (l1)--(l2)--(l3)--(l4)--cycle;

\path (8,4) node[name=l1, shape=coordinate]{};
\path (9,4) node[name=l2, shape=coordinate]{};
\path (9,5) node[name=l3, shape=coordinate]{};
\path (8,5) node[name=l4, shape=coordinate]{};
\path[draw, fill=gray!20] (l1)--(l2)--(l3)--(l4)--cycle;

\path (6,4) node[name=l1, shape=coordinate]{};
\path (7,4) node[name=l2, shape=coordinate]{};
\path (7,5) node[name=l3, shape=coordinate]{};
\path (6,5) node[name=l4, shape=coordinate]{};
\path[draw, fill=gray!20] (l1)--(l2)--(l3)--(l4)--cycle;

\path (6,2) node[name=l1, shape=coordinate]{};
\path (7,2) node[name=l2, shape=coordinate]{};
\path (7,3) node[name=l3, shape=coordinate]{};
\path (6,3) node[name=l4, shape=coordinate]{};
\path[draw, fill=gray!20] (l1)--(l2)--(l3)--(l4)--cycle;
 
\path (4,2) node[name=l1, shape=coordinate]{};
\path (5,2) node[name=l2, shape=coordinate]{};
\path (5,3) node[name=l3, shape=coordinate]{};
\path (4,3) node[name=l4, shape=coordinate]{};
\path[draw, fill=gray!20] (l1)--(l2)--(l3)--(l4)--cycle;

\path (2,2) node[name=l1, shape=coordinate]{};
\path (3,2) node[name=l2, shape=coordinate]{};
\path (3,3) node[name=l3, shape=coordinate]{};
\path (2,3) node[name=l4, shape=coordinate]{};
\path[draw, fill=gray!20] (l1)--(l2)--(l3)--(l4)--cycle;

\path (0,0) node[name=l1, shape=coordinate]{};
\path (1,0) node[name=l2, shape=coordinate]{};
\path (1,1) node[name=l3, shape=coordinate]{};
\path (0,1) node[name=l4, shape=coordinate]{};
\path[draw, fill=gray!20] (l1)--(l2)--(l3)--(l4)--cycle;
\path ($1/2*(l1)+1/2*(l3)$) node[]{$\tiny y^\delta$};

\path (2,0) node[name=l1, shape=coordinate]{};
\path (3,0) node[name=l2, shape=coordinate]{};
\path (3,1) node[name=l3, shape=coordinate]{};
\path (2,1) node[name=l4, shape=coordinate]{};
\path[draw, fill=gray!20] (l1)--(l2)--(l3)--(l4)--cycle;

\path[draw][line width=2pt] (3,8)--(6,8)--(6,12);
 
\end{tikzpicture}\caption{A path on the set ${\color{gray}\blacklozenge}_0^\delta$ from the square $y^\delta$ to the square adjacent to $\wvpugl{k}$.}
\label{gamma}\end{center}
\end{figure}
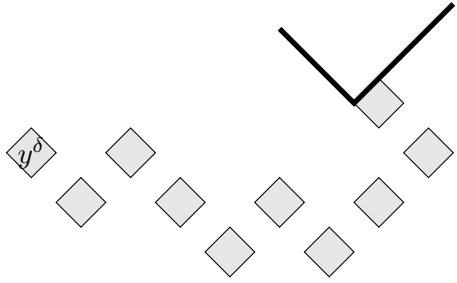

\begin{lemma}\label{omega > c}
Let $x^\delta$ be a black square in the middle of one of the arcs of the set $\partial B_{r}^{\delta}(\wvpugl{k})\cap\bb{0}^\delta$. Let $y^\delta\in \bb{0}^\delta$ be a black square on the boundary of $B_{r/2}^\delta(\wvpugl{k})$. Let $\gamma^\delta$ be a path on the set $\bb{0}^\delta$ starting in $y^\delta$ and ending at the black square of $\bb{0}^\delta$ adjacent to $\wvpugl{k}$ (see Fig. \ref{gamma}). Let $\operatorname{hm}^\delta(x^\delta,\gamma^\delta)$ be the harmonic measure on $B_{2r}^{\delta}(\wvpugl{k})\cap\bb{0}^\delta\smallsetminus\gamma^\delta$. Then there exists a constant $\widetilde{c} > 0$ that does not depend on $\delta$ such that for all $y^\delta\in\bb{0}^\delta\cap\partial B_{r/2}^{\delta}$, one has 
\[
\operatorname{hm}^\delta(x^\delta,\gamma^\delta)\geq \widetilde{c} = \widetilde{c}\,(\om) > 0.
\]
\end{lemma}

For a more general statement see~\cite[Lemma $3.14$]{C+S}.

\begin{proof}
Let us consider two gray discrete domains of width $\frac{r}{l}$, where $l$ is a large enough positive number, see Fig.~\ref{tunnel}.  Let these domains contain $x^\delta$ and cross the boundary of~$\om^\delta$. The probability that a random walk on a square lattice with mesh size $2\delta$ travels all the way from $x^\delta$ to the boundary of $\om^\delta$ inside the gray domain is uniformly bounded away from zero~\cite[Fig. 3B]{C+S}, for each of the two domains. 

Note that the path $\gamma^\delta$ necessarily intersects at least one of the gray domains. 
The probability of the event that a random walk travels all the way from $x^\delta$ to the boundary of $\om^\delta$ inside this gray domain is less then $\operatorname{hm}^\delta(x^\delta,\gamma^\delta)$.  
\end{proof}

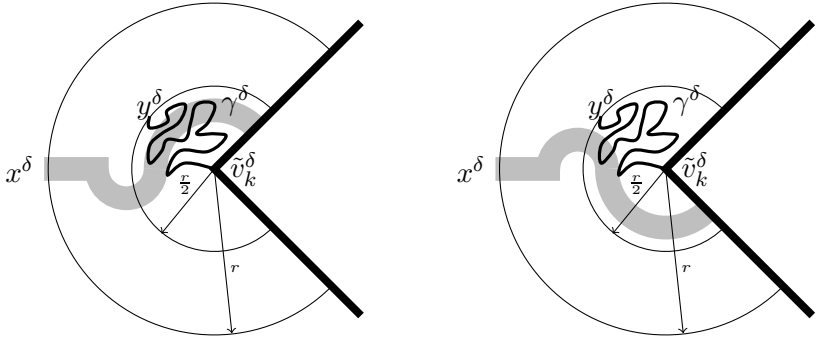
\begin{figure}
\begin{center}
\begin{tikzpicture}[x={(0.78cm,0cm)}, y={(0cm,0.78cm)}]
\begin{scope}
\draw[line width=9pt, gray!50] (-1,0) arc (180:45:1);
\draw[line width=9pt, gray!50] (-2.9,0)--(-1.8,0);
\draw[line width=9pt, gray!50] (-1,0) arc (0:-180:0.5);

\draw[line width=3pt] (2.5,-2.5) -- (0,0) node (v1) {} -- (2.5,2.5);
\draw (1,1) arc (45:315:1.41);
\draw (2,2) arc (45:315:2.83);
\draw[line width=1.2pt]  plot[smooth, tension=.7] coordinates {(-1.1,0.9) (-1.1,0.7) (-0.8,0.8) (-0.8,1) (-0.5,1.1) (-0.6,0.8) (-0.9,0.6) (-1.1,0.4) (-1.1,0.1) (-0.9,0.3) (-0.8,0.5) (-0.5,0.5) (-0.4,0.8) (-0.3,1.1) (0,1.1) (-0.1,0.8) (-0.3,0.6) (0,0.6)  (0.2,0.55) (-0.1,0.3) (-0.6,0.3) (-0.8,-0.1)(-0.5,0.1) (v1)};

\draw[->] (0,0) -- (-0.9,-1.1);
\draw[->] (0,0) -- (0.3,-2.8);
\path (-0.5,-0.6) node[above,font=\tiny]{$\frac{r}{2}$};
\path (0.1,-1.7) node[right,font=\tiny]{$r$};
\path (-0.05,1.17) node[right,]{$\gamma^\delta$};
\path (-1.5,1.1) node[right,]{$y^\delta$};
\path (-2.9,0) node[left,]{$x^\delta$};
\path (0.1,0) node[right,]{$\wvpugl{k}$};
\end{scope}

\begin{scope}[xshift=6cm]
\draw[line width=9pt, gray!50] (-1,0) arc (180:315:1);
\draw[line width=9pt, gray!50] (-2.9,0)--(-1.8,0);
\draw[line width=9pt, gray!50] (-1,0) arc (0:180:0.5);

\draw[line width=3pt] (2.5,-2.5) -- (0,0) node (v1) {} -- (2.5,2.5);
\draw (1,1) arc (45:315:1.41);
\draw (2,2) arc (45:315:2.83);
\draw[line width=1.2pt]  plot[smooth, tension=.7] coordinates {(-1.1,0.9) (-1.1,0.7) (-0.8,0.8) (-0.8,1) (-0.5,1.1) (-0.6,0.8) (-0.9,0.6) (-1.1,0.4) (-1.1,0.1) (-0.9,0.3) (-0.8,0.5) (-0.5,0.5) (-0.4,0.8) (-0.3,1.1) (0,1.1) (-0.1,0.8) (-0.3,0.6) (0,0.6)  (0.2,0.55) (-0.1,0.3) (-0.6,0.3) (-0.8,-0.1)(-0.5,0.1) (v1)};

\draw[->] (0,0) -- (-0.9,-1.1);
\draw[->] (0,0) -- (0.3,-2.8);
\path (-0.5,-0.6) node[above,font=\tiny]{$\frac{r}{2}$};
\path (0.1,-1.7) node[right,font=\tiny]{$r$};
\path (-0.05,1.17) node[right,]{$\gamma^\delta$};
\path (-1.5,1.1) node[right,]{$y^\delta$};
\path (-2.9,0) node[left,]{$x^\delta$};
\path (0.1,0) node[right,]{$\wvpugl{k}$};
\end{scope}

\end{tikzpicture}\caption{The probability that a random walk on a square lattice with mesh size $2\delta$ travels all the way from $x^\delta$ to the boundary of $\om^\delta$ inside the gray domain is bounded away from zero uniformly in $\delta$.
}\label{tunnel}\end{center}
\end{figure}

\begin{lemma}\label{ravnogr}
Let $\mdel_t(r)=\max\limits_{u^\delta\in\om_{r,t}^\delta}|\F^\delta(u^\delta)|$, where 
\[
\om_{r,t}^\delta=\om^\delta\smallsetminus \left(\bigcup\limits_{k=1}^{n-1} B_r^\delta(\wvpugl{k})\cup B_t^\delta(\www^\delta)\right).
\] 
Then for some fixed $t > 0$ small enough and for any sufficiently small fixed $r > 0$, as $\delta\to 0$ we have 
\[
\mdel_t\left(\frac{r}{2}\right)\leq \frac{4}{\widetilde{c}}\cdot\mdel_t(r),
\] 
where $\widetilde{c}$ is the absolute constant from Lemma~$\ref{omega > c}$.
\end{lemma}

\begin{proof} Note that it is enough to prove that 
\[
\max\limits_{u^\delta\in\om^\delta\cap\partial B_{r/2}^\delta(\wvpugl{k})}|\re\F^\delta(u^\delta)|\leq \frac{2}{\widetilde{c}}\cdot\mdel_t(r),
\] 
for all $k\in\{1,\ldots,n-1\}$, since similarly the same inequality holds for $\im\F^\delta$.

Let $y^\delta$ be the square in $\om^\delta\cap\partial B_{r/2}^\delta(\wvpugl{k})$ such that 
\[
|\re\F^\delta(y^\delta)|=\max\limits_{u^\delta\in\om^\delta\cap\partial B_{r/2}^\delta(\wvpugl{k})}|\re\F^\delta(u^\delta)|.
\] 
Without loss of generality we may assume that $\re\F^\delta(y^\delta)>0$. Note that $\re\F^\delta$ is a discrete harmonic function, and hence there exists a path $\gamma^\delta$ on the set  $\bb{0}^\delta$ from $y^\delta$ to the boundary of the domain $\om^\delta\cap B_r^\delta(\wvpugl{k})$ or to the square adjacent to $\wvpugl{k}$ along which the absolute value of the function $\re\F^\delta$ increases, since discrete harmonic functions satisfy the maximum principle. If the path $\gamma^\delta$ ends on the boundary of the domain $\om^\delta\cap B_r^\delta(\wvpugl{k})$, then
\[
\max\limits_{u^\delta\in\om^\delta\cap\partial B_{r/2}^\delta(\wvpugl{k})}|\re\F^\delta(u^\delta)|\leq \max\limits_{u^\delta\in\om^\delta\cap\partial B_{r}^\delta(\wvpugl{k})}|\re\F^\delta(u^\delta)|.
\]
Assume that $\gamma^\delta$ ends at the square adjacent to $\wvpugl{k}$. Let $\operatorname{hm}^\delta(\cdot,\gamma^\delta)$ be the harmonic measure in the domain $B_{2r}^{\delta}(\wvpugl{k})\cap\bb{0}^\delta\smallsetminus\gamma^\delta$. Due to Lemma \ref{omega > c} there exists a black square $x^\delta\in \bb{0}^\delta$ on the boundary of $B_{r}^{\delta}(\wvpugl{k})$ such that $\operatorname{hm}^\delta(x^\delta,\gamma^\delta)\geq \widetilde{c} > 0$. Note that
\[
\mdel_t(r)\geq\re\F^\delta(x^\delta)\geq\re\F^\delta (y^\delta)\cdot\operatorname{hm}^\delta(x^\delta,\gamma^\delta)-\mdel_t(r)\cdot(1-\operatorname{hm}^\delta(x^\delta,\gamma^\delta)).
\]
Hence,
\[
2\mdel_t(r)\geq\operatorname{hm}^\delta(x^\delta,\gamma^\delta)\cdot\re\F^\delta(y)\geq
\widetilde{c}\cdot\re\F^\delta(y^\delta).
\]
To complete the proof, recall that we assumed $\re\F^\delta(y^\delta)=\max\nolimits_{u^\delta\in\om^\delta\cap\partial B_{r/2}^\delta(\wvpugl{k})}|\re\F^\delta(u^\delta)|.$
\end{proof}

Let $\F^\delta_{\mathbb{C}, \www^\delta}$ be the unique discrete holomorphic function on the whole plane $\mathbb{C}^\delta \smallsetminus \{\www^\delta\}$ tending to zero at infinity and such that $[\bar{\partial}^\delta\F^\delta_{\mathbb{C}, \www^\delta}](\www^\delta) = \frac{\lambda}{\delta^2}$, see~\cite[Theorem 2.21]{C+S}. 
Note that $\re\F^\delta_{\mathbb{C}, \www^\delta}$ and $\im\F^\delta_{\mathbb{C}, \www^\delta}$ are discrete harmonic everywhere except two squares adjacent to $\www^\delta$. It is well known that $\F^\delta_{\mathbb{C}(z), \www^\delta}$ is asymptotically equal $\frac{1}{2\pi}\cdot\frac{\lambda}{z-\www}$ as~$\delta\downarrow 0$.  We need to introduce a similar function $\F^\delta_{\mathbb{H}, \www^\delta}$ on a half-plane $\mathbb{H}^\delta$, where $\partial\mathbb{H}^\delta$ goes to the direction $\lambda$ and $\www^\delta\in \partial_{\mathrm {int}}\mathbb{H}^\delta\cap\ww{0}^\delta$.  
The imaginary part of $\F^\delta_{\mathbb{H},\www^\delta}$ equals zero on the boundary, and $[\bar{\partial}^\delta\F^\delta_{\mathbb{H}, \www^\delta}](\www^\delta) = \frac{\lambda}{\delta^2}.$ There is the unique discrete holomorphic function with these two properties that tends to zero at infinity. 

Let us consider the sum $\F^\delta_{\mathbb{C}, \www^\delta} + \F^\delta_{\mathbb{C}, \www^\delta+2\bar\lambda\delta},$ where by $\www^\delta+2\bar\lambda\delta$ we denote a white square at distance $\delta$ from the square $\www^\delta$ that does not belong to $\mathbb{H}^\delta.$ This sum tends to zero at infinity, since both $\F^\delta_{\mathbb{C}, \www^\delta}$ and $\F^\delta_{\mathbb{C}, \www^\delta+2\bar\lambda\delta}$ tend to zero at the infinity. Note that $\F^\delta_{\mathbb{C}, \www^\delta+2\bar\lambda\delta}$ is discrete holomorphic on $\mathbb{H}^\delta$, therefore $\F^\delta_{\mathbb{C}, \www^\delta} + \F^\delta_{\mathbb{C}, \www^\delta+2\bar\lambda\delta}$ is holomorphic on $\mathbb{H}^\delta \smallsetminus \{\www^\delta\}$ and $[\bar{\partial}^\delta(\F^\delta_{\mathbb{C}, \www^\delta} + \F^\delta_{\mathbb{C}, \www^\delta+2\bar\lambda\delta})](\www^\delta) = \frac{\lambda}{\delta^2}.$ Finally, note that 
\[
\im\F^\delta_{\mathbb{C}, \www^\delta}(u)=G^\delta(u,\www^\delta+\bar\lambda\delta) - G^\delta(u,\www^\delta-\bar\lambda\delta),
\] 
where 
$G^\delta(u,u')$ is the classical Green's function on $\mathbb{C}^\delta\cap\bb{1}^\delta$ satisfies $\Delta^\delta G^\delta(u,u')=\mathbb{1}_{u=u'}\cdot\frac{1}{2\delta^3}$. The Green's function is symmetric, therefore $\im[\F^\delta_{\mathbb{C}, \www^\delta} + \F^\delta_{\mathbb{C}, \www^\delta+2\bar\lambda\delta}]$ vanishes on $\partial\mathbb{H}^\delta.$
As a consequence we have
$\F^\delta_\mathbb{H}(u) = \F^\delta_{\mathbb{C}, \www^\delta} + \F^\delta_{\mathbb{C}, \www^\delta+2\bar\lambda\delta}.$

\begin{cor}\label{cor2}
Let 
\[
\mdel_*(r)=\max\limits_{u\in\om_{r,0}^\delta}|\F^\delta(u)-\F^\delta_\mathbb{H}(u)|.
\]
Then, for all sufficiently small $\delta$, one has
\[
\mdel_*\left(\frac{r}{2}\right)\leq \frac{4}{\widetilde{c}}\cdot\mdel_*(r) + C_*,
\]
where $C_*$ is an absolute constant and $\widetilde{c}$ is the constant from Lemma~$\ref{omega > c}$.
\end{cor}

\begin{proof}
Note that $\F^\delta_\mathbb{H}$ is uniformly bounded away from $\www^\delta$ and vanishes on $\partial\mathbb{H}^\delta$, hence $\F^\delta-\F^\delta_\mathbb{H}$ is uniformly bounded on $\dom^\delta.$ Moreover, function  $\F^\delta-\F^\delta_\mathbb{H}$ is discrete holomorphic on $\om^\delta$, in particular it is discrete holomorphic on $B_t^\delta(\www^\delta)\cap\om^\delta$ and vanishes on $\dom^\delta\cap B_t^\delta(\www^\delta)$. Therefore the statement of Lemma~\ref{ravnogr} is valid for $\F^\delta-\F^\delta_\mathbb{H}$ and $t=0$.
\end{proof}

We are now in the position to prove the convergence of $\F^\delta$. Note that $\F^\delta$ can be thought of as defined in polygonal
representation of $\om^\delta$ by some standard continuation procedure, linear on edges
and multilinear inside faces. Then we have the following

\begin{Th}\label{convF}
Let $\om^\delta$ be a sequence of discrete $2k$-black-piecewise Temperleyan domains of mesh size $\delta$ approximating a continuous domain $\om$. Suppose that each $\om^\delta$ admits a domino tiling. Let the sets of white corner squares 
$\{\wvupugl{k}\}^{n+1}_{k=1}$ and 
$\{\wvpugl{k}\}^{n-1}_{k=1}$ 
approximate the sets of boundary points $\{\nwvupugl{k}\}^{n+1}_{k=1}$ 
and $\{\nwvpugl{k}\}^{n-1}_{k=1}$ 
correspondingly, and let $\www^{\delta}$ approximate a boundary point $\www$, which lies on a straight segment of the boundary of $\om$. Then $\F^\delta$ converges uniformly on compact subsets of  $\om$ to a continuous holomorphic function $f_\om$, where $f_\om$ is defined as in Proposition~$\ref{1-4}$.
\end{Th}

In the following proof we use the idea described in ~\cite{CHI} (proof of Theorem $2.16$).
\begin{proof}

{\it First case}: suppose that for each fixed positive $r$ the function $\mdel_*(r)$ remains bounded, as $\delta\to~0.$ 
Corollary~\ref{cor2} implies that discrete holomorphic functions $\F^\delta-\F^\delta_\mathbb{H}$ are uniformly bounded, and therefore equicontinuous due to Harnack principle on compact subsets of $\om$. Thus, due to the Arzelà–Ascoli theorem, the family $\F^\delta~-~\F^\delta_\mathbb{H}$ is precompact and hence converges along a subsequence to some holomorphic function~$\widetilde{f}$ uniformly on compact subsets of $\om$. 
Note that $\F^\delta_\mathbb{H}\rightrightarrows f_\mathbb{H}=\frac{1}{\pi}\cdot\frac{\lambda}{z-\www}$ as $\delta\to 0$, uniformly on compacts. 
Let $f_\om:=\widetilde{f} - f_\mathbb{H}$, then
$\F^\delta\rightrightarrows f_\om$, i.e. $\re\F^\delta\rightrightarrows \re f_\om$ and $\im\F^\delta\rightrightarrows \im f_\om$, uniformly on compacts.
Since a discrete solution of Dirichlet problem converges to its continuous counterpart up to the boundary~\cite[Section 3.3]{C+S}, the boundary conditions for the functions $\F^\delta$ yield the same boundary conditions for their limit. Thus, the function $f_\om$  solves the boundary value problem described in Proposition~\ref{1-4}, therefore it is determined uniquely. 
This implies that all convergent subsequences of the family $\{\F^{\delta}\}$ have the same limit and thus the whole family converges to $f_\om$.

{\it Second case}: suppose that $\mdel_*(r)$ tends to infinity along a subsequence as $\delta\to 0$ for some $r>0$. Let us show that this is impossible.
 Consider a discrete holomorphic function $\Ftilda^\delta_*:=\frac{\F^\delta-\F^\delta_\mathbb{H}}{\mdel_*(r)}$. 
Using the same arguments as above, we can show that the family $\Ftilda^\delta_*$ converges to some holomorphic function $f_*$.
Note that the limit is bounded near $\www$, since  $\F^\delta-\F^\delta_\mathbb{H}$ is discrete holomorphic and bounded near~$\www^\delta$. Also, note that $\frac{\F^\delta_\mathbb{H}}{\mdel_*(r)}$ tends to zero away from~$\www$.
Therefore, as in the previous case, the limit satisfies all boundary conditions described in Proposition~$\ref{1-4}$ except the first one: the behaviour near the point~$\www$. 
The only function satisfying these properties is zero.

Suppose that there exists a sequence of squares~$u_{\operatorname{inner}}^\delta$ converging to~$u_{\operatorname{inner}}\in\om$
 such that 
\begin{equation}\label{op}
\re\Ftilda^\delta_*(u_{\operatorname{inner}}^\delta)> \operatorname{const}_{\om}>0.
\end{equation}
Then we have $f_*(u_{\operatorname{inner}})>0$, which contradicts the fact that $f_*$ vanishes on~$\om$, and therefore the second case is impossible.

To complete the proof let us show the existence of the sequence~$\{u_{\operatorname{inner}}^\delta\}$.
Let $u_{\operatorname{max}}^\delta$ be chosen so that 
$ 1=\sup\limits_{u^\delta\in\om_{r,0}^\delta}|\Ftilda^\delta_*(u^\delta)|=|\Ftilda^\delta_*(u_{\operatorname{max}}^\delta)|.$
Assume that $u_{\operatorname{max}}^\delta\in\bb{0}^\delta$, i.e. $|\Ftilda^\delta_*(u_{\operatorname{max}}^\delta)|=|\re\Ftilda^\delta_*(u_{\operatorname{max}}^\delta)|$. Without loss of generality we may assume that $\re\Ftilda^\delta_*(u_{\operatorname{max}}^\delta)>0$.
Let $u_{\operatorname{max}}^\delta \to u_{\operatorname{max}}\in \overline{\om}_r$ as $\delta\to 0$, where $\om_{r}=\om\smallsetminus \left(\bigcup\limits_{k=1}^{n-1} B_r(\nwvpugl{k})
\right).$ The discrete maximum principle implies that  $u_{\operatorname{max}}\in\bigcup\limits_{k=1}^{n-1} \partial B_r(\nwvpugl{k})$. 
Note that $\re\Ftilda^\delta_*$ is a discrete harmonic function, and hence there exists a path $\gamma^\delta$ on the set  $\bb{0}^\delta$ from $u_{\operatorname{max}}^\delta$ to the boundary of the domain $\om^\delta$ or to the square adjacent to $\wvpugl{k}$ along which the absolute value of the function $\re\Ftilda^\delta_*$ increases.
The boundary conditions together with the fact that the limit function vanishes imply that $\gamma^\delta$ goes along a subarc $N_k^\delta\subset \partial\om^\delta$ where $\re\F^{\delta}$ has Neumann boundary condition and ends at the square adjacent to $\wvpugl{k}$. 

Assume that $B_r(\nwvpugl{k})\cap \om$ is connected, the other case is treated similarly. Denote by $U^{\delta}$ the discrete subdomain of $B^{\delta}_r(\wvpugl{k})\cap \om^{\delta}$ that is bounded by the subarc of $\partial\om^{\delta}$ where $\re\F^{\delta}$ has Dirichlet boundary condition, the path $\gamma^{\delta}\cap \om^{\delta}$ and the arc $\partial B^{\delta}_r(\wvpugl{k})\cap \om^{\delta}$. Note that $U^{\delta}$ converges to $B_r(\nwvpugl{k})\cap \om$.

The absolute value of $\re\Ftilda^\delta_*$ is bounded by $\epsilon_\delta$ away from the pieces of the boundary of $\om^\delta$ where $\re\F^{\delta}$ has Neumann boundary conditions. Note that the function~$\re\Ftilda^\delta_*$ is semi-bounded in a vicinity of the point $\wvpugl{k}$, therefore  near the boundary $\re\Ftilda^\delta_*> -c$, where $c>0$ is a constant. 
Let $u_{\operatorname{inner}}^\delta\in U^\delta$ be a black square in the middle of one of the arcs of the set $\partial B_{r/2}^{\delta}(\wvpugl{k})\cap\bb{0}^\delta$.
Then 
\[
\re\Ftilda^\delta_*(u^\delta)\geq -{\epsilon_\delta}\cdot 1+(-c)\cdot\operatorname{hm}_{U^\delta}(u^\delta, ( \operatorname{\widetilde{\epsilon}_\delta-vicinity\,of} N_k^\delta) \cap(\partial B^{\delta}_r(\wvpugl{k})\cap \om^{\delta})) + 1\cdot\operatorname{hm}_{U^\delta}(u^\delta, \gamma^\delta\cap\partial U^\delta).
\]
Due to Lemma~\ref{omega > c} we have $\operatorname{hm}_{U^\delta}(u_{\operatorname{inner}}^\delta, \gamma^\delta\cap\partial U^\delta)>\operatorname{const}(U)>0$. Note that ${\epsilon}_\delta$ tends to zero as $\delta\to 0$. Also, $\operatorname{hm}_{U^\delta}(u_{\operatorname{inner}}^\delta, ( \operatorname{\widetilde{\epsilon}_\delta-vicinity\,of} N_k^\delta) \cap(\partial B^{\delta}_r(\wvpugl{k})\cap \om^{\delta}))$ tends to zero as $\delta\to 0$. 
Hence we construct a sequence of squares~$u_{\operatorname{inner}}^\delta$ converging to~$u_{\operatorname{inner}}\in\om$
 such that~(\ref{op}) holds.
\end{proof}

\begin{remark}\label{remrem}
Let in the setup of Theorem~\ref{convF} the squares $\www^{\delta}$ approximate an inner point $\www$ of the domain $\om$, instead of a boundary one. Then $\F^\delta$ converges uniformly on compact subsets of  $\om\setminus \www$ to a continuous holomorphic function $f_\om$, where $f_\om$ is defined as in Remark~$\ref{v_0_inner}$.
\end{remark}

\section{Single dimer model and the Gaussian Free Field}\label{6}
In~\cite{KGff} Kenyon proved that the scaling limit of the height function in the dimer model on Temperleyan domains is the Gaussian Free Field. Our goal in this section is to prove Corollary~\ref{main-cor2}, i.e. to show that the same scaling limit appears for approximations by a more general class of discrete domains which we call \emph{black-piecewise Temperleyan} domains. Also, in Appendix we will show that the same holds for \emph{isoradial black-piecewise Temperleyan} graphs.
 
\subsection{Boundary conditions for the coupling function}
For a fixed $v'\in\ww{0}$, the function $C_\om(u,v')$ is discrete holomorphic as a function of $u$, with a simple pole at $v'$: 

$\begin{cases}
\Cm(\cdot, v')|_{\db}=0,\\
\Cm(\cdot, v')|_{\bb{0}} \in \mathbb{R}, \quad \Cm(\cdot, v')|_{\bb{1}} \in i\mathbb{R},\\
\bar{\partial}[C_\om(\cdot,v')](v)=0 ,\quad \forall v\in\ww{}, v\neq v'\\
\bar{\partial}[C_\om(\cdot,v')](v')=\frac{1}{4\bar{\lambda}}.
\end{cases}
$ 

Therefore in a $2n$-black-piecewise Temperleyan domain, for a fixed $v'\in\ww{}$, the boundary conditions of the coupling function $C_\om(u,v')$ as a function of $u$ change at all white corners, and there are $2n$ parts of the boundary with either $\re[\Cm(\cdot, v')]=0$ or $\im[\Cm(\cdot, v')]=0$.

In other words, a black-piecewise Temperleyan domain corresponds to mixed Dirichlet and Neumann boundary conditions for the discrete harmonic components of the coupling function.
Recall that in a Temperleyan domain we have a simple boundary conditions, namely, $\im[\Cm(u, v')]=0$ for all boundary squares $u$, in other words $\im[\Cm(u, v')]$ as a function of $u$ has Dirichlet boundary conditions.

\subsection{Asymptotic values of the coupling function}\label{sa_c_f}
Following~\cite{Kdom}, we define two functions $\ff{0} (z_1,z_2)$ and $\ff{1} (z_1,z_2)$. For a fixed $z_2$,
\begin{enumerate}
\item[$\rhd$] the function $\ff{0} (z_1, z_2)$ is analytic as a function of $z_1$, has a simple pole of residue $1/ \pi$ at $z_1=z_2$, and no other poles on $\overline{\om}$;
\item[$\rhd$] $\ff{0}(\cdot, z_2)$ is bounded in the vicinity of the points $\nwvupugl{s}$;
\item[$\rhd$] the function $\ff{0}(\cdot, z_2)$ is semi-bounded in the vicinity of the points $\nwvpugl{k}$;
\item[$\rhd$] on each segment into which points from the set $\{\nwvpugl{k}\}^{n-1}_{k=1}\cup\{\nwvupugl{s}\}^{n+1}_{s=1}$ split the boundary, we have either $\re[\ff{0}(\cdot, z_2)]=0$ or $\im[\ff{0}(\cdot,z_2)]=0$;
\item[$\rhd$] the boundary conditions of the function $\ff{0}(\cdot, z_2)$ change at all points $\nwvpugl{k},\, \nwvupugl{s}.$
\end{enumerate}
The function $\ff{1}(z_1, z_2)$ has the same definition, except for a difference in the boundary conditions: if on a segment between two points from the set $\{\nwvpugl{k}\}^{n-1}_{k=1}\cup\{\nwvupugl{s}\}^{n+1}_{s=1}$ we have $\re[\ff{0}(\cdot, z_2)]=0$ (or $\im[\ff{0}(\cdot, z_2)]=0$), then on that segment $\im[\ff{1}(\cdot, z_2)]=0$ (or $\re[\ff{1}(\cdot, z_2)]=0$).
The existence and uniqueness of such functions can be shown using the technique described in Section 3, see Remark~\ref{v_0_inner}. In particular, we can write these functions in the following way
\[
\ff{0} (z, w)=\prod^{n+1}_{k=1}(z-\nwvupugl{k})^{\frac12}
\cdot \prod^{n-1}_{k=1}(z-\nwvpugl{k})^{-\frac12}
\cdot\left(\frac{s(w)}{z-w} + \frac{\overline{s(w)}}{z-\overline w}\right),
\]

\[
\ff{1} (z, w)=\prod^{n+1}_{k=1}(z-\nwvupugl{k})^{\frac12}
\cdot \prod^{n-1}_{k=1}(z-\nwvpugl{k})^{-\frac12}
\cdot\left(\frac{s(w)}{z-w} - \frac{\overline{s(w)}}{z-\overline w}\right),
\]
where 
\begin{equation}\label{s(w)}
s(w)=\prod^{n+1}_{k=1}(w-\nwvupugl{k})^{-\frac12}
\cdot \prod^{n-1}_{k=1}(w-\nwvpugl{k})^{\frac12}.
\end{equation}

\begin{Th}\label{main-th2_1}
 Let $\om$ be a bounded, simply connected domain in $\mathbb{C}$ with $k$ marked points. Assume that a sequence of discrete $k$-black-piecewise Temperleyan domains $\om^\delta$ on a grid with mesh size $\delta$ approximates the domain $\om$, and each domain $\om^\delta$ has at least one domino tiling.
Let a sequence of white squares $v^\delta$ approximates a point $v\in \om$. Then the coupling function $\frac{1}{\delta}\Cmd(u,v)$ satisfies the following asymptotics: \\
for $v^\delta \in \ww{0}$
\[
\frac{1}{\delta}\Cmd(u,v^\delta) - \frac{2}{\lambda}\cdot\F^\delta_{\mathbb{C}, v^\delta}(u)
= \ff{0}(u,v)-\frac{1}{\pi(u-v)}+o(1);
\]
if $v^\delta \in \ww{1}$, then
\[
\frac{1}{\delta}\Cmd(u,v^\delta) - \frac{2}{\lambda}\cdot\F^\delta_{\mathbb{C}, v^\delta}(u)
= \ff{1}(u,v)-\frac{1}{\pi(u-v)}+o(1),
\]
where $\F^\delta_{\mathbb{C}, v^\delta}(u)$ is defined in Section~\ref{convergence}.
\end{Th}
\begin{proof} Recall that $\F^\delta_{\mathbb{C}(z), \www^\delta}$ is asymptotically equal $\frac{1}{2\pi}\cdot\frac{\lambda}{z-\www}$ as~$\delta\downarrow 0$.
Now, to obtain the result 
one can use 
the techniques described in Section~\ref{convergence}. More precisely, the first asymptotic can be obtained exactly from the proof of Theorem~\ref{convF}, see Remark~\ref{remrem}. The second one can be obtained similarly.
\end{proof}

\subsection{ Sketch of the proof of Corollary~\ref{main-cor2}}
This section contains the sketch of the proof of Corollary~\ref{main-cor2}. In \cite{KGff} Kenyon proved convergence of the height function on Temperleyan domains to the Gaussian free field. To obtain the same result for black-piecewise Temperleyan domains it is enough to show  that the limits of moments of height function in Temperleyan case and black-piecewise Temperleyan case are the same. 
Essentially the novel part of the argument is in~(\ref{newf+-}), then Lemma~\ref{lemma_s(w)} implies Corollary~\ref{eqlimits}, and the rest of the argument is exactly as in~\cite[Theorem~1.1]{KGff}.

Similarly to~\cite{Kdom} one can obtain the following result for black-piecewise Temperleyan approximations. Let $\ff{+}(z, w)=\ff{0}(z, w)+\ff{1}(z, w)$ and $\ff{-}(z, w)=\ff{0}(z, w)-\ff{1}(z, w)$. 
\begin{prop}
Let $\gamma_1, \ldots, \gamma_m$ be a collection of pairwise disjoint paths running from the boundary of $\Omega$ to $z_1, \ldots, z_m$ respectively. Let $h(z_i)$ denote the height function at a point in black-piecewise Temperleyan domain $\Omega^\delta$ lying within $O(\delta)$ of $z_i$. Then 
\begin{align*}
\lim_{\delta\to o}\mathbb{E}&[(h(z_1)-\mathbb{E}[h(z_1)])\cdot\ldots\cdot(h(z_m)-\mathbb{E}[h(z_m)])]~=~\\ 
&\sum_{\epsilon_1, \ldots, \epsilon_m \in \{-1,1\}} \epsilon_1\cdots\epsilon_m\int_{\gamma_1}\cdots\int_{\gamma_m}\det_{i,j\in [1,m]}(F_{\epsilon_i,\epsilon_j}(z_i, z_j))\,dz_1^{(\epsilon_1)}\cdots dz_m^{(\epsilon_m)},
\end{align*}
where $dz_j^{(1)}=dz_j$ and $dz_j^{(-1)}=d\overline{z_j}$, and
\[
F_{\epsilon_i, \epsilon_j}(z_i,z_j)=\begin{cases} 
0 &i=j\\
f_{+}(z_i,z_j) \quad &(\epsilon_i,\epsilon_j)=(1,1)\\
f_{-}(z_i,z_j) \quad &(\epsilon_i,\epsilon_j)=(-1,1)\\
\overline{f_{-}(z_i,z_j)} \quad &(\epsilon_i,\epsilon_j)=(1,-1)\\
\overline{f_{+}(z_i,z_j)} \quad &(\epsilon_i,\epsilon_j)=(-1,-1).
\end{cases}
\]
\end{prop}

\begin{proof}
See the proof of~\cite[Proposition 20]{Kdom}.
\end{proof}

Recall that in the Temperleyan case~\cite{Kdom} one has $\ff{+}(z,w)=\frac{2}{z-w}$ and $\ff{-}(z,w)=\frac{2}{z-\overline w}$. In the black-piecewise Temperleyan case we have
\begin{equation}\label{newf+-}
\begin{cases}
\ff{-}(z, w)=\frac{2}{z-\overline w}\cdot\frac{\overline {s(w)}}{{s(z)}}\\
\ff{+}(z, w)=\frac{2}{z-w}\cdot\frac{s(w)}{s(z)}
\end{cases}
\end{equation}
where the function $s(w)$ is defined by~(\ref{s(w)}).

One can easily check that the following lemma holds:
\begin{lemma}\label{lemma_s(w)} Let $\epsilon_1, \ldots, \epsilon_m \in \{-1,1\}$. Let us define function $S_{\epsilon_i, \epsilon_j}(z,w)$ as follows:
\[
S_{\epsilon_i, \epsilon_j}(z,w)=\begin{cases} 
0 &i=j\\
s(w)/s(z) \quad &(\epsilon_i,\epsilon_j)=(1,1)\\
\overline {s(w)}/s(z) \quad &(\epsilon_i,\epsilon_j)=(-1,1)\\
s(w)/\overline{s(z)} \quad &(\epsilon_i,\epsilon_j)=(1,-1)\\
\overline{s(w)}/\overline{s(z)} \quad &(\epsilon_i,\epsilon_j)=(-1,-1).
\end{cases}
\]
Then 
\[
S_{\epsilon_{\alpha(1)}, \epsilon_1}(z_{\alpha(1)},z_1) \cdot \ldots \cdot S_{\epsilon_{\alpha(m)}, \epsilon_m}(z_{\alpha(m)},z_m)=\begin{cases}
1 \quad  \alpha(i)\neq i \quad \forall i\in\{1, 2, \ldots, m\}\\
0 \quad otherwise,
\end{cases}
\] 
where $\alpha$ is a permutation of the set $\{1, 2, \ldots, m\}$.
\end{lemma}

\begin{cor}\label{eqlimits}
The limits of moments of the height function in the Temperleyan case and the black-piecewise Temperleyan case are the same.
\end{cor}

In \cite{KGff} Kenyon showed that:
\begin{prop}
Let $\Omega$ be a Jordan domain with smooth boundary. Let $z_1, \ldots, z_m$ (with $m$ even) be distinct points of $\om$. Let $\Omega^\delta$ be a Temperleyan approximation of $\Omega$ and 
$h_{\Omega^\delta}$ be the height function of a uniform domino tiling in the domain $\Omega^\delta$. Then 
\begin{align*}
\lim_{\delta\to o}\mathbb{E}[(h_{\Omega^\delta}(z_1)-\mathbb{E}[h_{\Omega^\delta}(z_1)])\cdot\ldots\cdot&(h_{\Omega^\delta}(z_m)-\mathbb{E}[h_{\Omega^\delta}(z_m)])]~=~\\ 
\left(-\frac{16}{\pi}\right)^{m/2}\sum_{pairings \,\alpha} &g_D(z_{\alpha(1)},z_{\alpha(2)})\cdot\ldots\cdot g_D(z_{\alpha(m-1)},z_{\alpha(m)}),
\end{align*}
where $g_D$ is the Green function with Dirichlet boundary conditions on $\Omega$.
\end{prop}

\begin{proof}
See the proof of ~\cite[Proposition 3.2]{KGff}.
\end{proof}

By Corollary~\ref{eqlimits} this proposition holds for black-piecewise Temperleyan domains as well. And the following lemma completes the proof of Corollary~\ref{main-cor2}.
 
 \begin{lemma}[\cite{proba}]
A sequence of multidimensional random variables whose moments converge to the moments of a Gaussian, converges itself to a Gaussian.
\end{lemma}

\appendix
\setcounter{secnumdepth}{0}
\section{Appendix. Generalization to isoradial graphs}\label{B}

\renewcommand{\thesection}{A}
\setcounter{equation}{0}
\renewcommand{\theequation}{A.\arabic{equation}}

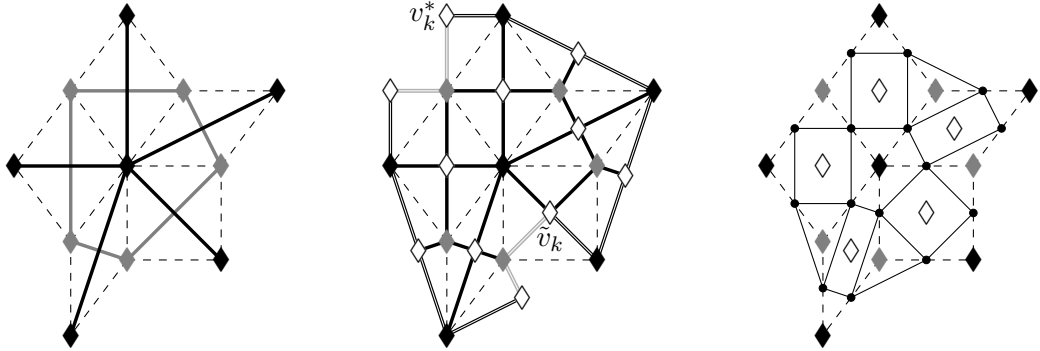
\begin{figure}
\begin{center}
\begin{tikzpicture}[x={(0.5cm,0cm)}, y={(0cm,0.5cm)}]
\begin{scope}
\path (0,0) node[name=A, shape=coordinate]{};
\path (0,-4) node[name=B', shape=coordinate]{};
\path (0,4) node[name=B, shape=coordinate]{};
\path (3,0) node[name=C', shape=coordinate]{};
\path (-3,0) node[name=C, shape=coordinate]{};
\path ($1/2*(B)+1/2*(C)$) node[name=x, shape=coordinate]{};
\path (2.5,0) node[name=y, shape=coordinate]{};
\path (0,-2.5) node[name=z, shape=coordinate]{};
\path ($1/2*(B)+1/2*(C')$) node[name=t, shape=coordinate]{};
\path ($(t)+(y)$) node[name=D, shape=coordinate]{};
\path ($(z)+(y)$) node[name=E, shape=coordinate]{};
\path ($1/2*(B')+1/2*(C)$) node[name=p, shape=coordinate]{};
\path ($(z)+(p)$) node[name=F, shape=coordinate]{};

\draw[gray, line width=1.3pt](x)--(t)--(y)--(z)--(p)--(x);

\draw[line width=1.3pt](A)--(B);
\draw[line width=1.3pt](A)--(C);
\draw[line width=1.3pt](A)--(D);
\draw[line width=1.3pt](A)--(E);
\draw[line width=1.3pt](A)--(F);

\draw[draw, dashed](x)--(A);
\draw[draw, dashed](x)--(B);
\draw[draw, dashed](x)--(C);

\draw[draw, dashed](t)--(A);
\draw[draw, dashed](t)--(B);
\draw[draw, dashed](t)--(D);

\draw[draw, dashed](y)--(A);
\draw[draw, dashed](y)--(E);
\draw[draw, dashed](y)--(D);

\draw[draw, dashed](z)--(A);
\draw[draw, dashed](z)--(E);
\draw[draw, dashed](z)--(F);

\draw[draw, dashed](p)--(A);
\draw[draw, dashed](p)--(C);
\draw[draw, dashed](p)--(F);

\path (A) node[]{${\blacklozenge}$};
\path (B) node[]{${\blacklozenge}$};
\path (C) node[]{${\blacklozenge}$};
\path (D) node[]{${\blacklozenge}$};
\path (E) node[]{${\blacklozenge}$};
\path (F) node[]{${\blacklozenge}$};

\path (x) node[]{${\color{gray}\blacklozenge}$};
\path (y) node[]{${\color{gray}\blacklozenge}$};
\path (z) node[]{${\color{gray}\blacklozenge}$};
\path (t) node[]{${\color{gray}\blacklozenge}$};
\path (p) node[]{${\color{gray}\blacklozenge}$};
\end{scope}

\begin{scope}[xshift=5cm]
\path (0,0) node[name=A, shape=coordinate]{};
\path (0,-4) node[name=B', shape=coordinate]{};
\path (0,4) node[name=B, shape=coordinate]{};
\path (3,0) node[name=C', shape=coordinate]{};
\path (-3,0) node[name=C, shape=coordinate]{};
\path ($1/2*(B)+1/2*(C)$) node[name=x, shape=coordinate]{};
\path (2.5,0) node[name=y, shape=coordinate]{};
\path (0,-2.5) node[name=z, shape=coordinate]{};
\path ($1/2*(B)+1/2*(C')$) node[name=t, shape=coordinate]{};
\path ($(t)+(y)$) node[name=D, shape=coordinate]{};
\path ($(z)+(y)$) node[name=E, shape=coordinate]{};
\path ($1/2*(B')+1/2*(C)$) node[name=p, shape=coordinate]{};
\path ($(z)+(p)$) node[name=F, shape=coordinate]{};

\draw[line width=1.3pt](A)--(B);
\draw[line width=1.3pt](A)--(C);
\draw[line width=1.3pt](A)--(D);
\draw[line width=1.3pt](A)--(E);
\draw[line width=1.3pt](A)--(F);

\draw[line width=1.3pt](B)--(D);
\draw[white,line width=0.3pt](B)--(D);
\draw[line width=1.3pt](E)--(D);
\draw[white,line width=0.3pt](E)--(D);
\draw[line width=1.3pt](C)--(F);
\draw[white,line width=0.3pt](C)--(F);
\draw[line width=1.3pt]($1/2*(C)+1/2*(F)$)--(p);
\draw[line width=1.3pt]($1/2*(B)+1/2*(D)$)--(t);
\draw[line width=1.3pt]($1/2*(D)+1/2*(E)$)--(y);
\path ($1/2*(C)+1/2*(F)$) node[]{${\color{white}\blacklozenge}$};
\path ($1/2*(C)+1/2*(F)$) node[]{${\lozenge}$};
\path ($1/2*(B)+1/2*(D)$) node[]{${\color{white}\blacklozenge}$};
\path ($1/2*(B)+1/2*(D)$) node[]{${\lozenge}$};
\path ($1/2*(E)+1/2*(D)$) node[]{${\color{white}\blacklozenge}$};
\path ($1/2*(E)+1/2*(D)$) node[]{${\lozenge}$};
\draw[gray!60, line width=1.5pt]($1/2*(E)+1/2*(F)$)--(z);
\draw[white, line width=0.3pt]($1/2*(E)+1/2*(F)$)--(z);
\draw[line width=1.3pt]($1/2*(E)+1/2*(F)$)--(F);
\draw[white,line width=0.3pt]($1/2*(E)+1/2*(F)$)--(F);
\path ($1/2*(E)+1/2*(F)$) node[]{${\color{white}\blacklozenge}$};
\path ($1/2*(E)+1/2*(F)$) node[]{${\lozenge}$};
\draw[line width=1.3pt]($1/2*(C)+1/2*(-3,4)$)--(C);
\draw[white,line width=0.3pt]($1/2*(C)+1/2*(-3,4)$)--(C);
\draw[gray!60, line width=1.5pt]($1/2*(C)+1/2*(-3,4)$)--(x);
\draw[white, line width=0.3pt]($1/2*(C)+1/2*(-3,4)$)--(x);
\path ($1/2*(C)+1/2*(-3,4)$) node[]{${\color{white}\blacklozenge}$};
\path ($1/2*(C)+1/2*(-3,4)$) node[]{${\lozenge}$};
\draw[line width=1.3pt]($1/2*(B)+1/2*(-3,4)$)--(B);
\draw[white,line width=0.3pt]($1/2*(B)+1/2*(-3,4)$)--(B);
\draw[gray!60, line width=1.5pt]($1/2*(B)+1/2*(-3,4)$)--(x);
\draw[white, line width=0.3pt]($1/2*(B)+1/2*(-3,4)$)--(x);
\path ($1/2*(B)+1/2*(-3,4)$) node[]{${\color{white}\blacklozenge}$};
\path ($1/2*(B)+1/2*(-3,4)$) node[]{${\lozenge}$};
\path ($1/2*(B)+1/2*(-3,4)$) node[anchor=east]{$\nwvupugl{k}$};
\path ($1/2*(A)+1/2*(E)+(0,-0.1)$) node[anchor=north]{$\nwvpugl{k}$};
\draw[white,line width=0.3pt]($1/2*(E)+1/2*(A)$)--(E);

\draw[line width=1.3pt](x)--(t)--(y)--(z)--(p)--(x);

\draw[white, line width=1.5pt]($1/2*(E)+1/2*(A)$)--(z);
\draw[gray!60, line width=1.5pt]($1/2*(E)+1/2*(A)$)--(z);
\draw[white, line width=0.3pt]($1/2*(E)+1/2*(A)$)--(z);

\draw[draw, dashed](x)--(A);
\draw[draw, dashed](x)--(B);
\draw[draw, dashed](x)--(C);

\draw[draw, dashed](t)--(A);
\draw[draw, dashed](t)--(B);
\draw[draw, dashed](t)--(D);

\draw[draw, dashed](y)--(A);
\draw[draw, dashed](y)--(E);
\draw[draw, dashed](y)--(D);

\draw[draw, dashed](z)--(A);
\draw[draw, dashed](z)--(E);
\draw[draw, dashed](z)--(F);

\draw[draw, dashed](p)--(A);
\draw[draw, dashed](p)--(C);
\draw[draw, dashed](p)--(F);

\path ($1/2*(x)+1/2*(t)$) node[]{${\color{white}\blacklozenge}$};
\path ($1/2*(x)+1/2*(t)$) node[]{${\lozenge}$};
\path ($1/2*(y)+1/2*(t)$) node[]{${\color{white}\blacklozenge}$};
\path ($1/2*(y)+1/2*(t)$) node[]{${\lozenge}$};
\path ($1/2*(y)+1/2*(z)$) node[]{${\color{white}\blacklozenge}$};
\path ($1/2*(y)+1/2*(z)$) node[]{${\lozenge}$};
\path ($1/2*(p)+1/2*(z)$) node[]{${\color{white}\blacklozenge}$};
\path ($1/2*(p)+1/2*(z)$) node[]{${\lozenge}$};
\path ($1/2*(x)+1/2*(p)$) node[]{${\color{white}\blacklozenge}$};
\path ($1/2*(x)+1/2*(p)$) node[]{${\lozenge}$};

\path (A) node[]{${\blacklozenge}$};
\path (B) node[]{${\blacklozenge}$};
\path (C) node[]{${\blacklozenge}$};
\path (D) node[]{${\blacklozenge}$};
\path (E) node[]{${\blacklozenge}$};
\path (F) node[]{${\blacklozenge}$};

\path (x) node[]{${\color{gray}\blacklozenge}$};
\path (y) node[]{${\color{gray}\blacklozenge}$};
\path (z) node[]{${\color{gray}\blacklozenge}$};
\path (t) node[]{${\color{gray}\blacklozenge}$};
\path (p) node[]{${\color{gray}\blacklozenge}$};

\end{scope}

\begin{scope}[xshift=10cm]
\path (0,0) node[name=A, shape=coordinate]{};
\path (0,-4) node[name=B', shape=coordinate]{};
\path (0,4) node[name=B, shape=coordinate]{};
\path (3,0) node[name=C', shape=coordinate]{};
\path (-3,0) node[name=C, shape=coordinate]{};
\path ($1/2*(B)+1/2*(C)$) node[name=x, shape=coordinate]{};
\path (2.5,0) node[name=y, shape=coordinate]{};
\path (0,-2.5) node[name=z, shape=coordinate]{};
\path ($1/2*(B)+1/2*(C')$) node[name=t, shape=coordinate]{};
\path ($(t)+(y)$) node[name=D, shape=coordinate]{};
\path ($(z)+(y)$) node[name=E, shape=coordinate]{};
\path ($1/2*(B')+1/2*(C)$) node[name=p, shape=coordinate]{};
\path ($(z)+(p)$) node[name=F, shape=coordinate]{};


\draw[draw, dashed](x)--(A);
\draw[draw, dashed](x)--(B);
\draw[draw, dashed](x)--(C);

\draw[draw, dashed](t)--(A);
\draw[draw, dashed](t)--(B);
\draw[draw, dashed](t)--(D);

\draw[draw, dashed](y)--(A);
\draw[draw, dashed](y)--(E);
\draw[draw, dashed](y)--(D);

\draw[draw, dashed](z)--(A);
\draw[draw, dashed](z)--(E);
\draw[draw, dashed](z)--(F);

\draw[draw, dashed](p)--(A);
\draw[draw, dashed](p)--(C);
\draw[draw, dashed](p)--(F);

\path ($1/2*(x)+1/2*(t)$) node[name=n1, shape=coordinate]{};
\path ($1/2*(y)+1/2*(t)$) node[name=n2, shape=coordinate]{};
\path ($1/2*(y)+1/2*(z)$) node[name=n3, shape=coordinate]{};
\path ($1/2*(p)+1/2*(z)$) node[name=n4, shape=coordinate]{};
\path ($1/2*(x)+1/2*(p)$) node[name=n5, shape=coordinate]{};

\path ($1/2*(x)+1/2*(t)$) node[]{${\color{white}\blacklozenge}$};
\path ($1/2*(x)+1/2*(t)$) node[]{${\lozenge}$};
\path ($1/2*(y)+1/2*(t)$) node[]{${\color{white}\blacklozenge}$};
\path ($1/2*(y)+1/2*(t)$) node[]{${\lozenge}$};
\path ($1/2*(y)+1/2*(z)$) node[]{${\color{white}\blacklozenge}$};
\path ($1/2*(y)+1/2*(z)$) node[]{${\lozenge}$};
\path ($1/2*(p)+1/2*(z)$) node[]{${\color{white}\blacklozenge}$};
\path ($1/2*(p)+1/2*(z)$) node[]{${\lozenge}$};
\path ($1/2*(x)+1/2*(p)$) node[]{${\color{white}\blacklozenge}$};
\path ($1/2*(x)+1/2*(p)$) node[]{${\lozenge}$};

\path (A) node[]{${\blacklozenge}$};
\path (B) node[]{${\blacklozenge}$};
\path (C) node[]{${\blacklozenge}$};
\path (D) node[]{${\blacklozenge}$};
\path (E) node[]{${\blacklozenge}$};
\path (F) node[]{${\blacklozenge}$};

\path (x) node[]{${\color{gray}\blacklozenge}$};
\path (y) node[]{${\color{gray}\blacklozenge}$};
\path (z) node[]{${\color{gray}\blacklozenge}$};
\path (t) node[]{${\color{gray}\blacklozenge}$};
\path (p) node[]{${\color{gray}\blacklozenge}$};

\path[draw] ($1/2*(A)+1/2*(x)$)--($1/2*(A)+1/2*(t)$) --($1/2*(A)+1/2*(y)$) --($1/2*(A)+1/2*(z)$) --($1/2*(A)+1/2*(p)$)--($1/2*(A)+1/2*(x)$); 

\path[draw] ($1/2*(A)+1/2*(x)$)--($1/2*(B)+1/2*(x)$)--($1/2*(B)+1/2*(t)$)--
($1/2*(A)+1/2*(t)$)--($1/2*(D)+1/2*(t)$)--($1/2*(D)+1/2*(y)$)--($1/2*(A)+1/2*(y)$)--
($1/2*(E)+1/2*(y)$)--($1/2*(E)+1/2*(z)$)--($1/2*(A)+1/2*(z)$)--
($1/2*(F)+1/2*(z)$)--($1/2*(F)+1/2*(p)$)--($1/2*(A)+1/2*(p)$)--
($1/2*(C)+1/2*(p)$)--($1/2*(C)+1/2*(x)$)--
($1/2*(A)+1/2*(x)$);

\path[draw] ($1/2*(B)+1/2*(t)$)--($1/2*(D)+1/2*(t)$);
\path[draw] ($1/2*(E)+1/2*(z)$)--($1/2*(F)+1/2*(z)$);
\path[draw] ($1/2*(F)+1/2*(p)$)--($1/2*(C)+1/2*(p)$);

\path[draw, fill=black] ($1/2*(A)+1/2*(x)$) circle[radius=0.05cm];
\path[draw, fill=black] ($1/2*(B)+1/2*(x)$) circle[radius=0.05cm];
\path[draw, fill=black] ($1/2*(C)+1/2*(x)$) circle[radius=0.05cm];

\path[draw, fill=black] ($1/2*(A)+1/2*(t)$) circle[radius=0.05cm];
\path[draw, fill=black] ($1/2*(B)+1/2*(t)$) circle[radius=0.05cm];
\path[draw, fill=black] ($1/2*(D)+1/2*(t)$) circle[radius=0.05cm];

\path[draw, fill=black] ($1/2*(A)+1/2*(y)$) circle[radius=0.05cm];
\path[draw, fill=black] ($1/2*(D)+1/2*(y)$) circle[radius=0.05cm];
\path[draw, fill=black] ($1/2*(E)+1/2*(y)$) circle[radius=0.05cm];

\path[draw, fill=black] ($1/2*(A)+1/2*(z)$) circle[radius=0.05cm];
\path[draw, fill=black] ($1/2*(F)+1/2*(z)$) circle[radius=0.05cm];
\path[draw, fill=black] ($1/2*(E)+1/2*(z)$) circle[radius=0.05cm];

\path[draw, fill=black] ($1/2*(A)+1/2*(p)$) circle[radius=0.05cm];
\path[draw, fill=black] ($1/2*(F)+1/2*(p)$) circle[radius=0.05cm];
\path[draw, fill=black] ($1/2*(C)+1/2*(p)$) circle[radius=0.05cm];
\end{scope}

\end{tikzpicture}\end{center}
\caption{Left: isoradial graph $\Gamma$ (black), its dual graph $\Gamma^{*}$ (gray) and the corresponding rhombic lattice (dashed). Center: isoradial $2n$-black-piecewise Temperleyan graph with $n=2$; the set of rhombic centers (white), the bipartite graph (vertices: white, gray and black; edges: solid lines), the elements of the sets of white corners 
 $\{\nwvupugl{k}\}^{n+1}_{k=1}$ and  $\{\nwvpugl{k}\}^{n-1}_{k=1}$. Right: the set of midedges of the rhombic lattice, the set $\mathcal{V}$ (circles).}\label{isorad}
\end{figure}
In this section we will discuss the result of Theorem~\ref{convF} for the dimer model on isoradial graphs. 
The notion of a rhombic lattice (or isoradial graph) was introduced by Duffin~\cite{Duffin} as a large family of graphs, where discretizations of Laplace and Cauchy-Riemann operators can be defined  similarly to the case of the square lattice. The class of isoradial graphs forms a large class of graphs where classical complex analysis results have discrete analogs, see~\cite{C+S}. A lot of planar graphs admit isoradial embeddings~\cite{ks04}.
Discrete complex analysis allows to obtain results for two-dimensional lattice models on isoradial graphs, notably the Ising~\cite{mer, CS} and dimer~\cite{KLD, BdT, Dubiso, Zli} models.

 Let $\Gamma$ be an isoradial graph, i.e. a planar graph 
in which each face is inscribed into a circle of a common radius~$\delta$. Also one can thing about $\delta$ as a mesh size of the isoradial graph. Suppose that all circle centers are inside the corresponding faces, then the dual graph $\Gamma^*$ is also isoradial with the same radius. The rhombic lattice is the graph on the union $\Lambda$ of the two vertex sets $\Gamma$ and $\Gamma^{*}$ (see Fig.~\ref{isorad}). We will use the following assumption (see~\cite{C+S})
\smallskip
\begin{enumerate}
\item[($\spadesuit$)] the rhombi angles are uniformly bounded from $0$ and $\pi$.
\end{enumerate}
\smallskip
The dimer model on isoradial graphs was introduced by Kenyon~\cite{KLD}.  
A dimer configuration in this setup is a perfect matching of the bipartite graph $\om^\delta$ defined as follows. The vertex set of $\om^\delta$ consists of a union of $\Lambda$ (two types of black vertices) and rhombi centers (white vertices), and there is an edge between black and white vertices if the black vertex and corresponding rhombi are adjacent (see Fig.~\ref{isorad}). Note that $\om^\delta$ is an isoradial graph, where each face is inscribed into a circle of radius~$\frac\delta2$, for more details see~\cite{Dubiso}.

 We will call a white vertex on the boundary of $\om^\delta$ a \emph{corner} if it is adjacent to two boundary black vertices of different types. We can define as before the notion of \emph{convex} and \emph{concave} white corners, see Fig.~\ref{isorad}. Note that Lemma~\ref{combi-lemma} holds also in the isoradial case. An isoradial graph $\om^\delta$ is called a {\it $2n$-black-piecewise Temperleyan graph} if it 
has $n+1$ convex white corners and $n-1$ concave white corners, see Fig.~\ref{isorad}.
Now we can formulate the similar result for isoradial graphs analogous to Theorem~\ref{convF}. 
 
\begin{Th} Let $\om^\delta$ be a sequence of isoradial $2k$-black-piecewise Temperleyan  graphs approximating a continuous domain $\om$. Assume that each $\om^\delta$ admits a perfect matching. Let the sets of white boundary vertex $\{\wvupugl{k}\}^{n+1}_{k=1}$ and $\{\wvpugl{k}\}^{n-1}_{k=1}$ approximate the sets of boundary points $\{\nwvupugl{k}\}^{n+1}_{k=1}$ and $\{\nwvpugl{k}\}^{n-1}_{k=1}$ correspondingly, and let $\www^{\delta}$ approximate a point boundary point $\www$ which lies on a straight segment of the boundary of $\om$. Then $\F_{\mathrm{iso}}^\delta$ converges uniformly on compact subsets of  $\om$ to a continuous holomorphic function $f_\om$, where $f_\om$ is defined as in Proposition~$\ref{1-4}$.\end{Th}

\begin{proof}
The proof mimics the proof of Theorem~\ref{convF} using the toolbox described in~\cite[Definition 2.1, Proposition 2.7, Definition 2.12, Proposition 2.14, Theorem 2.21]{C+S}.
\end{proof}


\begin{thebibliography}{1}
\bibitem{BLR}
N. Berestycki, B. Laslier, G.Ray, {\it Universality of fluctuations in the dimer model}, arXiv:1603.09740.

\bibitem{proba}
P. Billingsley, {\it Probability and measure}, Wiley, New York 1979.

\bibitem{BG}
A. Bufetov, V. Gorin, {\it Fourier transform on high-dimensional unitary groups with applications to random tilings},  	arXiv:1712.09925.

\bibitem{CHI}
D. Chelkak, C. Hongler, K. Izyurov, {\it Conformal invariance of spin correlations in the planar Ising model}, Ann. of Math.~(2), 181 (2015), 1087–1138.

\bibitem{C+S}
D. Chelkak and S. Smirnov, {\it Discrete complex analysis on isoradial graphs}, Adv. Math. 228 (2011), 1590--1630.

\bibitem{CS}
D. Chelkak and S. Smirnov, {\it Universality in the 2D Ising model and conformal invariance of fermionic observables}, Invent. Math., 189 (2012), no. 3, 515--580.

\bibitem{CKP}
H. Cohn, R. Kenyon, J. Propp, {\it A variational principle for domino tilings}, J. Amer. Math. Soc. 14 (2001), no. 2, 297–346.

\bibitem{BdT}
 B.  de  Tilière, {\it Scaling  limit  of  isoradial  dimer  models  and  the  case  of  triangular quadri-tilings}, Ann. Inst. H. Poincaré, Prob. et Stat. 43 (2007), no. 6, 729–750



\bibitem{Dubiso} J. Dubédat, {\it Dimers and families of Cauchy-Riemann operators I} , J. Amer. Math. Soc. 28 (2015), 1063--1167.


\bibitem{Ddim}
J. Dubédat, {\it Double dimers, conformal loop ensembles and isomonodromic deformations}, arXiv:1403.6076.

\bibitem{Duffin}
R.  Duffin, {\it Potential  theory  on  a  rhombic  lattice},  J.  Combinatorial  Theory 5 (1968), 258--272.

\bibitem{F}
M. E. Fisher, {\it Statistical mechanics of dimers on a plane lattice}, Phys. Rev. 124
(1961), no. 6, 1664–1672.

\bibitem{Ff}
M. E. Fisher, {\it On the dimer solution of planar Ising models}, Journal of Mathematical Physics 7 (1966), no. 10, 1776–1781.

\bibitem{FT}
M. E. Fisher and H. N. V. Temperley, {\it Dimer problem in statistical mechanics-an exact result}, Philos. Mag., 
6(68):1061–1063, 1961.

\bibitem{Kast}
P. W. Kasteleyn, {\it The statistics of dimers on a lattice. I. The number of dimer arrangements on a
quadratic lattice}, Physica~27 (1961), 1209–1225.

\bibitem{Kdom}
R. Kenyon, {\it Conformal invariance of domino tiling} , Ann. Probab. 28 (2000), 759--795.

\bibitem{Kloop}
R. Kenyon, {\it Conformal invariance of loops in the double-dimer model}, Comm. Math. Phys. 326 (2014), no. 2, 477–497.

\bibitem{KGff}
R. Kenyon, {\it Dominos and the Gaussian free field}, Ann. Probab. 29 (2001), 1128-1137.


\bibitem{Klocstat}
R. Kenyon, {\it Local statistics of lattice dimers}, Ann. Inst. H. Poincaré, Prob. et Stat. 33 (1997), 591--618.

\bibitem{KLD}
R. Kenyon, {\it The Laplacian and Dirac operators on critical planar graphs}, Invent. Math. 150 (2002), no 2, 409--439.

\bibitem{K-O}
R. Kenyon and A. Okounkov, {\it Limit shapes and the complex Burgers equation}, Acta
Math. 199 (2007), no. 2, 263–302.


\bibitem{K-O-Sh}
R. Kenyon, A. Okounkov, S. Sheffield {\it Dimers and Amoebae}, Ann. Math. 163
(2006), 1019–1056

\bibitem{ks04}
R. Kenyon, J.-M. Schlenker {\it Rhombic embeddings of planar quad-graphs}, Trans. AMS,
357:3443–3458, 2004.

\bibitem{Zli}
Z. Li {\it Conformal invariance of isoradial dimers}, arXiv preprint arXiv:1309.0151.

\bibitem{mer}
Ch. Mercat {\it Discrete Riemann Surfaces and the Ising Model}, Comm.  Math.  Phys.,
218(1):177–216, 2001.


\bibitem{Percus}
J. K. Percus, {\it One more technique for the dimer problem}, J. Mathematical Phys., 10:1881–1888, 1969.

\bibitem{P}
L. Petrov, {\it Asymptotics of Uniformly Random Lozenge Tilings of Polygons. Gaussian Free Field}, arXiv:1206.5123.


\bibitem{ICM2006}
O. Schramm, {\it Conformally invariant scaling limits, an overview and a collection of
problems}, in Proceedings of the ICM 2006, Madrid. 2006.

\bibitem{ScSh}
O. Schramm, S. Sheffield {\it Contour lines of the two-dimensional discrete Gaussian free field}, Acta Math., 202 (2009), 21–137.

\bibitem{CLE1}
S. Sheffield, {\it Exploration trees and conformal loop ensembles}, Duke Math. J. 147 (2009), no. 1, 79–129. 


\bibitem{Scott}
S. Sheffield, {\it Gaussian free fields for mathematicians}, Probab. Theory Related Fields 139 (2007), no. 3-4, 521–541. 


\bibitem{CLE2}
S. Sheffield, W. Werner {\it
Conformal loop ensembles: the Markovian characterization and the loop-soup construction}, Ann. of Math. (2) 176 (2012), no. 3, 1827–1917.

\bibitem{Stas}
S. Smirnov, {\it Towards conformal invariance of 2D lattice models}, proceedings of the international
congress of mathematicians (ICM), Madrid, Spain, August 22–30, 2006. Vol. II:
Invited lectures, 1421--1451. Zurich: European Mathematical Society (EMS), 2006.

\bibitem{Stas07}
S. Smirnov, {\it Conformal invariance in random cluster models. I. Holomorphic fermions in the Ising model}, Ann. of Math.~(2), (172) 1435–1467, 2010.

\bibitem{Tem}
H. Temperley, {\it Combinatorics: Proceedings of the British Combinatorial Conference 1973}, London Math. Soc. Lecture Notes Series $\#13$, (1974), 202–204.

\bibitem{Ter}
William P. Thurston, {\it Conway's tiling groups}, The American Mathematical Monthly, Vol. 97, No. 8, Special Geometry Issue (Oct.,1990), pp. 757-773.

\bibitem{WW}
M. Wang, H. Wu {\it Level Lines of Gaussian Free Field I: Zero-Boundary GFF}, SPA 127 (2017) 1045-1124.

\end{thebibliography}
\end{document}